\documentclass[onecolumn,journal]{IEEEtran}

\usepackage[utf8]{inputenc}
\usepackage{amsmath, amssymb, amsthm, bbm, mathtools, nccmath}
\usepackage{url, hyperref, soul}
\usepackage{caption, subcaption}

\usepackage{algorithm, algpseudocode}
\algdef{SE}[DOWHILE]{Do}{DoWhile}{\algorithmicdo}[1]{\algorithmicwhile\ #1}
 
\usepackage{tikz,pgfplots}
\usetikzlibrary{positioning, decorations.pathreplacing}
\pgfplotsset{compat=1.18}
\graphicspath{{figs/}}
\theoremstyle{plain}
\newtheorem{theorem}{Theorem}
\newtheorem{lemma}[theorem]{Lemma}
\newtheorem{corollary}[theorem]{Corollary}
\newtheorem{proposition}[theorem]{Proposition}

\theoremstyle{definition}

\theoremstyle{remark}
\newtheorem{remark}[theorem]{Remark}
\renewcommand{\mathbf}{\boldsymbol}
\renewcommand{\tilde}{\widetilde}
\renewcommand{\leq}{\leqslant}
\renewcommand{\geq}{\geqslant}
\newcommand{\set}[1]{\mathcal{#1}}
\newcommand{\fnc}[1]{\mathrm{#1}}
\newcommand{\rv}[1]{\mathsf{#1}}
\newcommand{\rvs}[1]{\mathbf{\mathsf{#1}}}
\newcommand{\defeq}{\coloneqq}
\newcommand{\defas}{\eqqcolon}
\newcommand{\abs}[1]{\left\lvert#1\right\rvert}

\newcommand{\norm}[1]{\left\lVert#1\right\rVert}
\newcommand{\size}[1]{\left\lvert#1\right\rvert}
\newcommand{\ceil}[1]{\left\lceil#1\right\rceil}
\newcommand{\floor}[1]{\left\lfloor#1\right\rfloor}

\DeclareMathOperator{\ID}{id} 
\newcommand{\tensor}{\otimes}

\DeclareMathOperator{\tri}{\raisebox{.3pt}{\mbox{\tiny$\blacktriangle$}}}

\DeclareMathOperator{\argmax}{argmax}

\DeclareMathOperator{\D}{d} 

\DeclareMathOperator{\supp}{supp}

\DeclarePairedDelimiterX{\infdiv}[2]{(}{)}{#1\delimsize\Vert#2}
\DeclarePairedDelimiterX{\inner}[2]{\langle}{\rangle}{#1,#2}
\newcommand{\tos}[2]{\stackrel{\mathclap{\small\mbox{#1}}}{#2}} 
\newcommand{\dscript}[2]{\genfrac{}{}{0pt}{2}{#1}{#2}} 
\ExplSyntaxOn
\NewDocumentCommand{\multiadjustlimits}{m}
 {
  \group_begin:
  \multiadjustlimits_measure:n { #1 }
  \multiadjustlimits_print:n { #1 }
  \group_end:
 }
\tl_new:N  \l__multiadjustlimits_operator_tl
\tl_new:N  \l__multiadjustlimits_limit_tl
\cs_new_protected:Nn \multiadjustlimits_measure:n
 {
  \clist_map_function:nN { #1 } \__multiadjustlimits_measure:n
 }
\cs_new_protected:Nn \__multiadjustlimits_measure:n
 {
  \__multiadjustlimits_measure:NNn #1
 }
\cs_new_protected:Nn \__multiadjustlimits_measure:NNn
 {
  \tl_put_right:Nn \l__multiadjustlimits_operator_tl { #1 }
  \tl_put_right:Nn \l__multiadjustlimits_limit_tl { #3 }
 }
\cs_new_protected:Nn \multiadjustlimits_print:n
 {
  \clist_map_function:nN { #1 } \__multiadjustlimits_print:n
 }
\cs_new_protected:Nn \__multiadjustlimits_print:n
 {
  \__multiadjustlimits_print:NNn #1
 }
\cs_new_protected:Nn \__multiadjustlimits_print:NNn
 {
  \mathop { \vphantom{\l__multiadjustlimits_operator_tl} \mathopen{} #1 }
  \limits
  \sb{ \vphantom{\cramped{\l__multiadjustlimits_limit_tl}} #3 }
 }
\ExplSyntaxOff
\makeatletter
\newcommand\ie{\textit{i.e.}}
\newcommand\eg{\textit{e.g.}}
\newcommand\cf{\textit{cf.}}
\newcommand\wrt{w.r.t.~}

\newcommand\vs{\textit{v.s.}}
\makeatother

\begin{document}

\title{Channel Simulation: Finite Blocklengths and Broadcast Channels}

\author{
    Michael X. Cao,~\IEEEmembership{Member, IEEE},
    \thanks{MC is with the Department of Electrical and Computer Engineering and the Centre for Quantum Technologies, National University of Singapore, Singapore, Email: \href{mailto:m.x.cao@ieee.org}{\url{m.x.cao@ieee.org}}.}
    Navneeth Ramakrishnan,
    \thanks{NR is with the Centre for Quantum Technologies, National University of Singapore and the Department of Computing, Imperial College London, United Kingdom}
    Mario Berta
    \thanks{MB is with the Institute for Quantum Information, RWTH Aachen University, Germany and the Department of Computing, Imperial College London, United Kingdom.}\\
    Marco Tomamichel,~\IEEEmembership{Senior Member, IEEE} 
    \thanks{MT is with the Department of Electrical and Computer Engineering and the Centre for Quantum Technologies, National University of Singapore, Singapore.}
    \thanks{Part of Section~\ref{sec:One-Shot(p2p)} was presented at the IEEE International Symposium on Information Theory 2022 as~\cite{cao2022one-shot}.}
    \thanks{Section~\ref{sec:broadcast_channels} was presented at the IEEE International Symposium on Information Theory 2023 as~\cite{cao2023broadcast}.}
    \thanks{This research is supported by the National Research Foundation, Prime Minister's Office, Singapore and the Ministry of Education, Singapore under the Research Centres of Excellence programme. MC and MT are supported by NUS startup grants (A-0009028-02-00). MB acknowledges funding from the European Research Council (ERC Grant Agreement No.~948139).}
}

\maketitle


\begin{abstract}
We study channel simulation under common randomness assistance in the finite-blocklength regime and identify the smooth channel max-information as a linear program one-shot converse on the minimal simulation cost for fixed error tolerance.
We show that this one-shot converse can be achieved exactly using no-signaling-assisted codes, and approximately achieved using common randomness-assisted codes. Our one-shot converse thus takes on an analogous role to the celebrated meta-converse in the complementary problem of channel coding, and we find tight relations between these two bounds.
We asymptotically expand our bounds on the simulation cost for discrete memoryless channels, leading to the second-order as well as the moderate deviation rate expansion, which can be expressed in terms of the channel capacity and channel dispersion known from noisy channel coding.
Our bounds imply the well-known fact that the optimal asymptotic rate of one channel to simulate another under common randomness assistance is given by the ratio of their respective capacities.
Additionally, our higher-order asymptotic expansion shows that this reversibility falls apart in the second order.
Our techniques extend to discrete memoryless broadcast channels.
In stark contrast to the elusive broadcast channel capacity problem, we show that the reverse problem of broadcast channel simulation under common randomness assistance allows for an efficiently computable single-letter characterization of the asymptotic rate region in terms of the broadcast channel's multipartite mutual information.
\end{abstract}

\tableofcontents


\section{Introduction}
\label{sec:introduction}
\IEEEPARstart{C}{hannel simulation} is the art of simulating the input-output conditional distribution of a target channel using a noiseless resource channel~\cite{bennett2002entanglement, bennett2014quantum}.
Given a target channel $W_{\rv{Y}|\rv{X}}$, we want to approximately simulate a system that, upon accepting an input $x$ at the receiver, outputs $\rv{Y}$ distributed according to the law $W_{\rv{Y}|\rv{X}}(\cdot|x)$.
More precisely, we want to simulate the input-output correlations of the channel $W_{\rv{Y}|\rv{X}}$ as a black box, using (unconstrained) common randomness and limited noiseless communication between the sender and the receiver.
The goal is to find the minimal amount of communication needed to perform this task, both in the one-shot setting and asymptotically, where the goal is to simulate the input-output correlations for many uses of a discrete memoryless channel (DMC) for an arbitrary input sequence.
Channel simulation is the dual problem (in a sense that we make more precise in the next paragraph) of Shannon's noisy channel coding task~\cite{shannon1948mathematical}, and its asymptotic characterization for DMCs is known as the reverse Shannon theorem~\cite[Theorem~2]{bennett2002entanglement}.

Both noisy channel coding and channel simulation can be seen as special cases of the general channel interconversion problem (see, \eg~\cite{sudan2019communication} and references therein).
The latter is to simulate a target channel $W$ using a resource channel $V$, where both channels can be noisy.
It is then a fundamental question to characterize the asymptotics of the minimal rate, $m/n$, where $m$ is the number of uses of the resource channel required for the faithful simulation of $n$ instances of the target channel under a worst-case error criterion, \ie, to determine the capacity $C(V\mapsto W)$ of one channel $V$ to simulate another channel $W$. 
On the one hand, for point-to-point DMCs, the channel simulation corresponds to the special case where $V=\ID$ is the identity channel.
The reverse Shannon theorem~\cite{bennett2002entanglement, bennett2014quantum} establishes that $C(\ID\mapsto W)=C(W)$, where $C(W)$ is the maximal mutual information of the channel $W$, \ie, the channel capacity.
On the other hand, the special case where $W=\ID$ corresponds to the channel coding problem under the maximum error constraint, and $C(V\mapsto\ID)^{-1} = C(V)$ is nothing but the channel capacity of $V$.
In this framework, channel simulation can be understood as the reverse task of channel coding.

Channel simulation is closely related, but differs in important aspects from an array of problems in distributed correlation synthesis.
The study of remote generation of correlations dates back to the pioneering work of Wyner~\cite{wyner1975common}, who studied the distributed generation of a pair of correlated random variables using a minimal amount of shared randomness between the two parties (and without communication).
The minimum amount of shared randomness needed for this task is asymptotically characterized by Wyner's common information.
A similar task, known as (strong) coordination~\cite{cuff2010coordination,cuff2013distributed} or channel synthesis~\cite{winter2002compression,harsha2010communication} concerns the realization of a joint distribution $p_\rv{XY}$ between Alice and Bob, where Alice receives $\rv{X}$ and Bob generates $\rv{Y}$ with the help of one-way communications from Alice to Bob (with or without additional shared resources). 
This can be achieved by simulating the channel $p_{\rv{Y}|\rv{X}}\defeq p_\rv{XY}/p_\rv{X}$, but not vice versa.
The latter is due to the fact that coordination tasks are source-specific, \ie, the protocol can make use of the input distribution $p_\rv{X}$, and only guarantees the approximation to $p_\rv{X}\cdot p_{\rv{Y}|\rv{X}}$.
In contrast, channel simulation requires the approximation to be valid for any input distribution.
The difference between coordination and channel simulation in the asymptotic case is even more pronounced, since the former only concerns the i.i.d.\ input sources.
In this sense, the channel simulation task can be viewed as a ``universal'' coordination of $\{p_\rv{X}\cdot p_{\rv{Y}|\rv{X}}\}_{p_\rv{X}\in\set{P}(\set{X})}$, where the word ``universal'' refers to the requirement that the protocol works for all $p_\rv{X}$ simultaneously without using any knowledge of $p_\rv{X}$.
This can also be viewed as a ``worst-case'' channel synthesis, as noted in~\cite[Footnote 17]{sudan2019communication}.
For the above two tasks, we refer to~\cite{yu2022common} for a survey.
The information task studied in this paper corresponds to the original reverse Shannon theorem~\cite[Theorem~2]{bennett2002entanglement} and the classical non-feedback simulation \wrt general sources with excess shared randomness in~\cite[Fig.~2]{bennett2014quantum}.
The above three tasks are illustrated in Fig.~\ref{fig:related:tasks}.
\begin{figure}
\begin{subfigure}[t]{0.3\textwidth}
    \centering \includegraphics{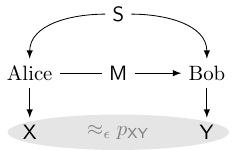}
    \caption{Distributed generation of two dependent random variables.}
\end{subfigure}
\hfill
\begin{subfigure}[t]{0.38\textwidth}
    \centering \includegraphics{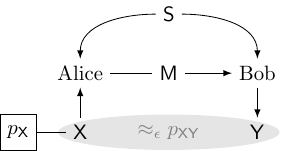}
    \caption{Strong coordination / channel synthesis / channel simulation on known sources}
\end{subfigure}
\hfill
\begin{subfigure}[t]{0.3\textwidth}
    \centering \includegraphics{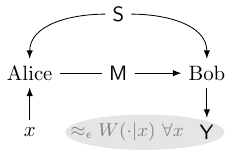}
    \caption{Channel simulation in this paper}
\end{subfigure}
\caption{Various tasks related to remote generation of correlations.}
\label{fig:related:tasks}
\end{figure}

It is noteworthy that a number of previous works are dedicated to the study of the trade-offs between the rate of common randomness and the rate of communication (\eg, see~\cite{cuff2013distributed, bennett2014quantum,yu2019exact}).
Despite the interest in this trade-off, we assume unconstrained common randomness in this paper, and leave this perspective to future studies.
One also has to distinguish between the feedback channel simulation, where the sender additionally gets a copy of the output of the channel, and the non-feedback setting as in our case (see~\cite[Fig.~2]{bennett2014quantum}).

The task of coordination is also generalized to broadcast channels~\cite{cuff2013distributed, haddadpour2016simulation} and multiple access channels~\cite{kurri2022multiple}.
In~\cite{cuff2013distributed}, the author considered the task of synthesizing a broadcast channel using a noiseless broadcast channel where each receiver gets the same message from the sender.
In~\cite{haddadpour2016simulation}, the author studies the same task, but with a noisy broadcast channel as resources.
To the best of our knowledge, the characterization of broadcast channel simulation with unlimited shared randomness leading to matching inner and outer bounds is novel to our work.

In this paper, we focus on the one-shot minimal communication cost for channel simulations given unlimited shared randomness, and its asymptotic expansions in the small and moderate-deviation regimes.
For the coordination problem, a finite-blocklength study is available in~\cite{cervia2021epsilon}.
In quantum information theory, one-shot bounds on measurement channel simulation are studied in~\cite{berta2014identifying}, and one-shot bounds on related tasks of quantum state splitting and state merging have been studied in~\cite{berta2011quantum} and~\cite{ramakrishnan2023moderate}.
However, restricting these bounds to classical channels leads to bounds that are not tight enough for deriving the finite-blocklength expansions that we are interested in here.
No-signaling-assisted quantum channel simulation is the topic of~\cite{fang2019quantum}, and a quantum version of the $\epsilon$-smooth channel max-information was defined in~\cite[Def.~5]{fang2019quantum}.
Our findings on the exact no-signaling achievability of the meta-converse for channel simulation are related to the results in~\cite[App.~A]{fang2019quantum}, but we give a more direct and conceptually different proof here for the classical case.
Moreover, the asymptotic expansion of the meta-converse given in \cite{fang2019quantum} is not second-order tight because it makes use of the so-called de Finetti reductions~\cite{christandl2009postselection}, also termed universal states \cite{hayashi2009universal}.

\smallskip

The remainder of the paper is organized as follows.
In Section~\ref{sec:overview} we first introduce our notation and collect the most important definitions.
We then formally introduce the task of channel simulation and state our main results.
In Section~\ref{sec:One-Shot(p2p)} we derive our one-shot converse and achievability bounds for the randomness-assisted point-to-point channel simulation.
In Section~\ref{sec:higher-order(p2p)} we provide various asymptotic expansions of these bounds in the small and moderate-deviation regimes.
In Section~\ref{sec:NS-channel_simulation} we consider codes with no-signaling assistance, and show that our one-shot converse is tight for such codes.
In Section~\ref{sec:broadcast_channels} we extend our results to the simulation of broadcast channels.
We conclude with a discussion of potential future work, open questions, and their difficulties in Section~\ref{sec:conclusion}.

\section{Overview of the main results}\label{sec:overview}

We summarize our notational conventions and important definitions here for easy reference.

\subsection{Notational conventions and definitions}

We use the following conventions and notations throughout the paper.
\begin{itemize}
    \item Logarithms (denoted by $\log$) in this paper are to base 2, unless otherwise stated.
    The natural logarithm is denoted by $\ln$.
    \item Sets are denoted by calligraphic fonts, \eg, $\set{X}$ reads ``set $\set{X}$''.
    \item Random variables are denoted in sans serif fonts, \eg, $\rv{X}$ reads ``random variable $\rv{X}$''.
    \item Vectors are denoted in boldface letters, \eg, $\mathbf{x}$ and $\rvs{X}$.
    In particular, we use the subscript and superscript to denote the starting and ending indexes of a vector.
    Namely, $\mathbf{x}_1^n\equiv (x_1, \ldots, x_n)$, and $\rvs{X}_1^n\equiv (\rv{X}_1,\ldots, \rv{X}_n)$.
    \item Given a discrete random variable $\rv{X}$, we denote its probability mass function (pmf) by $p_\rv{X}$. 
    \item $\Phi$ denotes the cumulative distribution function of the standard normal distribution.
    \item $\delta_{\cdot,\cdot}$ denotes the Kronecker delta, \ie, $\delta_{i,j}=1$ if $i=j$ and $\delta_{i,j}=0$ otherwise.
    \item For a positive integer $M$, we denote $[M]$ the set $\{1,\ldots, M\}$.
    \item Given a finite set $\set{X}$, $\set{P}(\set{X})$ denotes the set of all pmfs on $\set{X}$.
    Given an additional set $\set{Y}$, $\set{P}(\set{Y}|\set{X})$ denotes the set of all conditional pmfs on $\set{Y}$ conditioned on $\set{X}$.\end{itemize}

We also use the following definitions.
\begin{itemize}
    \item Given discrete random variables $\rv{X}$, $\rv{Y}$ and $\rv{Z}$, the Shannon entropy of $\rv{X}$ is given as $H(\rv{X}) \defeq - \sum_{x \in \set{X}} p_\rv{X}(x) \log p_\rv{X}(x)$, the the mutual information between $\rv{X}$ and $\rv{Y}$ is $I(\rv{X};\rv{Y}) \defeq H(\rv{X}) + H(\rv{Y}) - H(\rv{XY})$, and we define the common information between $\rv{X}$, $\rv{Y}$, and $\rv{Z}$ as 
    \begin{equation}
    I(\rv{X}; \rv{Y}; \rv{Z}) \defeq H(\rv{X}) + H(\rv{Y}) + H(\rv{Z}) - H(\rv{XYZ}).
    \end{equation}
    \item Given two pmfs $p$ and $q$ on the same finite alphabet, say $\set{X}$, we define
    \begin{itemize}
        \item the total variation distance (TVD) between $p$ and $q$: $\norm{p-q}_\fnc{tvd}\defeq\frac{1}{2}\sum_{x\in\set{X}}\abs{p(x)-q(x)}$;
        \item the max-divergence (or Rényi divergence of order $\infty$) between $p$ and $q$: 
        \begin{equation}
            D_{\max}\infdiv*{p}{q} \defeq \inf \big\{ \lambda \in \mathbb{R} \,\big|\, p(x) \leq 2^\lambda\cdot q(x) \ \forall\ x \in \set{X} \big\} ;
        \end{equation} 
        \item the $\epsilon$-information spectrum divergence between $p$ and $q$: 
        \begin{equation}\label{eq:def:Ds}
        D_{s+}^\epsilon\infdiv*{p}{q} \defeq \inf\left\{a\geq 0 \middle\vert \Pr\nolimits_{\rv{X}\sim p} \left[\log{\frac{p(\rv{X})}{q(\rv{X})}}>a\right]<\epsilon\right\} \!;
        \end{equation}
        \item the minimal type-II error given $\epsilon$-bound on type-I error for the hypothesis testing between $p$ and $q$:
        \begin{equation}\label{eq:def:beta}
        \beta^\star_\epsilon(p\|q)\defeq \inf\left\{ q(\set{A}): \set{A}\subset \set{X},\  p(\set{A}) \geq 1-\epsilon \right\}.
        \end{equation}
    \end{itemize}
    \item Given a conditional pmf $W_{\rv{Y}|\rv{X}}\in\set{P}(\set{Y}|\rv{X})$, also called a channel, we define 
    \begin{itemize}
        \item the channel mutual information: $C(W_{\rv{Y}|\rv{X}})\defeq \sup_{p_\rv{X}\in\set{P}(\set{X})} I(\rv{X}; \rv{Y})_{p_\rv{X}\cdot W_{\rv{Y}|\rv{X}}}$;
        \item the channel max-information: 
        \begin{equation}\label{eq:maxmutual}
        I_{\max}(W_{\rv{Y}|\rv{X}}) \defeq 
        \adjustlimits
        \sup_{p_\rv{X}\in\set{P}(\set{X})}
        \inf_{q_\rv{Y}\in\set{P}(\set{Y})}
        D_{\max}\infdiv*{p_\rv{X}\cdot {W}_{\rv{Y}|\rv{X}}}{p_\rv{X} \times q_\rv{Y}};
        \end{equation}
        \item the the $\epsilon$-smooth channel max-information \wrt some distance function $\Delta$ on $\set{P}(\set{Y}|\rv{X})$
        \begin{equation}
        I_{\max}^\epsilon(W_{\rv{Y}|\rv{X}}) \defeq
        \inf_{\tilde{W}_{\rv{Y}|\rv{X}}\in\set{P}(\set{Y}|\set{X}): \Delta(\tilde{W}_{\rv{Y}|\rv{X}}, W_{\rv{Y}|\rv{X}}) \leq\epsilon} I_{\max}(\tilde{W}_{\rv{Y}|\rv{X}});
        \end{equation}
        note that, throughout this paper, we consider a specific distance function $\Delta$ as in~\eqref{eq:choice:Delta};
        \item the $\epsilon$-channel dispersion (see, \eg,~\cite{polyanskiy2010channel, hayashi2009information}): 
        \begin{equation}\label{eq:def:V:eps}
        V_\epsilon(W_{\rv{Y}|\rv{X}}) \defeq \begin{cases}
        V_{\min}(W_{\rv{Y}|\rv{X}}) \defeq \min_{p_\rv{X}\in\Pi} V(p_\rv{X})
        & \text{if }\epsilon < \frac{1}{2} \\
        V_{\max}(W_{\rv{Y}|\rv{X}}) \defeq \max_{p_\rv{X}\in\Pi} V(p_\rv{X})
        & \text{if }\epsilon \geq \frac{1}{2}
        \end{cases}
        \end{equation}
        where
        \begin{align}
        \Pi & \defeq \left\{ p_{\rv{X}} \in \set{P}(\set{X}) \middle\vert \, I(\rv{X}; \rv{Y})_{p_{\rv{X}} \cdot W_{\rv{Y}|\rv{X}}} = C(W_{\rv{Y}|\rv{X}}) \right\} ,\\
        V(p_\rv{X}) &\defeq V\infdiv*{p_\rv{X}\cdot W_{\rv{Y}|\rv{X}}}{p_\rv{X}\times \sum_{x}p_\rv{X}(x)\cdot W_{\rv{Y}|\rv{X}}(\cdot|x)}\\
        V\infdiv{p}{q} &\defeq \fnc{Var}\left(\log{\frac{p(\rv{X})}{q(\rv{X})}}\right)= \sum_{x\in\set{X}} p(x)\cdot\left(\log{\frac{p(x)}{q(x)}}-D\infdiv{p}{q}\right)^2 \text{assuming $p \ll q$}.
        \end{align}
    \end{itemize}
    \item Given a discrete bipartite broadcast channel $W_{\rv{YZ}|\rv{X}}\in\set{P}(\set{Y}\times\set{Z}|\set{X})$, we define the bipartite channel mutual information $\tilde{C}(W_{\rv{YZ}|\rv{X}})$ (see~\cite{mcgill1954multivariate, watanabe1960information}) (also see the general form~\eqref{eq:def:K-partite:tilde:C}) as
    \begin{equation}\label{eq:def:tilde:C}
    \tilde{C}(W_{\rv{YZ}|\rv{X}}) \defeq \sup_{p_\rv{X}\in\set{P}(\set{X})} I(\rv{X};\rv{Y};\rv{Z})_{p_\rv{X}\cdot W_{\rv{YZ}|\rv{X}}}.
    \end{equation}
    \item A positive sequence $\{a_n\}_{n\in\mathbb{N}}$ is said to be moderate if $a_n\to 0$ but $\sqrt{n}\cdot a_n \to \infty$ as $n\to\infty$.
\end{itemize}

\begin{figure}
\begin{subfigure}[b]{0.33\textwidth}
    \centering \includegraphics{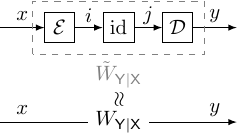}
    \caption{Non-assisted Channel Simulation}
    \label{fig:simulation:task:nonassisted}
\end{subfigure}
\begin{subfigure}[b]{0.33\textwidth}
    \centering \includegraphics{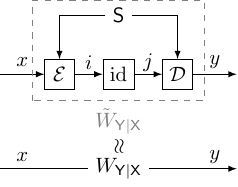}
    \caption{Random-assisted Channel Simulation}
    \label{fig:simulation:task:random}
\end{subfigure}
\begin{subfigure}[b]{0.33\textwidth}
    \centering \includegraphics{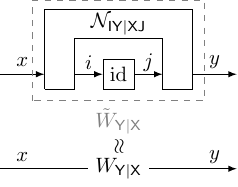}
    \caption{No-Signaling-Assisted Channel Simulation}
    \label{fig:simulation:task:NS}
\end{subfigure}
\caption{The task of (point-to-point) channel simulation.}
\label{fig:simulation:task}
\end{figure}


\subsection{Problem setup}
\label{sec:overview:setup}

Let us first formally introduce the channel simulation problem.
Let $W_{\rv{Y}|\rv{X}}\in\set{P}(\set{Y}|\set{X})$ be a channel.
Let $\Delta$ be a distance function on $\set{P}(\set{Y}|\set{X})$, and let $\epsilon\geq 0$.
We are interested in constructing a channel $\tilde{W}_{\rv{Y}|\rv{X}}\in\set{P}(\set{Y}|\set{X})$ that ``resembles'' $W_{\rv{Y}|\rv{X}}$, \ie,
\begin{equation}
\Delta(W_{\rv{Y}|\rv{X}},\tilde{W}_{\rv{Y}|\rv{X}}) \leq \epsilon
\end{equation}
using some $M$-alphabet-sized identity channel $\ID_M:[M]\times[M]\to\{0,1\}$ where $\ID_M(j|i)=\delta_{i,j}$.
In the randomness-assisted scenario, we allow encoders and decoders to be assisted by unlimited common randomness. In the no-signaling-assisted scenario, we allow more powerful joint encoders and decoders with the only restriction that they cannot communicate (except through the use of the noiseless channel).

\begin{itemize}
    \item In the \emph{random-assisted} scenario (also see Fig.~\ref{fig:simulation:task:random}), $\tilde{W}_{\rv{Y}|\rv{X}}$ is constructed as
    \begin{equation}
        \tilde{W}_{\rv{Y}|\rv{X}}(y|x) = \sum_{s\in\set{S}} p_\rv{S}(s) \cdot \sum_{i,j\in[M]} \mathcal{E}_{\rv{M}|\rv{XS}}(i|x,s) \cdot \ID_M(j|i) \cdot \mathcal{D}_{\rv{Y}|\rv{M}\rv{S}}(y|j,s)
        = \sum_{s\in\set{S}} p_\rv{S}(s) \cdot \sum_{i\in[M]} \mathcal{E}_{\rv{M}|\rv{XS}}(i|x,s) \cdot \mathcal{D}_{\rv{Y}|\rv{M}\rv{S}}(y|i,s)
    \end{equation}
    where $\mathcal{E}_{\rv{M}|\rv{XS}}\in\set{P}([M]|\set{X}\times\set{S})$ and $\mathcal{D}_{\rv{Y}|\rv{M}\rv{S}}\in\set{P}(\set{Y}|[M]\times\set{S})$ are some randomized encoding and decoding map on the sender and the receiver side, respectively; and $\rv{S}$ is some random variable shared between the sender and the receiver.
    \item In the \emph{no-signaling-assisted} scenario (also see Fig.~\ref{fig:simulation:task:NS}), $\tilde{W}_{\rv{Y}|\rv{X}}$ is constructed as
    \begin{equation}
        \tilde{W}_{\rv{Y}|\rv{X}}(y|x) = \sum_{i,j\in[M]} \mathcal{N}_{\rv{IY}|\rv{XJ}}(i,y|x,j) \cdot \delta_{i,j}
    \end{equation}
    where the joint encoding-decoding map $\mathcal{N}_{\rv{IY}|\rv{XJ}}\in\set{P}([M]\times\set{Y}|\set{X}\times[M])$ is no-signaling (see, \eg,~\cite{cubitt2011zero, matthews2012linear, fang2019quantum}), \ie
    \begin{align}
    \label{eq:ns:requirement:1}
    \sum_{y\in\set{Y}} N_{\rv{IY}|\rv{XJ}}(i,y|x,j) &= N_{\rv{I}|\rv{X}}(i|x) \qquad\forall j\in\set{J},\\
    \label{eq:ns:requirement:2}
    \sum_{i\in\set{I}} N_{\rv{IY}|\rv{XJ}}(i,y|x,j) &= N_{\rv{Y}|\rv{J}}(y|j) \qquad\forall x\in\set{X}.
    \end{align}
    The random variables $\rv{X}$ and $\rv{I}$ (or $\rv{J}$ and $\rv{Y}$) are the input and output of the encoder (or decoder), respectively.
    The requirements~\eqref{eq:ns:requirement:1} and~\eqref{eq:ns:requirement:2}, as the name suggests, prohibit the use of $\mathcal{N}_{\rv{IY}|\rv{XJ}}$ to send information between parties directly.
    No-signaling correlations in particular include correlations that can be achieved using quantum entanglement.
\end{itemize}
In both scenarios above, our goal is to characterize the minimal integer $M$ so that the above tasks are feasible.

Moreover, in this paper, we choose the distance function $\Delta$ as the maximal total variation distance (TVD) between output distributions, \ie,
\begin{equation} \label{eq:choice:Delta}
\Delta(W_{\rv{Y}|\rv{X}},\tilde{W}_{\rv{Y}|\rv{X}}) =  \norm{\tilde{W}_{\rv{Y}|\rv{X}}-W_{\rv{Y}|\rv{X}}}_\fnc{tvd} \defeq \sup_{x\in\set{X}}\norm{\tilde{W}_{\rv{Y}|\rv{X}}(\cdot|x)-W_{\rv{Y}|\rv{X}}(\cdot|x)}_\fnc{tvd}.
\end{equation}
In other words, we are working with a \emph{worst case error criterion} and require that
\begin{equation} \label{eq:channel-distance}
\norm{p_\rv{X}\cdot\tilde{W}_{\rv{Y}|\rv{X}}-p_\rv{X}\cdot W_{\rv{Y}|\rv{X}}}_\fnc{tvd} \leq\epsilon 
\end{equation}
\emph{for all} $p_\rv{X}\in\set{P}(\set{X})$, \ie, we ask for $\epsilon$-approximation of the joint input-output distribution of the channel for \emph{for every possible input distribution $p_\rv{X}\in\set{P}(\set{X})$}.

In Section~\ref{sec:higher-order(p2p)}, we further consider this problem in asymptotic setups.
Specifically, we consider the task of simulating $n$ uses of a DMC in the small and moderate deviation regimes. 
In Section~\ref{sec:broadcast_channels}, we extend the channel simulation problem to broadcast channels.
We defer the detailed description of these setups to Sections~\ref{sec:higher-order(p2p)} and~\ref{sec:broadcast_channels}, respectively.


\subsection{One-shot bounds}


\subsubsection{One-Shot meta-converse}\label{sec:metaconverse}

If the criterion~\eqref{eq:channel-distance} is satisfied in the aforementioned random-assisted scenario, we call the triple $(\mathcal{E}_{\rv{M}|\rv{XS}},\mathcal{D}_{\rv{Y}|\rv{MS}},\rv{S})$ a size-$M$ $\epsilon$-simulation code for $W_{\rv{Y}|\rv{X}}$.
An integer $M$ is said to be \emph{attainable} for a given $\epsilon\in[0,1]$ if there exists a size-$M$ $\epsilon$-simulation code.
As a first observation, if $M$ is attainable, so is any integer greater than $M$.
Our goal is thus to characterize $M_\epsilon^\star(W_{\rv{Y}\vert\rv{X}})$, the minimal attainable size of the $\epsilon$-simulation codes for a given DMC $W_{\rv{Y}\vert\rv{X}}$.
We find that the minimal simulation cost with fixed TVD tolerance $\epsilon\in(0,1)$, denoted by $M_\epsilon^\star(W_{\rv{Y}|\rv{X}})$, is lower bounded as (see Theorem~\ref{thm:one-shot:converse})
\begin{align}\label{eq:meta-converse}
{M_\epsilon^\star(W_{\rv{Y}|\rv{X}})} \geq 
 \ceil{ 2^{I_{\max}^\epsilon(W_{\rv{Y}|\rv{X}})} }, 
\end{align}
Note that the $\epsilon$-smooth channel max-mutual information is a linear program, which we call \emph{meta-converse for channel simulation}.
Our derivations are based on information inequalities for partially smoothed entropy measures~\cite{anshu2020partially}.

We then go on to show that $\ceil{2^{I_{\max}^\epsilon(W)}}$ corresponds exactly to the simulation cost with no-signaling-assisted codes.
This is in analogy to similar findings for the meta-converse for channel coding in terms of hypothesis testing.
Namely, with $N_\epsilon^\star(W_{\rv{Y}|\rv{X}})$ denoting the largest number of distinct messages that can be transmitted through $W_{\rv{Y}|\rv{X}}$ within the average error $\epsilon$, one has the meta-converse bound~\cite{polyanskiy2010channel} (also see~\cite{hayashi2009information})
\begin{equation}
N_\epsilon^\star(W_{\rv{Y}|\rv{X}}) \leq \floor{ \adjustlimits\sup_{p_\rv{X}\in\set{P}(\set{X})} \inf_{q_\rv{Y}\in\set{P}(\set{Y})}  {\frac{1}{\beta^\star_\epsilon(p_\rv{X}\cdot W_{\rv{Y}|\rv{X}},p_\rv{X}\times q_\rv{Y})}} },
\end{equation}
where the value on the right-hand side corresponds exactly to the number of distinct messages that can be transmitted over the channel with average error $\epsilon$ using no-signaling-assisted codes~\cite{matthews2012linear}.
\smallskip


\subsubsection{One-Shot Achievability via Rejection-Sampling}

We find that our meta-converse is also achievable up to small fudge terms with common randomness-assisted codes.
That is, combined with the meta-converse, we have (see Theorem~\ref{thm:one-shot:achievability})
\begin{equation}\label{eq:achievability}
I_{\max}^{\epsilon}(W_{\rv{Y}|\rv{X}}) \leq \log{M_\epsilon^\star(W_{\rv{Y}|\rv{X}})}\leq I_{\max}^{\epsilon-\delta}(W_{\rv{Y}|\rv{X}}) + \log\log{\frac{1}{\delta}} \quad \text{for any} \quad 0<\delta\leq\epsilon.
\end{equation}
These one-shot bounds are tight in a fixed error asymptotic analysis for memoryless channels up to logarithmic terms, since both the upper and lower bounds are stated in terms of the same quantity up to slack terms that grow slower than logarithmic in $n$.\footnote{To convince ourselves of this, we may recall that in an asymptotic analysis for memoryless channels with block length $n$, the parameter $\delta$ is commonly chosen of the order $\Omega(1/\sqrt{n})$ as this yields only a constant penalty in the second-order term. This then yields a gap between meta-converse and the above achievability bound that scales at most as $\log \log n$.}
Our achievability proofs are based on rejection sampling techniques, which is an established method in statistics (see, \eg,~\cite{vonNeumann1951Various} and~\cite[Chapter 2.3]{robert1999monte}, and has seen its use in various problems in coding and information theory (see, \eg,~\cite{jain2003direct}).


\subsection{Asymptotic results}


\subsubsection{Finite Block-Length Computation}

For DMCs, our meta-converse can be efficiently computed for small blocklength $n\in\mathbb{N}$.
In particular, $I_{\max}^\epsilon(W_{\rv{Y}|\rv{X}}^{\tensor n})$ for $n$ repetitions of a memoryless binary symmetric channel is a linear program whose computational complexity grows only polynomially in $n$.
We showcase this by a numerical example in Fig.~\ref{fig:BSC:NS:sim-intro}, where we plot the regularized $\epsilon$-smooth channel max-information for the binary symmetric channel.


\begin{figure}
\includegraphics{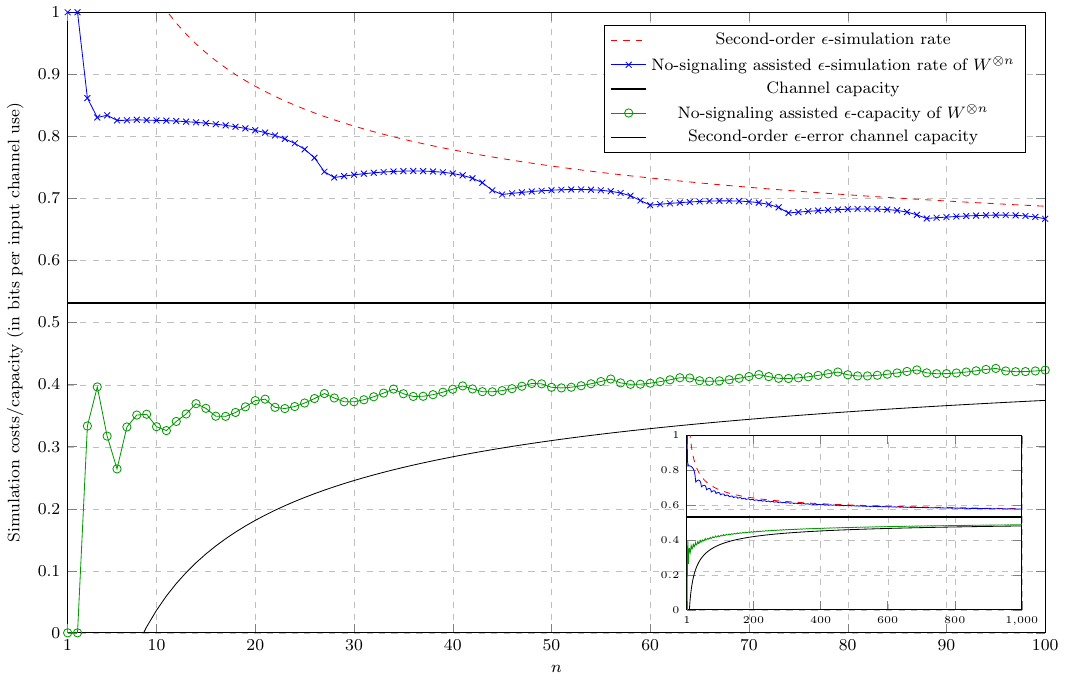}
\caption{Consider i.i.d.~copies of a binary symmetric channel $W$ with crossover probability $\delta = 0.1$ and an error parameter $\epsilon = 0.05$.
The figure shows the second-order approximation of the common randomness assisted $\epsilon$-error simulation cost and the exact no-signaling assisted $\epsilon$-error simulation cost. This is contrasted with the corresponding quantities for channel coding, as discussed in~\cite{polyanskiy2010channel}, for example. The capacity of the channel is achieved for both tasks in the asymptotic limit.}
\label{fig:BSC:NS:sim-intro}
\end{figure}


\subsubsection{Asymptotic Expansions}

For any $\epsilon\in(0,1)$ and $\delta>0$ small enough, we establish the following relationship between the channel simulation task and the channel coding task, \ie, (see Theorem~\ref{thm:cs:cc})
\begin{equation}
\log{N_{1-\epsilon-\delta}^\star(W_{\rv{Y}|\rv{X}})} + \log{\delta} \leq \log{M_\epsilon^\star(W_{\rv{Y}|\rv{X}})} \leq \log{N_{1-\epsilon+\delta}^\star(W_{\rv{Y}|\rv{X}})} + \log{\frac{2}{\delta}} + \log{\log{\frac{4}{\delta^2}}}.
\end{equation}
Using this relationship and asymptotic expansions of channel coding in the small deviation regime~\cite{strassen1962asymptotische, polyanskiy2010channel, hayashi2009information}
\begin{equation}\label{eq:second-order-coding}
\log{N_\epsilon^\star(W_{\rv{Y}|\rv{X}}^{\tensor n})} = n\cdot C(W_{\rv{Y}|\rv{X}}) + \sqrt{n V_\epsilon(W_{\rv{Y}|\rv{X}})}\cdot \Phi^{-1}(\epsilon) + O(\log{n}),
\end{equation}
we find asymptotic expansions for the channel simulation in the small-deviation regime as a direct result of Theorem~\ref{thm:cs:cc} (see Corollary~\ref{cor:cs:expand}), \ie, 
\begin{equation}
\label{eq:second-order-simulation}
\log{M_\epsilon^\star(W_{\rv{Y}|\rv{X}}^{\tensor n})} = n\cdot C(W_{\rv{Y}|\rv{X}}) + \sqrt{n V_{1-\epsilon}(W_{\rv{Y}|\rv{X}})}\cdot \Phi^{-1}(1-\epsilon) + O(\log{n}).
\end{equation}
In the moderate deviation regime, we bound the $n$-fold channel simulation cost using the information spectrum divergence, and consider the asymptotic expansion of the latter in the moderate-deviation regime.
This yields the following expansion (see Theorem~\ref{thm:cs:expand:moderate}).
\begin{align}
\frac{1}{n} \log{M_{\epsilon_n}^\star(W_{\rv{Y}|\rv{X}}^{\tensor n})} &= C(W_{\rv{Y}|\rv{X}}) + \sqrt{2V_{\max}(W_{\rv{Y}|\rv{X}})}\cdot a_n + o(a_n).
\end{align}
where $\epsilon_n\defeq 2^{-na_n^2}$ and $\{a_n\}_{n=1}^\infty$ is some moderate sequence.

As a numerical example, we plot in Fig.~\ref{fig:BSC:NS:sim-intro}, for the binary symmetric channel, the meta-converse and its asymptotic second-order expansion for both channel simulation and channel coding.
Whereas the asymptotic first-order rate is given by the same number for both tasks\,---\,the channel capacity\,---\,in the finite-blocklength regime there is a second-order gap
\begin{equation}
\left(\sqrt{n V_{1-\epsilon}(W_{\rv{Y}|\rv{X}})}+\sqrt{n V_\epsilon(W_{\rv{Y}|\rv{X}})}\right)\cdot \Phi^{-1}(\epsilon)
\end{equation}
between the two values.
In particular, when $\epsilon<\frac{1}{2}$, as one can expect from~\eqref{eq:second-order-coding} and~\eqref{eq:second-order-simulation}, the second-order simulation rate and the second-order coding rate approach the capacity from above and from below, respectively.
As a consequence, even if we allow no-signaling-assisted codes, the asymptotic first-order reversibility of channel interconversion breaks down in second-order\,---\,unless the channel dispersion terms become zero. 
For example, the conversion between the ternary identity channel and the following ternary channel is reversible up to second order:
\begin{equation}
W_{\rv{Y}|\rv{X}}(y|x) = \begin{cases} \frac{1}{2} & \text{if }x\neq y \\ 0 &\text{otherwise}\end{cases},
\end{equation}
where $x,y\in\{0,1,2\}$.


\subsection{Results on broadcast channels}


\subsubsection{Broadcast Channel Simulation}\label{sec:broadcast}
We extend our results to network topologies and derive a channel simulation theorem for broadcast channels\footnote{Please refer to the beginning of Section~\ref{sec:broadcast_channels} for a detailed description of the setup of the task of simulating a broadcast channel.}.
We provide a detailed discussion for the case where $K=2$.
Namely, for the task of simulating $W_{\rv{YZ}|\rv{X}}$ under common randomness assistance, we characterize the asymptotic simulation rate region $\set{R}^\star_\epsilon(W_{\rv{YZ}|\rv{X}})$, \ie, the closure of the set of all attainable \emph{rate} pairs $(r_1, r_2)$ of $\epsilon$-simulation codes for $W_{\rv{YZ}|\rv{X}}^{\tensor n}$ for $n$ large enough, as (see Theorem~\ref{thm:broadcast:asymptotic})
\begin{equation}
\set{R}^\star_\epsilon(W_{\rv{YZ}|\rv{X}})
= \left\{(r_1,r_2)\in\mathbb{R}_{>0}^2 \middle\vert r_1 \geq C(W_{\rv{Y}|\rv{X}}),\,
    r_2 \geq C(W_{\rv{Z}|\rv{X}}),\,
    r_1+r_2\geq \tilde{C}(W_{\rv{YZ}|\rv{X}}) \right\}.
\end{equation}
Here, $W_{\rv{Y}|\rv{X}}$ and $W_{\rv{Z}|\rv{X}}$ are reduced channels of $W_{\rv{YZ}|\rv{X}}$.
Our achievability bounds utilize a modified version of the bipartite convex split lemma from~\cite{anshu2017unified, anshu2017quantum}, while the converse bound uses similar inequalities as in the point-to-point case.

Our characterizations give a direct operational interpretation to the multipartite mutual information of broadcast channels.
We can also efficiently compute the rate region $\set{R}^\star_\epsilon(W_{\rv{YZ}|\rv{X}})$ using iterative Blahut--Arimoto techniques.
As the inequality 
\begin{equation}
\tilde{C}(W_{\rv{YZ}|\rv{X}})\geq C(W_{\rv{Y}|\rv{X}})+C(W_{\rv{Z}|\rv{X}})
\end{equation}
is typically strict, the sum rate constraint on $r_1+r_2$ in terms of the multipartite mutual information is in general necessary.
Fig.~\ref{fig:degraded_BSC_region-intro} shows an example in which this happens.
\begin{figure}
\centering\includegraphics{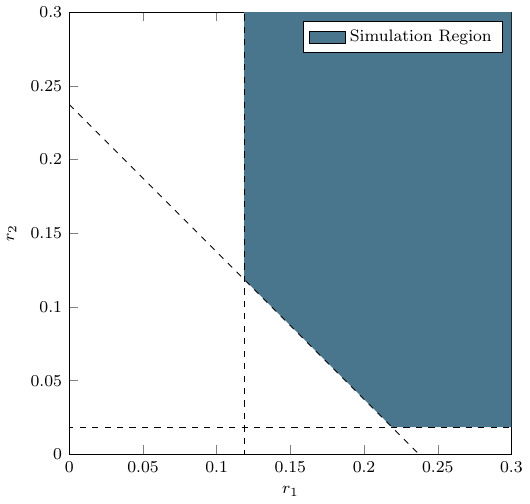}
\caption{Asymptotic simulation region for the broadcast channel $W_{\rv{YZ}|\rv{X}} = \fnc{BSC}_{\rv{Z}|\rv{Y}}\circ \fnc{BSC}_{\rv{Y}|\rv{X}}$, where $\fnc{BSC}_{\rv{Z}|\rv{Y}}$ and $\fnc{BSC}_{\rv{Y}|\rv{X}}$ are binary symmetric channels with crossover probability $\delta = 0.3$.}
\label{fig:degraded_BSC_region-intro}
\end{figure}
We conclude that unlike the capacity region of broadcast channels\,---\,for which in general only non-matching inner and outer bounds are known~\cite{cover1998comments}\,---\,our characterization of the simulation region is exact, has a single-letter form, and can be computed efficiently.
Finally, by sub-additivity of the Shannon entropy we additionally have $\tilde{C}(W_{\rv{YZ}|\rv{X}})\geq C(W_{\rv{YZ}|\rv{X}})$, which shows that employing a global decoder for channel simulation\footnote{Namely, we treat the two receivers as a single receiver with access to both messages $\rv{M}$ and $\rv{N}$. This will make the task at least no more difficult, and any converse bound on the message size will automatically be a lower bound for $\size{\set{M}}\cdot\size{\set{N}}$.} would lead to a less restrictive sum rate constraint.
However, by the same argument of a global decoder, $C(W_{\rv{YZ}|\rv{X}})$ is also an outer bound on the channel capacity region of broadcast channels.
Consequently, we find that common randomness\,---\,or even no-signaling\,---\,assisted broadcast channel interconversion already becomes asymptotically irreversible in the first order.
This is in contrast to the point-to-point case, where we have seen that irreversibility only appears in the second order.

For general $K$-receiver broadcast channels $W_{\rv{\rvs{Y}}_1^K|\rv{X}}$ from $\set{X}$ to $\set{Y}_1\times\cdots\times\set{Y}_K$, we find the single-letter characterization (see Theorem~\ref{thm:broadcast:asymptotic:K})
\begin{equation}
\set{R}^\star_\epsilon(W_{\rvs{Y}_1^K|\rv{X}})=\left\{(r_1,\cdots,r_n)\in\mathbb{R}_{>0}^n \,\middle|\, \sum_{i\in\set{J}} r_i\geq\tilde{C}(W_{\rvs{Y}_\set{J}|\rv{X}})\quad\forall \set{J}\subset\{1,\ldots,K\}\text{ with }\set{J}\neq\emptyset\right\}
\end{equation}
where for each $k=1,\ldots,K$, 
\begin{equation}\label{eq:def:K-partite:tilde:C}
\tilde{C}(W_{\rvs{Y}_1^k|\rv{X}})\defeq \max_{p_\rv{X}\in\set{P}(\set{X})} I(\rv{X};\rv{Y}_1;\ldots;\rv{Y}_k)_{p_\rv{X}\cdot W_{\rvs{Y}_1^k|\rv{X}}}, 
\end{equation}
and the multipartite mutual information $I(\rv{X};\rv{Y}_1;\ldots;\rv{Y}_k)$ is defined in~\eqref{def:K-partite:Imax}.


\section{One-Shot Bounds for Randomness-Assisted Channel Simulation}
\label{sec:One-Shot(p2p)}

In this section, we consider the problem of simulating a single copy of a channel described in Section~\ref{sec:overview:setup}.
Recall that, given $W_{\rv{Y}|\rv{X}}\in\set{P}(\set{Y}|\set{X})$, a size-$M$ $\epsilon$-simulation code is a triple $(\mathcal{E}_{\rv{M}|\rv{XS}},\mathcal{D}_{\rv{Y}|\rv{MS}},\rv{S})$ such that 
\begin{equation} 
\norm{\tilde{W}_{\rv{Y}|\rv{X}}-W_{\rv{Y}|\rv{X}}}_\fnc{tvd} \leq \epsilon
\end{equation}
where 
\begin{equation}
    \tilde{W}_{\rv{Y}|\rv{X}}(y|x) \defeq \sum_{s\in\set{S}} p_\rv{S}(s) \cdot \sum_{i\in[M]} \mathcal{E}_{\rv{M}|\rv{XS}}(i|x,s) \cdot \mathcal{D}_{\rv{Y}|\rv{M}\rv{S}}(y|i,s)
\end{equation}
(also see~\ref{sec:metaconverse} and Fig.~\ref{fig:simulation:task:random}).
An integer $M$ is said to be \emph{attainable} for a given $\epsilon\in[0,1]$ if there exists a size-$M$ $\epsilon$-simulation code.
We present a pair of achievability and converse bounds for $M_\epsilon^\star(W_{\rv{Y}|\rv{X}})$ the minimal attainable size of $\epsilon$-simulation codes for channel $W_{\rv{Y}|\rv{X}}$.
The achievability bound applies the rejection sampling method and is inspired by~\cite{jain2003direct}. 
The converse bound is the result of a series of relationships among information quantities.


\subsection{Achievability Bound}

We use the rejection sampling method for channel simulation.
To simulate the output of a distribution $p_\rv{Y}\in\set{P}(\set{Y})$ using many copies of i.i.d. random variables $\rv{Y}_1\rv{Y}_2\ldots\sim q_\rv{Y}\times q_\rv{Y} \times \cdots$, rejection sampling applies the accept-reject algorithm for each $\rv{Y}_i$, \ie, $\rv{Y}_i$ is ``accepted'' and chosen as an output to simulate $p_\rv{Y}$ with probability $\lambda\cdot\frac{p_\rv{Y}(\rv{Y}_i)}{q_\rv{Y}(\rv{Y}_j)}$ (for some normalizer $\lambda>0$ such that this number is bounded by $1$).
The procedure is stopped if $\rv{Y}_i$ is accepted, or moves onto $\rv{Y}_{i+1}$ otherwise.
We formalize the details of this procedure and its analysis as the following lemma.
\begin{lemma}\label{lem:rejection:sampling}
Let $p_\rv{Y}, q_\rv{Y}\in\set{P}(\set{Y})$ be such that $p\ll q$.
Let $M \geq 1$ be an integer.
Suppose $\rv{Y}_1$, $\rv{Y}_2$, \ldots, $\rv{Y}_M$ are i.i.d. random variables where $\rv{Y}_j$ is distributed according to $q_\rv{Y}$ for all $j=1,\ldots,M$.
Let $\lambda\defeq\big(\max_{y\in\set{Y}} \frac{p_\rv{Y}(y)}{q_\rv{Y}(y)}\big)^{-1} = 2^{-D_{\max}\infdiv*{p_\rv{Y}}{q_\rv{Y}}}$.
We generate~$\tilde{\rv{Y}}$ conditioned on $\rvs{Y}_1^M$ using the following procedure:
\begin{algorithmic}[1]
    \State{$j\gets 0$;}
    \Do
    \State{$j\gets j+1$;}
    \State{Generate random number $u\sim\mathcal{U}([0,1])$;}
    \DoWhile{$j<M$ and $u>\lambda \cdot \frac{p_\rv{Y}(\rv{Y}_j)}{q_\rv{Y}(\rv{Y}_j)}$}
    \State{$\tilde{\rv{Y}}\gets\rv{Y}_j$}
\end{algorithmic}
The procedure will create a random variable $\tilde{\rv{Y}}$ jointly distributed with $\rv{Y}_1,\ldots,\rv{Y}_M$.
Let $\tilde{p}_\rv{Y}$ denote the marginal distribution of $\tilde{\rv{Y}}$.
For any $\epsilon\in(0,1)$, if $M$ is large enough such that
\begin{equation}
\label{eq:rejection:sampling:M}
\left( 1-\lambda\right)^M \leq \epsilon,
\end{equation}
then
\begin{equation}
\norm{p_\rv{Y} - \tilde{p}_\rv{Y}}_{\fnc{tvd}} \leq \epsilon.
\end{equation}
\end{lemma}
\begin{proof}
Consider the procedure stated above.
Let $\rv{A}$ be a binary random variable such that $\rv{A}=1$ if the program exited the loop (lines 2--5) due to a violation of the second condition ``$u>\lambda \cdot \frac{p_\rv{Y}(\rv{Y}_j)}{q_\rv{Y}(\rv{Y}_j)}$''.
Otherwise, let $\rv{A}=0$.
(In other words, $\rv{A}$ equals $1$ if one of the $\rv{Y}_i$'s was ``accepted'', and equals $0$ otherwise.)
By direct computation, we have
\begin{align}
p_{\rv{A}}(0) &= \sum_{y_1,\ldots,y_M} \prod_{i=1}^M q_\rv{Y}(y_i) \cdot \prod_{i=1}^M \left(1-\lambda\cdot \frac{p_\rv{Y}(y_i)}{q_\rv{Y}(y_i)}\right) \\
&= \sum_{y_1,\ldots,y_M} \prod_{i=1}^M \left(q_\rv{Y}(y_i) - \lambda\cdot p_\rv{Y}(y_i)\right)\\
\label{eq:rejection:sampling:reject:prob}
&= \left(1-\lambda\right)^M
\leq \epsilon
\end{align}
where we used~\eqref{eq:rejection:sampling:M} (together with the definition of $\lambda$) for~\eqref{eq:rejection:sampling:reject:prob}.
On the other hand, for each $y\in\set{Y}$,
\begin{align}
\label{eq:rejection:sampling:accept:y:prob}
p_{\tilde{\rv{Y}}|\rv{A}}(y|1)
&= \frac{1}{p_\rv{A}(1)} \cdot \sum_{j=1}^{M} p_{\tilde{Y}\rv{JA}}(y,j,1)
= \frac{1}{p_\rv{A}(1)} \cdot \sum_{j=1}^{M} p_{\rv{Y}_j\rv{JA}}(y,j,1)
= \frac{1}{p_\rv{A}(1)} \cdot \sum_{j=1}^{M} p_{\rv{JA}|\rv{Y}_j}(j,1|y) \cdot q_\rv{Y}(y)\\
& = \frac{1}{1-p_\rv{A}(0)} \cdot \sum_{j=1}^{M} \sum_{y_1\ldots,y_{j\!-\!1}} \prod_{i=1}^{j-1} q_\rv{Y}(y_i) \cdot \prod_{i=1}^{j-1} \left(1-\lambda\cdot\frac{p_\rv{Y}(y_i)}{q_\rv{Y}(y_i)}\right) \cdot \lambda \cdot \frac{p_\rv{Y}(y)}{q_\rv{Y}(y)} \cdot q_\rv{Y}(y)\\
& = \frac{1}{1-\left(1-\lambda\right)^M} \cdot \sum_{j=1}^{M} \left(1-\lambda\right)^{j-1} \cdot \lambda \cdot p_\rv{Y}(y) = p_\rv{Y}(y)
\end{align}
where in~\eqref{eq:rejection:sampling:accept:y:prob}, we summed over the probability of ``acceptance'' at $j$-th attempt while $\rv{Y}_j=y$.
Thus, 
\begin{align}
\norm{p_\rv{Y}-\tilde{p}_\rv{Y}}_{\fnc{tvd}} &= \norm{p_\rv{Y}- (p_\rv{A}(1)\cdot p_\rv{Y}+p_\rv{A}(0)\cdot p'_\rv{Y})}_{\fnc{tvd}}
= \norm{p_\rv{A}(0)\cdot (p_\rv{Y}-p'_\rv{Y})}_{\fnc{tvd}} \leq p_\rv{A}(0) \leq \epsilon
\end{align}
where $p'_\rv{Y}$ is some pmf on $\set{Y}$.
\end{proof}

Lemma~\ref{lem:rejection:sampling} provides a quantitative measure of the trade-offs between the accuracy of the rejection sampling and the length of the random sequence to be sampled from.
Observe that this method only requires the knowledge of the target distribution, and its performance is guaranteed as long as the sampled sequence is long enough.
This enables the technique to be used for channel simulation.
In particular, upon receiving the input $x$, Alice (the sender) applies rejection sampling on a shared random sequence for the desired output distribution $W(\cdot|x)$, and \emph{send} the \emph{index} of the accepted instance to Bob (the receiver).
Bob finishes the sampling by reading out the indexed instance from the shared random sequence.
The details are listed in the following theorem and its proof.
\begin{theorem}[One-Shot Achievability Bound]\label{thm:one-shot:achievability}
Let $W_{\rv{Y}|\rv{X}}$ be a channel from $\set{X}$ to $\set{Y}$, and let $\epsilon\in(0,1)$.
It holds that 
\begin{equation}\label{eq:achievable:bound}
    \log{M_\epsilon^\star(W_{\rv{Y}|\rv{X}})} \leq I_{\max}^{\epsilon-\delta}(W_{\rv{Y}|\rv{X}}) + \log{\log{\frac{1}{\delta}}}.
\end{equation}
for all $\delta\in(0,\epsilon]$.
\end{theorem}
\begin{proof}
For arbitrary $q_\rv{Y}\in\set{P}(\set{Y})$ and channel $\tilde{W}_{\rv{Y}|\rv{X}}$ $(\epsilon-\delta)$-close to $W_{\rv{Y}|\rv{X}}$ in TVD, we describe the following protocol for simulating the channel $\tilde{W}_{\rv{Y}|\rv{X}}$ within a TVD tolerance of $\delta$ where $M$ is some positive integer to be determined later.
\begin{enumerate}
\item Let Alice and Bob share i.i.d. random variables $\rv{Y}_1$, \ldots, $Y_M$, each distributed according to $q_\rv{Y}$. (Here, $(\rv{Y}_1,\ldots\rv{Y}_n)$ serves as the common randomness $\rv{S}$ in the original setup.)
\item Upon knowing $x\in\set{X}$, Alice applies the accept-reject algorithm described in the previous lemma with parameters $M$, $p_\rv{Y}\gets \tilde{W}_{\rv{Y}|\rv{X}}(\cdot|x)$ and $q_\rv{Y}$, and sends the index $j$ to Bob without loss (using a code of size $M$).
\item Upon receiving $j$, Bob picks $\rv{Y}_j$ from the shared randomness and uses it as the output of the simulated channel.
\end{enumerate}
By Lemma~\ref{lem:rejection:sampling}, the channel created by the above protocol is $\delta$-close (in TVD) to $\tilde{W}_{\rv{Y}|\rv{X}}$ if
\begin{equation}
\left( 1-2^{-D_{\max}\infdiv*{\tilde{W}_{\rv{Y}|\rv{X}}(\cdot|x)}{q_\rv{Y}}}\right)^M \leq \delta
\end{equation}
for all $x\in\set{X}$.
In this case, by the triangular inequality of the TVD for conditional probabilities, the simulated channel is $\epsilon$-close (in TVD) to $W_{\rv{Y}|\rv{X}}$.
In other words, the $M$ chosen above is an upper bound of $M_\epsilon^\star(W_{\rv{Y}|\rv{X}})$.
Since $q_\rv{Y}\in\set{P}(\set{Y})$ and $\tilde{W}_{\rv{Y}|\rv{X}}$ (which is $(\epsilon-\delta)$-close  to $W_{\rv{Y}|\rv{X}}$) in the above protocol can be chosen arbitrarily, by optimizing these two terms, we have 
\begin{align}
M_\epsilon^\star(W_{\rv{Y}|\rv{X}}) &\leq
\multiadjustlimits{
   \inf_{q_\rv{Y}\in\set{P}(\set{Y})}, \inf_{\tilde{W}_{\rv{Y}|\rv{X}}\in\set{P}(\set{Y}|\set{X}):\norm{\tilde{W}_{\rv{Y}|\rv{X}}-W_{\rv{Y}|\rv{X}}}_{\fnc{tvd}}\leq \epsilon-\delta}, 
   \sup_{x\in\set{X}}
   } \frac{\log{\delta}}{\log\left( 1-2^{-D_{\max}\infdiv*{\tilde{W}_{\rv{Y}}|\rv{X}(\cdot|x)}{q_\rv{Y}}}\right)} \\
\label{eq:one-shot:achievability:0}
&=\log{\delta}\cdot \left\{\log{\left[1-2^{-\multiadjustlimits{ 
    \inf_{q_\rv{Y}\in\set{P}(\set{Y})}, 
    \inf_{\tilde{W}_{\rv{Y}|\rv{X}}\in\set{P}(\set{Y}|\set{X}): \norm{\tilde{W}_{\rv{Y}|\rv{X}}-W_{\rv{Y}|\rv{X}}}_{\fnc{tvd}}\leq \epsilon-\delta}, 
    \max_{x\in\set{X}}
    } D_{\max}\infdiv*{\tilde{W}_{\rv{Y}}|\rv{X}(\cdot|x)}{q_\rv{Y}}}\right]}\right\}^{-1} \\
\label{eq:one-shot:achievability:1}
&= \log{\delta}\cdot \left(\log{\left(1-2^{-I_{\max}^{\epsilon-\delta}(W_{\rv{Y}|\rv{X}})}\right)}\right)^{-1}
= \log{\frac{1}{\delta}}\cdot \left(-\log{\left(1-2^{-I_{\max}^{\epsilon-\delta}(W_{\rv{Y}|\rv{X}})}\right)}\right)^{-1}\\
\label{eq:one-shot:achievability:2}
&\leq \log{\frac{1}{\delta}}\cdot 2^{I_{\max}^{\epsilon-\delta}(W_{\rv{Y}|\rv{X}})}
\end{align}
where we used the inequality $\log{(1-x)}\leq -x$ when $x\in[0,1)$ for~\eqref{eq:one-shot:achievability:2}.
Finally,~\eqref{eq:achievable:bound} is obtained by taking the logarithm on both sides.
\end{proof}


\subsection{Converse Bound}\label{sec:converse:oneshot}
We exploit the Markov structure $\rv{X} - \rv{SM} - \rv{Y}$ for the converse bound.
Here, $\rv{S}$ is the shared randomness, and $\rv{M}$ is the transmitted message to simulate the channel.
On the technical side, we utilize (and provide a simpler proof of) the classical special case of the non-lockability of the smooth max-relative entropy~\cite{berta2013quantum} from quantum information theory.

Given a pair of discrete random variables $\rv{X}\rv{Y}$ with joint pmf $p_\rv{XY}\in\set{P}(\set{X}\times\set{Y})$, the max-mutual information of $\rv{X}$ \vs $\rv{Y}$ is defined as
\begin{equation}\label{eq:def:Imax}
I_{\max}(\rv{X};\rv{Y}) \defeq \inf_{q_\rv{Y}\in\set{P}(\set{Y})} D_{\max}\infdiv*{p_\rv{XY}}{p_\rv{X}\times q_\rv{Y}}
\end{equation}
where $p_\rv{X}$ is the marginal distribution of $\rv{X}$ induced from $p_\rv{XY}$.
\begin{lemma}[Classical Special Case of~{\cite[Cor.~A.14]{berta2013quantum}}]\label{lem:non-lockability}
    Let $(\rv{X},\rv{Y},\rv{Z})$ be joint random variables distributed on $\set{X}\times\set{Y}\times\set{Z}$.
    In particular, suppose that the set $\set{Z}$ is finite.
    Then
    \begin{equation}
        I_{\max}(\rv{X};\rv{YZ}) \leq I_{\max}(\rv{X};\rv{Y}) + \log{\size{\set{Z}}}.
    \end{equation}
\end{lemma}
This lemma is a special case of \cite[Cor.~A.14]{berta2013quantum}. 
In the original lemma, the first two random variables $\rv{X}$ and $\rv{Y}$ are instead some quantum systems.
Due to the classical nature of the statement we need, we are able to provide an alternative proof which is significantly simpler.
\begin{proof}
Let $q_\rv{Y}^\star$ denote the optimal distribution on $\set{Y}$ achieving the infimum in the definition of $I_{\max}(\rv{X};\rv{Y})$.
Then
\begin{align}
    I_{\max}(\rv{X};\rv{YZ}) &= \inf_{q_\rv{YZ}\in\set{P}(\set{Y}\times\set{Z})} D_{\max}\infdiv*{p_\rv{XYZ}}{p_\rv{X}\times q_\rv{YZ}}\\
    &\leq D_{\max}\infdiv*{p_\rv{XYZ}}{p_\rv{X}\times q_\rv{Y}^\star \times \frac{1}{\size{\set{Z}}}}\\
    &=\log{\left( \size{\set{Z}} \cdot \sup_{x,y,z} \frac{p_\rv{XY}(x,y)\cdot p_{\rv{Z}|\rv{XY}}(z|x,y)}{p_\rv{X}\cdot q_\rv{Y}^\star(y)} \right)}\\
    &\leq \log{\left( \sup_{x,y} \frac{p_\rv{XY}(x,y)}{p_\rv{X}\cdot q_\rv{Y}^\star(y)} \right)} + \log{\size{\set{Z}}}\\
    &= I_{\max}(\rv{X};\rv{Y}) + \log{\size{\set{Z}}}. \qedhere
\end{align}
\end{proof}
\begin{theorem}[One-Shot Converse Bound]\label{thm:one-shot:converse}
Let $W_{\rv{Y}|\rv{X}}$ be a channel from $\set{X}$ to $\set{Y}$, and let $\epsilon\in(0,1)$.
Then
\begin{equation}\label{eq:converse:bound}
\log{M^\star_\epsilon(W_{\rv{Y}|\rv{X}})} \geq I_{\max}^\epsilon(W_{\rv{Y}|\rv{X}})
\end{equation}
and, since $M^\star_\epsilon(W_{\rv{Y}|\rv{X}})$ is an integer, it must also satisfy $M^\star_\epsilon(W_{\rv{Y}|\rv{X}}) \geq \left\lceil 2^{I_{\max}^\epsilon(W_{\rv{Y}|\rv{X}})} \right\rceil$.
\end{theorem}
\begin{proof}
Suppose there exists some size-$M$ $\epsilon$-simulation code for $W_{\rv{Y}|\rv{X}}$.
Let the random variable $\rv{S}$ denote the shared randomness, and $\rv{M}$ denote the codeword transmitted (see Fig.~\ref{fig:simulation:task:random}).
Then, for \emph{any} input source $\rv{X}\sim p_\rv{X}$, we have a Markov chain $\rv{X}-\rv{MS}-\rv{Y}$ where
\begin{itemize}
    \item The distribution of $\rv{XY}$, denoted by $\tilde{p}_\rv{XY}$, is $\epsilon$-close (in TVD) to $p_\rv{XY}\defeq p_\rv{X}\cdot W_{\rv{Y}|\rv{X}}$.
    \item The marginal distribution $\tilde{p}_\rv{X}=\sum_{y}\tilde{p}_\rv{XY}(\cdot,y)=p_\rv{X}$.
    \item $\rv{M}$ is distributed on $\{1,\ldots,M\}$.    
\end{itemize}
Choose $p_\rv{X}$ to be some pmf with full support.
We have the following series of inequalities.
\begin{enumerate}
    \item Using the fact that $\rv{X}$ and $\rv{S}$ are independent, and by Lemma~\ref{lem:non-lockability}, we have 
    \begin{equation}
    \log{M} = I_{\max}(\rv{X};\rv{S}) + \log{M} \geq I_{\max}(\rv{X};\rv{MS}).
    \end{equation}
    \item By the data-processing inequality of $I_{\max}$, we have 
    \begin{equation}
    I_{\max}(\rv{X};\rv{MS}) \geq I_{\max}(\rv{X};\rv{Y})
    \end{equation}
    \item By the definition of $I_{\max}$ and the fact that $\tilde{p}_\rv{X}=p_\rv{X}$, we have
    \begin{align}
    I_{\max}(\rv{X};\rv{Y}) &= \inf_{q_\rv{Y}} D_{\max}\infdiv*{\tilde{p}_\rv{XY}}{p_\rv{X}\times q_\rv{Y}}\\
    &= \adjustlimits\inf_{q_\rv{Y}} \max_{x} D_{\max}\infdiv*{\tilde{p}_{\rv{Y}|\rv{X}}(\cdot|x)}{q_\rv{Y}}\\
    \label{proof:one-shot:converse:1}
    &\geq \multiadjustlimits{
    \inf_{q_\rv{Y}}, 
    \inf_{\tilde{W}_{\rv{Y}|\rv{X}}\in\set{P}(\set{Y}|\set{X}): \norm{\tilde{W}_{\rv{Y}|\rv{X}}-W_{\rv{Y}|\rv{X}}}_\fnc{tvd}\leq\epsilon}, 
    \max_{x}
    } D_{\max}\infdiv*{\tilde{W}_{\rv{Y}|\rv{X}}(\cdot|x)}{q_\rv{Y}}\\
    \label{proof:one-shot:converse:2}
    &= I_{\max}^\epsilon(W_{\rv{Y}|\rv{X}})
    \end{align}
    where~\eqref{proof:one-shot:converse:1} holds since $\tilde{p}_{\rv{Y}|\rv{X}}$, \ie, the simulated channel as the result of the size-$M$ simulation code, is $\epsilon$-close to $W_{\rv{Y}|\rv{X}}$ in TVD, and~\eqref{proof:one-shot:converse:2} holds since
    \begin{equation}
        I_{\max}(W_{\rv{Y}|\rv{X}}) 
        =  \adjustlimits
        \inf_{q_\rv{Y}\in\set{P}(\set{Y})}
        \max_{x\in\set{X}}
        D_{\max}\infdiv*{{W}_{\rv{Y}|\rv{X}}(\cdot|x)}{q_\rv{Y}}.
    \end{equation}
\end{enumerate}

Combining the above, we have shown $\log{M}\geq I_{\max}^\epsilon(W_{\rv{Y}|\rv{X}})$ for any admissible $M$.
Since $M^\star_\epsilon(W_{\rv{Y}|\rv{X}})$ is the infimum of all such $M$'s, the inequality must also hold with $M^\star_\epsilon(W_{\rv{Y}|\rv{X}})$.
\end{proof}


\section{Asymptotic Analysis in Small-Deviation and Moderate-Deviation Regimes}
\label{sec:higher-order(p2p)}

In this section, we consider the problem of simulating $n$ copies of $W_{\rv{Y}|\rv{X}}$, \ie, $W_{\rv{Y}|\rv{X}}^{\tensor n}$, under both fixed TVD-tolerance $\epsilon\in(0,1)$ and sub-exponentially decaying TVD-tolerance
\begin{equation}
    \epsilon_n = 2^{-n a_n^2},
\end{equation}
where $\{a_n\}_{n\in\mathbb{N}}$ is some moderate sequence, \ie, a sequence of positive numbers such that $a_n \to 0$ but $na_n^2\to+\infty$ as $n \to \infty$.
The former is known as the small-deviation or fixed-error regime, whereas the latter is known as the moderate-deviation regime.
We are interested in the asymptotic expansion of $\log{M_{\epsilon_n}^\star(W_{\rv{Y}|\rv{X}}^{\tensor n})}$ (as a function of $n$) in both cases.

\subsection{Small-Deviation / Fixed-Error Regime}

For the small-deviation analysis, we exploit the relationship between the minimal channel simulation cost and the maximal channel coding size (under the average error).
This is possible with the help of several inequalities among the smoothed max divergence, the hypothesis testing divergence, and the information spectrum divergence.
Intuitively, such a connection should exist from an operational point of view.
Despite that one would also end up with similar results by retracing the small-deviation analysis for channel coding, \eg,~\cite{tomamichel2013tight}, we chose to directly explore the connection between the channel coding and channel simulation, as the latter is less tedious yet more insightful.
\begin{theorem}\label{thm:cs:cc}
Let $W_{\rv{Y}|\rv{X}}$ be a channel and $\epsilon\in(0,1)$.
For each $n \in \mathbb{N}$.
Let $N_\epsilon^\star(W_{\rv{Y}|\rv{X}})$ denote the largest integer $N$ such that there exists a size-$N$ code with transmission error probability (over $W_{\rv{Y}|\rv{X}}$) at most $\epsilon$ given equiprobable codewords.
It holds that
\begin{equation}\label{eq:cs:cc}
\log{N^\star_{1-\epsilon-\delta}(W_{\rv{Y}|\rv{X}})} + \log{\delta} \leq \log{M_\epsilon^\star(W_{\rv{Y}|\rv{X}})} \leq \log{N^\star_{1-\epsilon+\delta}(W_{\rv{Y}|\rv{X}})} + \log{\frac{2}{\delta}} + \log{\log{\frac{4}{\delta^2}}}
\end{equation}
for any $\delta\in(0,\min\{\epsilon,1-\epsilon\})$.
\end{theorem}
We need the following two lemmas for the proof of Theorem~\ref{thm:cs:cc}.
\begin{lemma}\label{lem:Dmax:Dh}
Let $p$ and $q$ be two pmf on the same alphabet.
We have for all $\epsilon\in(0,1)$, $\delta\in(0,1-\epsilon)$ 
\begin{equation}\label{eq:Dmax:Dh}
    D_{\max}^\epsilon\infdiv*{p}{q} \geq \log{\frac{1}{\beta_{1-\epsilon-\delta}\infdiv*{p}{q}}} + \log{\delta}.
\end{equation}
\end{lemma}
\begin{proof}
Suppose $p$ and $q$ are pmfs on $\set{X}$.
Let $\tilde{p}\in\set{P}(\set{X})$ be an optimizer for $D_{\max}^\epsilon\infdiv*{p}{q}$, \ie, $\norm{\tilde{p}-p}_\fnc{tvd}\leq\epsilon$, and $\tilde{p}(x)\leq 2^{D_{\max}^\epsilon\infdiv*{p}{q}} \cdot q(x)$ for all $x\in\set{X}$.
Let $\set{A}\in\set{X}$ be an optimizing subset for $\beta_{1-\epsilon-\delta}\infdiv*{p}{q}$, \ie, $p(\set{A})\geq \epsilon+\delta$ and $q(\set{A}) = \beta_{1-\epsilon-\delta}\infdiv*{p}{q}$.
Combining these two facts, we have
\begin{equation}
\epsilon \geq \norm{\tilde{p}-p}_\fnc{tvd} \geq p(\set{A})-\tilde{p}(\set{A}) \geq \epsilon +\delta - 2^{D_{\max}^\epsilon\infdiv*{p}{q}} \cdot q(\set{A}) = \epsilon +\delta - 2^{D_{\max}^\epsilon\infdiv*{p}{q}} \cdot \beta_{1-\epsilon-\delta}\infdiv*{p}{q}.
\end{equation}
Rearranging the terms, we get
\begin{equation}
2^{D_{\max}^\epsilon\infdiv*{p}{q}} \cdot \beta_{1-\epsilon-\delta}\infdiv*{p}{q} \geq \delta.
\end{equation}
Eq.~\eqref{eq:Dmax:Dh} can be proven by taking the logarithm of the above inequality on both sides and shuffle the term $\log{\beta_{1-\epsilon-\delta}\infdiv*{p}{q}}$ to the other side.
\end{proof}
\begin{lemma}\label{lem:supDs:maxDmax}
Let $W_{\rv{Y}|\rv{X}}$ be a channel from $\set{X}$ to $\set{Y}$ where both $\set{X}$ and $\set{Y}$ are finite sets.
It holds for all $\epsilon\in(0,1)$ that
\begin{equation}\label{eq:supDs:maxDmax:classical}
\sup_{p_\rv{X}\in\set{P}(\set{X})} D_{s+}^{\epsilon}\infdiv*{p_\rv{X}\cdot W_{\rv{Y}|\rv{X}}}{p_\rv{X}\times q_\rv{Y}} + 1 \geq \adjustlimits \inf_{\tilde{W}_{\rv{Y}|\rv{X}}\in\set{P}(\set{Y}|\set{X}):\norm{\tilde{W}_{\rv{Y}|\rv{X}}-W_{\rv{Y}|\rv{X}}}_\fnc{tvd}\leq\epsilon} \max_{x\in\set{X}} D_{\max}\infdiv*{\tilde{W}_{\rv{Y}|\rv{X}}(\cdot|x)}{q_\rv{Y}}
\end{equation}
for any $q_\rv{Y}\in\set{P}(\set{Y})$.
\end{lemma}
The following proof is partially inspired by some of the arguments in the proof of~\cite[Theorem~9]{anshu2020partially}.
\begin{proof}
If the LHS $=\infty$, the inequality holds trivially.
Assuming otherwise, let $a\defeq\sup_{p_\rv{X}\in\set{P}(\set{X})} D_{s+}^{\epsilon}\infdiv*{p_\rv{X}\cdot W_{\rv{Y}|\rv{X}}}{p_\rv{X}\times q_\rv{Y}}$.
By the definition of $D_{s+}^\epsilon$ and the right semi-continuity of $a\mapsto \Pr_{\rv{XY}\sim p_\rv{X}\cdot W_{\rv{Y}|\rv{X}}} [\log{\frac{W_{\rv{Y}|\rv{X}}(\rv{Y}|\rv{X})}{q_\rv{Y}(\rv{Y})}}>a]$, we have, for all $p_\rv{X}\in\set{P}(\set{X})$, 
\begin{equation}
\epsilon
\geq \Pr\nolimits_{\rv{XY}\sim p_\rv{X}\cdot W_{\rv{Y}|\rv{X}}} \left[\log{\frac{W_{\rv{Y}|\rv{X}}(\rv{Y}|\rv{X})}{q_\rv{Y}(\rv{Y})}}>a\right] 
= \sum_{x\in\set{X}} p_\rv{X}(x) \cdot \underbrace{\sum_{y\in\set{Y}} W_{\rv{Y}|\rv{X}}(y|x) \cdot \mathbbm{1}\left\{\log{\frac{W_{\rv{Y}|\rv{X}}(y|x)}{q_\rv{Y}(y)}}>a\right\}}_{\defas\epsilon_x},
\end{equation}
which forces $\epsilon_x\leq \epsilon$ for all $x\in\set{X}$.
For each $x\in\set{X}$, let $\set{A}_x$ denote the following subset of $\set{Y}$
\begin{equation}
\set{A}_x \defeq \left\{y\in\set{Y}: \log{\frac{W_{\rv{Y}|\rv{X}}(y|x)}{q_\rv{Y}(y)}}\leq a\right\}.
\end{equation}
We construct a channel $\hat{W}_{\rv{Y}|\rv{X}}$ as
\begin{equation}
\hat{W}_{\rv{Y}|\rv{X}}(y|x) \defeq \begin{cases}
    W_{\rv{Y}|\rv{X}}(y|x) + \epsilon_x \cdot q_\rv{Y}(y) & \text{if } y\in\set{A}_x, \\
    \epsilon_x \cdot q_\rv{Y}(y) & \text{otherwise}.
\end{cases}
\end{equation}
The construction ensures $\hat{W}_{\rv{Y}|\rv{X}}(\cdot|x) \leq (2^a+\epsilon_x)\cdot q_\rv{Y} \leq (2^a+\epsilon)\cdot q_\rv{Y}$ for all $x$, and thus we have
\begin{equation}\label{eq:supDs:maxDmax:classical:1}
\max_{x\in\set{X}} D_{\max}\infdiv*{\hat{W}_{\rv{Y}|\rv{X}}(\cdot|x)}{q_\rv{Y}} \leq \log\left(2^a + \epsilon\right) \leq a+1 = \sup_{p_\rv{X}\in\set{P}(\set{X})} D_{s+}^{\epsilon}\infdiv*{p_\rv{X}\cdot W_{\rv{Y}|\rv{X}}}{p_\rv{X}\times q_\rv{Y}} + 1.
\end{equation}
On the other hand we can show $\norm{\hat{W}_{\rv{Y}|\rv{X}} - W_{\rv{Y}|\rv{X}}}_{\fnc{tvd}}\leq \epsilon$ as 
\begin{align}
\max_{x\in\set{X}} \norm{\hat{W}_{\rv{Y}|\rv{X}}(\cdot|x) - W_{\rv{Y}|\rv{X}}(\cdot|x)}_{\fnc{tvd}} 
&= \max_{x\in\set{X}} \frac{1}{2}\sum_{y\in\set{A}_x^C} \abs{\epsilon_x q_\rv{Y}(y) - W_{\rv{Y}|\rv{X}}(y|x)} + \sum_{y\in\set{A}_x} \epsilon_x q_\rv{Y}(y)\\
&\leq \max_{x\in\set{X}} \frac{1}{2}\left( \epsilon_x + \sum_{y\in\set{Y}} W_{\rv{Y}|\rv{X}}(y|x) \cdot \mathbbm{1}\left\{\log{\frac{W_{\rv{Y}|\rv{X}}(y|x)}{q_\rv{Y}(y)}}>a\right\} \right) \\
&= \max_{x\in\set{X}} \epsilon_x \leq \epsilon.
\end{align}
Therefore, 
\begin{equation}\label{eq:supDs:maxDmax:classical:2}
\max_{x\in\set{X}} D_{\max}\infdiv*{\hat{W}_{\rv{Y}|\rv{X}}(\cdot|x)}{q_\rv{Y}} \geq \inf_{\tilde{W}_{\rv{Y}|\rv{X}}:\norm{\tilde{W}_{\rv{Y}|\rv{X}}-W_{\rv{Y}|\rv{X}}}_\fnc{tvd}} \max_{x\in\set{X}} D_{\max}\infdiv*{\tilde{W}_{\rv{Y}|\rv{X}}(\cdot|x)}{q_\rv{Y}}
\end{equation}
Finally, \eqref{eq:supDs:maxDmax:classical} is proven by combining~\eqref{eq:supDs:maxDmax:classical:1} and~\eqref{eq:supDs:maxDmax:classical:2}.
\end{proof}
\begin{proof}[Proof of Theorem~\ref{thm:cs:cc}]
To prove the first inequality, we have
\begin{align}
\log{M_\epsilon^\star(W_{\rv{Y}|\rv{X}})}
&\geq I_{\max}^\epsilon(W_{\rv{Y}|\rv{X}})
= \multiadjustlimits{ 
    \inf_{\tilde{W}_{\rv{Y}|\rv{X}}: \norm{\tilde{W}_{\rv{Y}|\rv{X}}-W_{\rv{Y}|\rv{X}}}_\fnc{tvd}\leq\epsilon},
    \sup_{p_\rv{X}},
    \inf_{q_\rv{Y}}} D_{\max}\infdiv*{p_\rv{X}\cdot \tilde{W}_{\rv{Y}|\rv{X}}}{p_\rv{X} \times q_\rv{Y}} \\
\label{eq:cs:cc:1}
&\geq \multiadjustlimits{
    \sup_{p_\rv{X}\in\set{P}(\set{X})},
    \inf_{q_\rv{Y}\in\set{P}(\set{Y})},
    \inf_{\tilde{W}_{\rv{Y}|\rv{X}}: \norm{p_\rv{X} \cdot \tilde{W}_{\rv{Y}|\rv{X}}-p_\rv{X} \cdot W_{\rv{Y}|\rv{X}}}_\fnc{tvd}\leq\epsilon}} D_{\max}\infdiv*{p_\rv{X}\cdot \tilde{W}_{\rv{Y}|\rv{X}}}{p_\rv{X} \times q_\rv{Y}} \\
\label{eq:cs:cc:2}
&\geq \multiadjustlimits{
    \sup_{p_\rv{X}\in\set{P}(\set{X})},
    \inf_{q_\rv{Y}\in\set{P}(\set{Y})},
    \inf_{\tilde{p}_{\rv{XY}}\in\set{P}(\set{X}\times\set{Y}): \norm{\tilde{p}_\rv{XY} - p_\rv{X}\cdot W_{\rv{Y}|\rv{X}}}_\fnc{tvd}\leq\epsilon}} D_{\max}\infdiv*{\tilde{p}_\rv{XY}}{p_\rv{X} \times q_\rv{Y}} \\
\label{eq:cs:cc:3}
&\geq \adjustlimits\sup_{p_\rv{X}\in\set{P}(\set{X})} \inf_{q_\rv{Y}\in\set{P}(\set{Y})} \log{\frac{1}{\beta_{1-\epsilon-\delta}^\star\infdiv*{p_\rv{X}\cdot W_{\rv{Y}|\rv{X}}}{p_\rv{X}\cdot q_\rv{Y}}}} + \log{\delta} && \forall \delta\in(0,1-\epsilon)\\
\label{eq:cs:cc:4}
&\geq \log{N_{1-\epsilon-\delta}(W_{\rv{Y}|\rv{X}})} + \log{\delta} 
\end{align}
where, in~\eqref{eq:cs:cc:1}, we swap $\inf_{\tilde{W}_{\rv{Y}|\rv{X}}}$ inside and relax the domain for $\tilde{W}_{\rv{Y}|\rv{X}}$; 
in~\eqref{eq:cs:cc:2}, we relax the domain of the infimum;
in~\eqref{eq:cs:cc:3}, we use Lemma~\ref{lem:Dmax:Dh}; 
and in~\eqref{eq:cs:cc:4}, we use Theorem~27 from~\cite{polyanskiy2010channel}.

To prove the second inequality, we have
\begin{align}
\log{M_\epsilon^\star(W_{\rv{Y}|\rv{X}})}
&\leq I_{\max}^{\epsilon-\delta_1}(W_{\rv{Y}|\rv{X}}) + \log{\log{\frac{1}{\delta_1}}} && \forall \delta_1\in(0,\epsilon_1)\\
&= \multiadjustlimits{
    \inf_{q_\rv{Y}\in\set{P}(\set{Y})},
    \inf_{\tilde{W}_{\rv{Y}|\rv{X}}: \norm{\tilde{W}_{\rv{Y}|\rv{X}}-W_{\rv{Y}|\rv{X}}}_\fnc{tvd}\leq\epsilon-\delta_1},
    \max_{x\in\set{X}}} D_{\max}\infdiv*{\tilde{W}_{\rv{Y}|\rv{X}}(\cdot|x)}{q_\rv{Y}} + \log{\log{\frac{1}{\delta_1}}} \\
\label{eq:cs:cc:5}
&\leq \multiadjustlimits{
    \inf_{q_\rv{Y}\in\set{P}(\set{Y})},
    \sup_{p_\rv{X}\in\set{P}(\set{X})}} D_{s+}^{\epsilon-\delta_1}\infdiv*{p_\rv{X}\cdot W_{\rv{Y}|\rv{X}}}{p_\rv{X} \times q_\rv{Y}} + \log{\log{\frac{1}{\delta_1^2}}} \\
\label{eq:cs:cc:6}
&\leq \multiadjustlimits{
    \inf_{q_\rv{Y}\in\set{P}(\set{Y})},
    \sup_{p_\rv{X}\in\set{P}(\set{X})}} \log{\frac{1}{\beta_{1-\epsilon+\delta_1}^\star\infdiv*{p_\rv{X}\cdot W_{\rv{Y}|\rv{X}}}{p_\rv{X} \times q_\rv{Y}}}} + \log{\log{\frac{1}{\delta_1^2}}} \\
\label{eq:cs:cc:7}
&= \adjustlimits \sup_{p_\rv{X}\in\set{P}(\set{X})} \inf_{q_\rv{Y}\in\set{P}(\set{Y})} \log{\frac{1}{\beta_{1-\epsilon+\delta_1}^\star\infdiv*{p_\rv{X}\cdot W_{\rv{Y}|\rv{X}}}{p_\rv{X} \times q_\rv{Y}}}} + \log{\log{\frac{1}{\delta_1^2}}} \\
\label{eq:cs:cc:8}
&\leq \sup_{p_\rv{X}\in\set{P}(\set{X})} \log{\frac{1}{\beta_{1-\epsilon+\delta_1}^\star\infdiv*{p_\rv{X}\cdot W_{\rv{Y}|\rv{X}}}{p_\rv{X} \times p_\rv{Y}}}} + \log{\log{\frac{1}{\delta_1^2}}} \\
\label{eq:cs:cc:9}
&\leq \log{N^\star_{1-\epsilon+\delta_1+\delta_2}}(W_{\rv{Y}|\rv{X}}) + \log{\frac{1}{\delta_2}} + \log{\log{\frac{1}{\delta_1^2}}} && \forall \delta_2\in(0,\epsilon-\delta_1) \\
\label{eq:cs:cc:10}
&= \log{N^\star_{1-\epsilon+\delta}}(W_{\rv{Y}|\rv{X}}) + \log{\frac{2}{\delta}} + \log{\log{\frac{4}{\delta^2}}} && \forall \delta\in(0,\epsilon)
\end{align}
where, in~\eqref{eq:cs:cc:5}, we use Lemma~\ref{lem:supDs:maxDmax};
in~\eqref{eq:cs:cc:6}, we use Lemma~12 from~\cite{tomamichel2013hierarchy};
in~\eqref{eq:cs:cc:7}, we use the saddle point property of $\beta_\epsilon^\star\infdiv{p_\rv{X}\cdot W_{\rv{Y}|\rv{X}}}{p_\rv{X}\times q_\rv{Y}}$ (see~\cite{polyanskiy2012saddle});
in~\eqref{eq:cs:cc:9}, we use the following achievability bound~\cite[Eq.~(127)]{polyanskiy2010channel}
\begin{equation}
N_{\epsilon+\delta}^\star(W_{\rv{Y}|\rv{X}}) \geq  \sup_{p_\rv{X}\in\set{P}(\set{X})}  \frac{\delta}{\beta^\star_\epsilon(p_\rv{X}\cdot W_{\rv{Y}|\rv{X}},p_\rv{X}\times p_\rv{Y})} \quad\forall \delta\in(0,1-\epsilon);
\end{equation}
and in~\eqref{eq:cs:cc:10}, we substitute $\delta_1=\delta_2\gets\delta/2$.
\end{proof}

Theorem~{\ref{thm:cs:cc}} provides a direct relationship between the minimal cost $M_\epsilon^\star(W_{\rv{Y}|\rv{X}})$ of the $\epsilon$-tolerance channel simulation and the maximal code size $N_{1-\epsilon}^\star(W_{\rv{Y}|\rv{X}})$ for $(1-\epsilon)$-error channel coding.
Note that the asymptotic second-order expansion for channel coding has been well studied \cite{strassen1962asymptotische, polyanskiy2010channel, hayashi2009information}, \ie, 
\begin{equation}\label{eq:cc:expand}
\log{N_\epsilon^\star(W_{\rv{Y}|\rv{X}}^{\tensor n})} = n\cdot C(W_{\rv{Y}|\rv{X}}) + \sqrt{n V_\epsilon(W_{\rv{Y}|\rv{X}})}\cdot \Phi^{-1}(\epsilon) + O(\log{n}).
\end{equation}
Combing the above expansion with Theorem~\ref{thm:cs:cc}, we have the following corollary.
\begin{corollary}\label{cor:cs:expand}
Let $W_{\rv{Y}|\rv{X}}$ be a channel and $\epsilon\in(0,1)$.
For each $n \in \mathbb{N}$, it holds that
\begin{equation}\label{eq:cs:expand}
\log{M_\epsilon^\star(W_{\rv{Y}|\rv{X}}^{\tensor n})} = n\cdot C(W_{\rv{Y}|\rv{X}}) + \sqrt{n V_{1-\epsilon}(W_{\rv{Y}|\rv{X}})}\cdot \Phi^{-1}(1-\epsilon) + O(\log{n}) .
\end{equation}
\end{corollary}
\begin{proof}
This is a direct result of Theorem~\ref{thm:cs:cc} with $\delta\gets\frac{1}{\sqrt{n}}$ and~\eqref{eq:cc:expand}, namely
\begin{align}
\log{M_\epsilon^\star(W_{\rv{Y}|\rv{X}}^{\tensor n})}
&\leq n\cdot C(W_{\rv{Y}|\rv{X}}) + \sqrt{n\cdot V_{1-\epsilon+\frac{1}{\sqrt{n}}}(W_{\rv{Y}|\rv{X}})}\cdot \Phi^{-1}(1-\epsilon+\frac{1}{\sqrt{n}}) + O(\log{n}) ,\\
\log{M_\epsilon^\star(W_{\rv{Y}|\rv{X}}^{\tensor n})}
&\geq n\cdot C(W_{\rv{Y}|\rv{X}}) + \sqrt{n\cdot V_{1-\epsilon-\frac{1}{\sqrt{n}}}(W_{\rv{Y}|\rv{X}})}\cdot \Phi^{-1}(1-\epsilon-\frac{1}{\sqrt{n}}) + O(\log{n}) .
\end{align}
Note that, for $n$ large enough
\begin{equation}
\abs{\Phi^{-1}(1-\epsilon\pm\frac{1}{\sqrt{n}}) - \Phi^{-1}(1-\epsilon)} \leq \frac{1}{\sqrt{n}}\cdot \max_{t\in[\frac{1-\epsilon}{2},1-\frac{\epsilon}{2}]} \frac{\D}{\D t}\Phi^{-1}(t) \leq \frac{1}{\sqrt{n}}\cdot \text{constant}.
\end{equation}
Thus, we have
\begin{align}
\log{M_\epsilon^\star(W_{\rv{Y}|\rv{X}}^{\tensor n})}
&\leq n\cdot C(W_{\rv{Y}|\rv{X}}) + \sqrt{n\cdot V_{1-\epsilon+\frac{1}{\sqrt{n}}}(W_{\rv{Y}|\rv{X}})}\cdot \Phi^{-1}(1-\epsilon) + O(\log{n}) \\
&= n\cdot C(W_{\rv{Y}|\rv{X}}) + \sqrt{n\cdot V_{1-\epsilon}(W_{\rv{Y}|\rv{X}})}\cdot \Phi^{-1}(1-\epsilon) + O(\log{n}) \\
\log{M_\epsilon^\star(W_{\rv{Y}|\rv{X}}^{\tensor n})}
&\geq n\cdot C(W_{\rv{Y}|\rv{X}}) + \sqrt{n\cdot V_{1-\epsilon-\frac{1}{\sqrt{n}}}(W_{\rv{Y}|\rv{X}})}\cdot \Phi^{-1}(1-\epsilon) + O(\log{n}) \\
&= n\cdot C(W_{\rv{Y}|\rv{X}}) + \sqrt{n\cdot V_{1-\epsilon}(W_{\rv{Y}|\rv{X}})}\cdot \Phi^{-1}(1-\epsilon) + O(\log{n}) 
\end{align}
for $n$ large enough.
\end{proof}
One may notice the interesting resemblance between~\eqref{eq:cc:expand} and~\eqref{eq:cs:expand} since the only difference is the swap of $\epsilon$ by $1-\epsilon$.
Technically, this traces back to Lemmas~\ref{lem:Dmax:Dh} and~\ref{lem:supDs:maxDmax}, and subsequently to Theorem~\ref{thm:cs:cc}.
Since $\Phi^{-1}(\epsilon)+\Phi^{-1}(1-\epsilon)=0$, we see from~\eqref{eq:cc:expand} and~\eqref{eq:cs:expand} that the finite-blocklength coding rate and the simulation rate approach the capacity from the two opposite directions.
This fits the intuition, since as $n$ becomes larger, both tasks become easier: For channel coding, this means higher coding rates; for channel simulation, this means that a smaller number of resource channels is needed per channel simulated.

\subsection{Moderate-Deviation Regime}\label{sec:moderate-deviation(p2p)}
For the moderate-deviation analysis, we bound the $n$-fold channel simulation cost using the information spectrum divergence from both sides, and analyze the latter under the moderate-deviation regime.
Note that the moderate-deviation analysis is a long established method in statistics (\eg, see~\cite{dembo2009large} as a textbook and~\cite{feller1943generalization} as an early study).
In this section, we use a specific result~\cite[Lemma~1 and~2]{chubb2017moderate} (see Lemma~\ref{lem:MDUL}) for our analysis.
\subsubsection{Information Spectrum divergence bounds}
We prove the following upper and lower bounds of the $n$-fold channel simulation cost in terms of the information spectrum divergence.
\begin{proposition}\label{prop:Ds:bounds}
Given any channel $W_{\rv{Y}|\rv{X}}\in\set{P}(\set{Y}|\set{X})$ and $\epsilon\in(0,1)$, for any $n\in\mathbb{N}$, it holds that
\begin{align}
\log{M_\epsilon^\star(W_{\rv{Y}|\rv{X}}^{\tensor n})} &\leq \adjustlimits\inf_{q_{\rvs{Y}_1^n}}\sup_{\mathbf{x}_1^n} D_{s+}^{\epsilon-\delta}\infdiv*{W_{\rv{Y}|\rv{X}}^{\tensor n}(\cdot|\mathbf{x}_1^n)}{q_{\rvs{Y}_1^n}} + \log\log{\frac{1}{\delta^2}} && \forall \delta\in(0,\epsilon), \label{eq:achievability:Ds}\\
\log{M_\epsilon^\star(W_{\rv{Y}|\rv{X}}^{\tensor n})} &\geq \sup_{p_{\rvs{X}_1^n}\in\set{P}(\set{X}^n)} D_{s+}^{\epsilon+2\delta}\infdiv*{p_{\rvs{X}_1^n}\cdot W^{\tensor n}_{\rv{Y}|\rv{X}}}{p_{\rvs{X}_1^n}\times p_{\rvs{Y}_1^n}} + \log{\delta^2}  && \forall \delta\in(0,\frac{1-\epsilon}{2}), \label{eq:converse:Ds}
\end{align}
where $p_{\rvs{Y}_1^n}$ is the output distribution corresponding to $p_{\rvs{X}_1^n}$ under the channel $W_{\rv{Y}|\rv{X}}$.
\end{proposition}
We use the quasi-convexity of the information spectrum divergence (see Lemma~\ref{lem:quasi:convex:Ds}) for the proof of~\eqref{eq:achievability:Ds}.
For the proof of Lemma~\ref{lem:quasi:convex:Ds}, please refer to Appendix~\ref{app:proof:lem:quasi:convex:Ds}.
\begin{lemma}[Quasi-Convexity of the Information Spectrum Divergence]\label{lem:quasi:convex:Ds}
Let $p_{\rv{Y}|\rv{X}}, q_{\rv{Y}|\rv{X}}\in\set{P}(\set{Y}|\set{X})$, and $p_\rv{X}\in\set{P}(\set{X})$.
For all $\epsilon\in(0,1)$, it holds that
\begin{equation}\label{eq:quasi:convex:Ds}
D_{s+}^\epsilon\infdiv*{p_{\rv{Y}|\rv{X}}\cdot p_\rv{X}}{q_{\rv{Y}|\rv{X}}\cdot p_\rv{X}} \leq \sup_{x\in\set{X}} D_{s+}^\epsilon\infdiv*{p_{\rv{Y}|\rv{X}}(\cdot|x)}{q_{\rv{Y}|\rv{X}}(\cdot|x)}.
\end{equation}
\end{lemma}
\begin{proof}[Proof of~\eqref{eq:achievability:Ds}]
By Theorem~\ref{thm:one-shot:achievability}, we have (for any $\delta\in(0,\epsilon)$)
\begin{equation}
\log{M_\epsilon^\star(W_{\rv{Y}|\rv{X}}^{\tensor n})} \leq \multiadjustlimits{
    \inf_{q_\rv{Y}}, 
    \inf_{\tilde{W}_{\rv{Y}|\rv{X}}: \norm{\tilde{W}_{\rv{Y}|\rv{X}}-W_{\rv{Y}|\rv{X}}}_\fnc{tvd}\leq\epsilon}, 
    \max_{x}
    } D_{\max}\infdiv*{\tilde{W}_{\rv{Y}|\rv{X}}(\cdot|x)}{q_\rv{Y}} + \log{\log{\frac{1}{\delta}}}.
\end{equation}  
By Lemma~\ref{lem:supDs:maxDmax} and Lemma~\ref{lem:quasi:convex:Ds}, we further have
\begin{align}
\log{M_\epsilon^\star(W_{\rv{Y}|\rv{X}}^{\tensor n})} &\leq \adjustlimits\inf_{q_{\rvs{Y}_1^n}} \sup_{p_{\rvs{X}_1^n}} D_{s+}^{\epsilon-\delta}\infdiv*{p_{\rvs{X}_1^n}\cdot W_{\rv{Y}|\rv{X}}^{\tensor n}}{p_{\rvs{X}_1^n}\times q_{\rvs{Y}_1^n}} + \log\log{\frac{1}{\delta^2}} \\
&\leq \adjustlimits\inf_{q_{\rvs{Y}_1^n}} \sup_{\mathbf{x}_1^n} D_{s+}^{\epsilon-\delta}\infdiv*{W_{\rv{Y}|\rv{X}}^{\tensor n}(\cdot|\mathbf{x}_1^n)}{q_{\rvs{Y}_1^n}} + \log\log{\frac{1}{\delta^2}}. \qedhere
\end{align}
\end{proof}
We need the following two lemmas for the proof of~\eqref{eq:converse:Ds}.
The proof of Lemma~\ref{lem:maxDmax:supDs} is partially inspired by ~\cite[Eq.~(48)]{anshu2020partially} and~\cite[Eq.~(17)]{tomamichel2013hierarchy}, and is deferred to Appendix~\ref{app:proof:lem:maxDmax:supDs}.
Lemma~\ref{lem:Ds} can be proven following the same arguments of~\cite[Lem.~10]{anshu2020partially} with the only difference being the replacement of $D_s$ by $D_{s+}$.
\begin{lemma}\label{lem:maxDmax:supDs}
Let $W_{\rv{Y}|\rv{X}}$ be a channel from $\set{X}$ to $\set{Y}$ where both $\set{X}$ and $\set{Y}$ are finite sets.
It holds for all $\epsilon\in(0,1)$ and $\delta\in(0,1-\epsilon)$ that
\begin{equation}\label{eq:maxDmax:supDs}
\adjustlimits \inf_{\tilde{W}_{\rv{Y}|\rv{X}}:\norm{\tilde{W}_{\rv{Y}|\rv{X}}-W_{\rv{Y}|\rv{X}}}_\fnc{tvd}\leq\epsilon} \max_{x\in\set{X}} D_{\max}\infdiv*{\tilde{W}_{\rv{Y}|\rv{X}}(\cdot|x)}{q_\rv{Y}} \geq \sup_{p_\rv{X}\in\set{P}(\set{X})} D_{s+}^{\epsilon+\delta}\infdiv*{p_\rv{X}\cdot W_{\rv{Y}|\rv{X}}}{p_\rv{X}\times q_\rv{Y}} + \log{\delta}
\end{equation}
for any $q_\rv{Y}\in\set{P}(\set{Y})$.
\end{lemma}
\begin{lemma}\label{lem:Ds}
For any $p_\rv{XY}\in\set{P}(\set{X}\times\set{Y})$, $\epsilon\in(0,1)$, and $\delta\in(0,1-\epsilon)$, it holds that 
\begin{equation}
\inf_{q_\rv{Y}\in\set{P}(\set{Y})} D_{s+}^{\epsilon}\infdiv*{p_\rv{XY}}{p_\rv{X}\times q_\rv{Y}} \geq D_{s+}^{\epsilon+\delta}\infdiv*{p_\rv{XY}}{p_\rv{X}\times p_\rv{Y}} + \log{\delta}.
\end{equation}
\end{lemma}
\begin{proof}[Proof of~\eqref{eq:converse:Ds}]
This proposition is a direct result of Theorem~\ref{thm:one-shot:converse} and the above two lemmas, namely
\begin{align}
\log{M_\epsilon^\star(W_{\rv{Y}|\rv{X}}^{\tensor n})}
& \geq \multiadjustlimits{
    \inf_{q_{\rvs{Y}_1^n}}, 
    \inf_{\tilde{W}_{\rvs{Y}_1^n|\rvs{X}_1^n}: \norm{\tilde{W}_{\rvs{Y}_1^n|\rvs{X}_1^n}-W^{\tensor n}_{\rv{Y}|\rv{X}}}_\fnc{tvd}\leq\epsilon}, 
    \max_{x}
    } D_{\max}\infdiv*{\tilde{W}_{\rvs{Y}_1^n|\rvs{X}_1^n}(\cdot|\mathbf{x}_1^n)}{q_{\rvs{Y}_1^n}} \hspace{-15pt}\\
& \geq \adjustlimits \inf_{q_{\rvs{Y}_1^n}} \sup_{p_{\rvs{X}_1^n}\in\set{P}(\set{X}^n)} D_{s+}^{\epsilon+\delta}\infdiv*{p_{\rvs{X}_1^n}\cdot W^{\tensor n}_{\rv{Y}|\rv{X}}}{p_{\rvs{X}_1^n}\times q_{\rvs{Y}_1^n}} + \log{\delta} && \text{by Lemma~\ref{lem:maxDmax:supDs}, for all }\delta\in(0,1-\epsilon) \\
& \geq \sup_{p_{\rvs{X}_1^n}\in\set{P}(\set{X}^n)} D_{s+}^{\epsilon+2\delta}\infdiv*{p_{\rvs{X}_1^n}\cdot W^{\tensor n}_{\rv{Y}|\rv{X}}}{p_{\rvs{X}_1^n}\times p_{\rvs{Y}_1^n}} + \log{\delta^2} && \text{by Lemma~\ref{lem:Ds}, for all }\delta\in(0,\frac{1-\epsilon}{2}). \qedhere
\end{align}
\end{proof}
\subsubsection{Asymptotic analysis of the information spectrum divergence}
Now, we consider the asymptotic expansion of the information spectrum divergence under the moderate-deviation regime.
We start with the following lemma from~\cite{chubb2017moderate}.
\begin{lemma}[{\cite[Lemma~1 and~2]{chubb2017moderate}}] \label{lem:MDUL}
Let $\{t_n\}_{n\in\mathbb{N}}$ be a moderate sequence, \ie, $\{t_n\}_n$ is positive, non-increasing, tending to $0$, but $\{\sqrt{n}t_n\}$ tends to (positive) infinity.
For each $n\in\mathbb{N}$, let $\{\rv{X}_{i,n}\}_{i=1}^n$ be independent $0$-mean random variables.
Let $V_n\defeq \frac{1}{n} \sum_{i=1}^n \fnc{Var}(\rv{X}_{i,n})$ denote the average variance of $\{\rv{X}_{i,n}\}_{i=1}^{n}$ for each $n$.
Let $T_n\defeq \frac{1}{n} \sum_{i=1}^n \mathbb{E}[\abs{\rv{X}_{i,n}}^3]$ denote the average third absolute moment of $\{\rv{X}_{i,n}\}_{i=1}^{n}$ for each $n$.
\begin{itemize}
    \item If $\{V_n\}_{n\in\mathbb{N}}$ is bounded away from $0$, and $\{T_n\}_{n\in\mathbb{N}}$ is bounded, then for any $\eta>0$ there exists some positive integer $N_\eta$ such that 
    \begin{equation}\label{eq:MDLB}
    \log{\Pr\left[\frac{1}{n}\sum_{i=1}^n \rv{X}_{i,n} \geq t_n\right]} \geq -(1+\eta) \frac{nt_n^2}{2V_n} 
    \end{equation}
    for all $n\geq N_\eta$.
    \item If the following sequence
    \begin{equation}\label{eq:MDUB:condition}
    n \mapsto \frac{1}{n}\sum_{i=1}^n \sup_{s\in[0,\frac{1}{2}]} \abs{\frac{\D^3}{\D{s^3}}\log{\mathbb{E}}\left[e^{s\rv{X}_{i,n}}\right]} 
    \end{equation}
    is bounded, then for any $\eta>0$ there exists some positive integer $N'_\eta$ such that 
    \begin{equation}\label{eq:MDUB}
     \log{\Pr\left[\frac{1}{n}\sum_{i=1}^n \rv{X}_{i,n} \geq t_n\right]} \leq - \frac{nt_n^2}{2V_n+\eta}
    \end{equation}
    for all $n\geq N'_\eta$.
\end{itemize}
\end{lemma}
Using the above lemma, we have the following upper and lower bounds of $\frac{1}{n} D_{s+}^{\epsilon_n}\infdiv*{q_{x_1}\times\cdots\times q_{x_n}}{q^{\tensor n}}$.
\begin{proposition}
Let $\{q_x\}_{x\in\set{X}}$, $q\in\set{P}(\set{Y})$ where $\set{X}$ and $\set{Y}$ are finite sets, and assume $\supp(q_x)\subset\supp(q)$ for all $x\in\set{X}$.
Let $\mathbf{x}_1^n\in\set{X}^n$, and define
\begin{equation}
D_n \defeq \frac{1}{n} \sum_{k=1}^n D\infdiv*{q_{x_k}}{q},\quad
V_n \defeq \frac{1}{n} \sum_{k=1}^n V\infdiv*{q_{x_k}}{q}.
\end{equation}
Let $\{a_n\}$ be a moderate sequence, and let $\epsilon_n\defeq 2^{-na_n^2}$ for each $n\in\mathbb{N}$.
It holds that
\begin{equation}\label{eq:MD:upper:Ds}\begin{aligned}
\frac{1}{n} D_{s+}^{\epsilon_n}\infdiv*{q_{x_1}\times\cdots\times q_{x_n}}{q^{\tensor n}} \leq D_n + \sqrt{2V_n} a_n + o(a_n).
\end{aligned}\end{equation}
Furthermore, if there does \emph{not} exist any $x\in\set{X}$ such that $q_x\propto q$ on $\supp(q_x)$, it holds that
\begin{equation}\label{eq:MD:lower:Ds}\begin{aligned}
\frac{1}{n} D_{s+}^{\epsilon_n}\infdiv*{q_{x_1}\times\cdots\times q_{x_n}}{q^{\tensor n}} \geq D_n + \sqrt{2V_n}a_n + o(a_n).
\end{aligned}\end{equation}
\end{proposition}
\begin{proof}
To prove~\eqref{eq:MD:upper:Ds}, we first rewrite the LHS of it as
\begin{align}
\frac{1}{n} D_{s+}^{\epsilon_n}\infdiv*{q_{x_1}\times\cdots\times q_{x_n}}{q^{\tensor n}} 
&= D_n + \inf\left\{a\geq -D_n \middle\vert \Pr\left[\frac{1}{n}\sum_{i=1}^n \log{\frac{q_{x_i}(\rv{Y}_{x_i})}{q(\rv{Y}_{x_i})}} - D\infdiv*{q_{x_i}}{q}>a\right]<\epsilon_n\right\}\\
&\leq  D_n + \inf\left\{a\geq 0 \middle\vert \Pr\left[\frac{1}{n}\sum_{i=1}^n \log{\frac{q_{x_i}(\rv{Y}_{x_i})}{q(\rv{Y}_{x_i})}} - D\infdiv*{q_{x_i}}{q}>a\right]<\epsilon_n\right\}.
\end{align}
For each $\eta>0$, let $t_n\defeq (\sqrt{2V_n}+\eta)\cdot a_n >0$. 
(One can quickly verify $\{t_n\}_{n=1}^\infty$ to be a moderate sequence since $V_n$ is bounded.)
Using Lemma~\ref{lem:MDUL} (and~\cite[Lemma~7]{chubb2017moderate}), for any $\eta>0$, independent of $n$, we have 
\begin{equation}\label{eq:MD:Ds:applied:MDUB}
\log{\Pr\left[\frac{1}{n}\sum_{i=1}^n \log{\frac{q_{x_i}(\rv{Y}_{x_i})}{q(\rv{Y}_{x_i})}} - D\infdiv*{q_{x_i}}{q} > t_n\right]}
\leq - \frac{nt_n^2}{2V_n + \eta^2/2}
= - \frac{(\sqrt{2V_n}+\eta)^2\cdot a_n^2\cdot n}{2V_n + \eta^2/2}
< -a_n^2 \cdot n =\log{\epsilon_n}
\end{equation}
for all $n\geq N_{\eta}$, for some $N_\eta\in\mathbb{N}$.
Thus, 
\begin{equation}
\frac{1}{n} D_{s+}^{\epsilon_n}\infdiv*{q_{x_1}\times\cdots\times q_{x_n}}{q^{\tensor n}} \leq D_n + t_n = D_n + \sqrt{2V_n}a_n + \eta\cdot a_n
\quad \forall n\geq N_\eta.
\end{equation}
Note that $N_\eta$ is independent of the sequence $\mathbf{x}_1^n$, since one can obtain a universal bound for~\eqref{eq:MDUB:condition} for all sequences $\mathbf{x}_1^n$.
On the other hand, $N_\eta$ does depend on the sequence $\{a_\ell\}_{\ell=1}^\infty$ and the distributions $\{q_x\}_{x\in\set{X}}$ and $q$ which, however, are considered to be fixed in our discussion.
Finally, the statement can be shown by letting $\eta$ tend to $0$.
\par
We now move on to prove~\eqref{eq:MD:lower:Ds}.
Note that we have the following three assumptions:
\begin{itemize}
    \item $q_x\not\propto q$ on $\supp(q_x)$ for all $x\in\set{X}$;
    \item $\supp(q_x)\subset\supp(q)$ for all $x\in\set{X}$;
    \item $\set{X}$ is a finite set.
\end{itemize}
It is straightforward to see that $\{V_n\}_n$ is bounded away from $0$, and the bound is independent of $\mathbf{x}_1^n$.
In particular, $V_n \geq V_{\min}\defeq \min_{x\in\set{X}} V\infdiv*{q_x}{q}>0$, where the last inequality is guaranteed by the assumption that $q_x\not\propto q$ on $\supp(q_x)$ for all $x$.
Similarly, so is the sequence $\{T_n\defeq\frac{1}{n}\sum_{k=1}^n T\infdiv*{q_{x_k}}{q}\}_n$.
By the finiteness of $\set{X}$, one can also find some $V_+>0$ such that $V_n\leq V^+$ for all $n$.

Now, for any fixed $\eta\in(0,\sqrt{2V_{\min}})$, let $t_n \defeq (\sqrt{2V_n}-\eta)\cdot a_n$.
By Lemma~\ref{lem:MDUL}, we have
\begin{align} 
    \log{\Pr\left[\frac{1}{n}\sum_{i=1}^n \log{\frac{q_{x_i}(\rv{Y}_{x_i})}{q(\rv{Y}_{x_i})}} - D\infdiv*{q_{x_i}}{q} > t_n\right]}
    &\geq \log{\Pr\left[\frac{1}{n}\sum_{i=1}^n \log{\frac{q_{x_i}(\rv{Y}_{x_i})}{q(\rv{Y}_{x_i})}} - D\infdiv*{q_{x_i}}{q} \geq (\sqrt{2V_n}-\frac{\eta}{2})\cdot a_n\right]}\nonumber \\
    \label{eq:MD:Ds:applied:MDLB}
    &\geq -(1+\frac{\eta}{2\sqrt{2V^+}})\cdot \frac{(\sqrt{2V_n}-\frac{\eta}{2})^2}{2V_n}\cdot na_n^2 \\
    \label{eq:MD:Ds:applied:MDLB:2}
    &\geq -na_n^2 = \log{\epsilon_n}.
\end{align}
for all $n\geq N_\eta$, for some $N_\eta\in\mathbb{N}$.
Thus,
\begin{align}
\frac{1}{n} D_{s+}^{\epsilon_n}\infdiv*{q_{x_1}\times\cdots\times q_{x_n}}{q^{\tensor n}} 
&= D_n + \inf\left\{a\geq -D_n \middle\vert \Pr\left[\frac{1}{n}\sum_{i=1}^n \log{\frac{q_{x_i}(\rv{Y}_{x_i})}{q(\rv{Y}_{x_i})}} - D\infdiv*{q_{x_i}}{q}>a\right]<\epsilon_n\right\}\\
&\geq D_n + t_n = D_n + \sqrt{2V_n} a_n -\eta a_n.
\end{align}
Again, note that $N_\eta$ is independent of the sequence $\mathbf{x}_1^n$, since $V_{\min}$ is independent from $\mathbf{x}_1^n$.
Finally, again, the statement can be shown by letting $\eta$ tend to $0$.
\end{proof}
\begin{remark}\label{rmk:MD:upper:Ds}
One can rewrite~\eqref{eq:MD:Ds:applied:MDUB} into a chain of tighter inequalities as
\begin{equation}
    \text{LHS of~\eqref{eq:MD:Ds:applied:MDUB}} \leq -\left(1+\frac{\eta^2/2}{2V_n+\eta^2/2}\right) \cdot a_n^2 \cdot n
    < -na_n^2 -\log{\kappa} = \log{\frac{\epsilon_n}{\kappa}}
\end{equation}
for any constant $\kappa\geq 1$, for $n$ large enough, since $na_n^2\to\infty$ and $\frac{\eta^2/2}{2V_n+\eta^2/2}>0$.
In other words, it holds that
\begin{equation}
\frac{1}{n} D_{s+}^{\epsilon_n/\kappa}\infdiv*{q_{x_1}\times\cdots\times q_{x_n}}{q^{\tensor n}} \leq D_n+\sqrt{2V_n} a_n + o(a_n).
\end{equation}
\end{remark}
\begin{remark}\label{rmk:MD:lower:Ds}
One could rewrite~\eqref{eq:MD:Ds:applied:MDLB:2} into a tighter inequality as
\begin{equation}
\text{\eqref{eq:MD:Ds:applied:MDLB}} \geq \left(-1+\frac{\eta}{2\sqrt{2V_n}}\right) na_n^2 \geq -na_n^2 +\log{\kappa} = \log{\kappa\epsilon_n}
\end{equation}
for any constant $\kappa\geq 1$, for $n$ large enough, since $na_n^2\to\infty$.
In other words, it holds for any $\kappa\geq 1$ that
\begin{equation}
\frac{1}{n} D_{s+}^{\kappa\epsilon_n}\infdiv*{q_{x_1}\times\cdots\times q_{x_n}}{q^{\tensor n}} \geq D_n+\sqrt{2V_n} a_n + o(a_n).
\end{equation}
\end{remark}
\subsubsection{Combining 1) and 2)}
We now combine 1) and 2) to derive the following expansion of the channel simulation rate under the moderate-deviation regime.
\begin{theorem}\label{thm:cs:expand:moderate}
Let $W_{\rv{Y}|\rv{X}}$ be a channel, and let $\{a_n\}_{n\in\mathbb{N}}$ be a moderate sequence.
Let $\epsilon_n\defeq 2^{-a_n^2 n}$.
It holds that
\begin{equation}
\frac{1}{n} \log{M_{\epsilon_n}^\star(W_{\rv{Y}|\rv{X}}^{\tensor n})} = C(W_{\rv{Y}|\rv{X}}) + \sqrt{2V_{\max}(W_{\rv{Y}|\rv{X}})}\cdot a_n + o(a_n).
\end{equation}
\end{theorem}
We prove the above theorem by proving the following three statements.
\begin{align}
\label{eq:MD:achievability}
\frac{1}{n}\log{M_{\epsilon_n}^\star(W_{\rv{Y}|\rv{X}}^{\tensor n})} &\leq C(W_{\rv{Y}|\rv{X}}) + \sqrt{2V_{\max}(W_{\rv{Y}|\rv{X}})}\cdot a_n + o(a_n), \\
\label{eq:MD:converse}
\text{if }V_{\max}(W_{\rv{Y}|\rv{X}})>0\text{, then }\frac{1}{n}\log{M_{\epsilon_n}^\star(W_{\rv{Y}|\rv{X}}^{\tensor n})} &\geq C(W_{\rv{Y}|\rv{X}}) + \sqrt{2V_{\max}(W_{\rv{Y}|\rv{X}})}\cdot a_n + o(a_n), \\
\label{eq:MD:converse:ve=0}
\text{if }V_{\max}(W_{\rv{Y}|\rv{X}})=0\text{, then }\frac{1}{n}\log{M_{\epsilon_n}^\star(W_{\rv{Y}|\rv{X}}^{\tensor n})} &\geq C(W_{\rv{Y}|\rv{X}}) + o(a_n).
\end{align}
\begin{proof}[Proof of~\eqref{eq:MD:achievability}]
We start with~\eqref{eq:achievability:Ds} with $\delta\gets \frac{\epsilon_n}{2}$, \ie, 
\begin{align}
\frac{1}{n}\log{M_{\epsilon_n}^\star(W_{\rv{Y}|\rv{X}}^{\tensor n})} &\leq \adjustlimits\inf_{q_{\rvs{Y}_1^n}}\sup_{\mathbf{x}_1^n} \frac{1}{n} D_{s+}^{\epsilon_n/2}\infdiv*{W_{\rv{Y}|\rv{X}}^{\tensor n}(\cdot|\mathbf{x}_1^n)}{q_{\rvs{Y}_1^n}} + \frac{1}{n}\log\log{\frac{4}{\epsilon_n^2}} \\
&= \adjustlimits\inf_{q_{\rvs{Y}_1^n}}\sup_{\mathbf{x}_1^n} \frac{1}{n} D_{s+}^{\epsilon_n/2}\infdiv*{W_{\rv{Y}|\rv{X}}^{\tensor n}(\cdot|\mathbf{x}_1^n)}{q_{\rvs{Y}_1^n}} + \underbrace{\frac{1}{n} + \frac{\log(1+a_n^2 n)}{n}}_{=o(a_n)}
\end{align}
Restricting the infimum over $q_{\rvs{Y}_1^n}$ to positive product distributions, and applying Remark~\ref{rmk:MD:upper:Ds}, we have 
\begin{align}
\frac{1}{n}\log{M_{\epsilon_n}^\star(W_{\rv{Y}|\rv{X}}^{\tensor n})} &\leq \adjustlimits\inf_{q_{\rv{Y}}>0}\sup_{\mathbf{x}_1^n} \frac{1}{n} D_{s+}^{\epsilon_n/2}\infdiv*{W_{\rv{Y}|\rv{X}}^{\tensor n}(\cdot|\mathbf{x}_1^n)}{q_{\rv{Y}}^{\tensor n}} + o(a_n) \\
&\leq \adjustlimits\inf_{q_{\rv{Y}}>0}\sup_{\mathbf{x}_1^n} \frac{1}{n}\sum_{i=1}^n D\infdiv*{W_{\rv{Y}|\rv{X}}(\cdot|x_i)}{q_\rv{Y}} + \sqrt{2\cdot\frac{1}{n}\sum_{i=1}^n V\infdiv*{W_{\rv{Y}|\rv{X}}(\cdot|x_i)}{q_\rv{Y}}} \cdot a_n + o(a_n) \\
&\leq \adjustlimits\inf_{q_{\rv{Y}}>0}\sup_{p_\rv{X}} D\infdiv*{p_\rv{X}\cdot W_{\rv{Y}|\rv{X}}}{p_\rv{X}\times q_\rv{Y}} + \sqrt{2 V\infdiv*{p_\rv{X}\cdot W_{\rv{Y}|\rv{X}}}{p_\rv{X}\times q_\rv{Y}}} \cdot a_n + o(a_n)\\
&\leq \sup_{p_\rv{X}} D\infdiv*{p_\rv{X}\cdot W_{\rv{Y}|\rv{X}}}{p_\rv{X}\times q^\star_\rv{Y}} + \sqrt{2 V\infdiv*{p_\rv{X}\cdot W_{\rv{Y}|\rv{X}}}{p_\rv{X}\times q^\star_\rv{Y}}} \cdot a_n + o(a_n)
\end{align}
where $q^\star_\rv{Y}$ denotes the unique capacity-achieving output distribution, which must be a positive distribution (unless certain output is unattainable for all inputs, in which case we can simply consider a truncated output alphabet).
Notice that the first term (\ie, $D\infdiv*{p_\rv{X}\cdot W_{\rv{Y}|\rv{X}}}{p_\rv{X}\times q_\rv{Y}}$) is the dominating term, and achieves its maximum only when $p_\rv{X}$ is a capacity-achieving input distribution.
Thus,
\begin{equation}
\frac{1}{n}\log{M_{\epsilon_n}^\star(W_{\rv{Y}|\rv{X}}^{\tensor n})} \leq  C(W_{\rv{Y}|\rv{X}}) + \sup_{p_\rv{X}\in\Pi} \sqrt{2 V\infdiv*{p_\rv{X}\cdot W_{\rv{Y}|\rv{X}}}{p_\rv{X}\times q^\star_\rv{Y}}} \cdot a_n + o(a_n)
= C(W_{\rv{Y}|\rv{X}}) + \sqrt{2 V_{\max}} \cdot a_n + o(a_n)
\end{equation}
for $n$ large enough.
\end{proof}
\begin{proof}[Proof of~\eqref{eq:MD:converse}]
Using~\eqref{eq:converse:Ds} with $\delta\gets\epsilon_n$ we have
\begin{equation}
\frac{1}{n}\log{M_{\epsilon_n}^\star(W_{\rv{Y}|\rv{X}}^{\tensor n})} \geq \sup_{p_{\rvs{X}_1^n}} \frac{1}{n}D_{s+}^{3\epsilon_n}\infdiv*{p_{\rvs{X}_1^n}\cdot W_{\rv{Y}|\rv{X}}^{\tensor n}}{p_{\rvs{X}_1^n}\times p_{\rvs{Y}_1^n}} - 2a_n^2
\end{equation}
where $p_{\rvs{Y}_1^n}\defeq \sum_{\mathbf{x}_1^n} W_{\rv{Y}|\rv{X}}^{\tensor n}(\cdot|\mathbf{x}_1^n) \cdot p_{\rvs{X}_1^n}(\mathbf{x}_1^n)$.
Pick $p_{\rvs{X}_1^n}=(p^\star_\rv{X})^{\tensor n}$ where $p^\star_\rv{X}$ is a capacity-achieving input distribution that maximizes $V(p_\rv{X})$, and apply Remark~\ref{rmk:MD:lower:Ds}.
We have (note that $a_n^2\ll a_n$)
\begin{align}
\frac{1}{n}\log{M_{\epsilon_n}^\star(W_{\rv{Y}|\rv{X}}^{\tensor n})}  &\geq \frac{1}{n}D_{s+}^{3\epsilon_n}\infdiv*{\left(p^\star_\rv{X}\cdot W_{\rv{Y}|\rv{X}}\right)^{\tensor n}}{\left(p^\star_\rv{X}\times p^\star_\rv{Y}\right)^{\tensor n}} + o(a_n)\\
&\geq \frac{1}{n}\sum_{i=1}^n D\infdiv*{p^\star_\rv{X}\cdot W_{\rv{Y}|\rv{X}}}{p^\star_\rv{X}\times p^\star_\rv{Y}} + \sqrt{2\cdot \frac{1}{n} \sum_{i=1}^n V\infdiv*{p^\star_\rv{X}\cdot W_{\rv{Y}|\rv{X}}}{p^\star_\rv{X}\times p^\star_\rv{Y}}} \cdot a_n + o(a_n) \\
& = D\infdiv*{p^\star_\rv{X}\cdot W_{\rv{Y}|\rv{X}}}{p^\star_\rv{X}\times p^\star_\rv{Y}} + \sqrt{2\cdot V\infdiv*{p^\star_\rv{X}\cdot W_{\rv{Y}|\rv{X}}}{p^\star_\rv{X}\times p^\star_\rv{Y}}} \cdot a_n + o(a_n) \\
& = C(W_{\rv{Y}|\rv{X}}) + \sqrt{2 V_{\max}} \cdot a_n + o(a_n). \nonumber \qedhere
\end{align}
\end{proof}

\begin{proof}[Proof of~\eqref{eq:MD:converse:ve=0}]
This can be proven by using~\eqref{eq:converse:Ds} with $\delta\gets\epsilon_n$, together with the following Chebyshev-type bound (see Appendix~\ref{app:Chebyshev:lower:Ds})
\begin{equation}\label{eq:Chebyshev:lower:Ds}
D_{s+}^\epsilon\infdiv*{p^{\tensor}}{q^{\tensor n}} \geq nD\infdiv*{p}{q} - \sqrt{\frac{nV\infdiv*{p}{q}}{1-\epsilon}}.
\end{equation}
where $p_{\rvs{X}_1^n}\gets p_\rv{X}^{\tensor n}$ with $p_\rv{X}$ being some capacity-achieving input distribution.
(Note that $V(p_\rv{X})=V_{\max}=0$ in this case.)
We omit the details.
\end{proof}
As a remark to Theorem~\ref{thm:cs:expand:moderate}, the moderate-deviation expansion of the channel coding rate with maximum error $1-\epsilon_n$ takes exactly the same form, \ie, 
\begin{equation}
\frac{1}{n} \log{N_{1-\epsilon_n}^\star(W_{\rv{Y}|\rv{X}}^{\tensor n})} = C(W_{\rv{Y}|\rv{X}}) + \sqrt{2V_{\max}(W_{\rv{Y}|\rv{X}})}\cdot a_n + o(a_n),
\end{equation}
where $\epsilon_n\defeq 2^{-na_n^2}$ and $\{a_n\}$ is some moderate sequence.
(This specific expression is the classical special case of~\cite[Eq.~(162)]{chubb2017moderate}, which can be derived from earlier works~\cite{altuug2010moderate, altuug2014moderate}.)
This is consistent \wrt our result in the fixed-error regime (see Theorem~\ref{thm:cs:cc}).



\section{Channel Simulation with No-Signaling Assistance}
\label{sec:NS-channel_simulation}

In this section, we extend our discussion to channel simulations with no-signaling correlations as described at the beginning of Section~\ref{sec:overview}.
Recall that we consider the problem of simulating a channel $W_{\rv{Y}|\rv{X}}\in\set{P}(\set{Y}|\set{X})$ using a pair of no-singling correlated encoder and decoder, namely their joint encoder-decoder map $N_{\rv{IY}|\rv{XJ}}\in\set{P}([M]\times\set{Y}|\set{X}\times[M])$ satisfies conditions~\eqref{eq:ns:requirement:1} and~\eqref{eq:ns:requirement:2}, and the approximate channel $\tilde{W}_{\rv{Y}|\rv{X}}$ is constructed as $\tilde{W}_{\rv{Y}|\rv{X}}(y|x) = \sum_{m\in[M]} \mathcal{N}_{\rv{IY}|\rv{XJ}}(m,y|x,m)$ (see Fig.~\ref{fig:simulation:task:NS}).
No-signaling channel simulations have been studied in different setups (see, \eg, \cite{fang2019quantum, cubitt2011zero, matthews2012linear}).
We emphasize that our no-signaling results are conceptually different from the quantum results in~\cite{fang2019quantum}.
In principle, one could take~\cite{fang2019quantum} and do further symmetry reductions for classical channels to arrive at the same no-signaling linear programs as in this section.
However, such a calculation would be less insightful as it would not automatically give the form of the optimal no-signaling box as in this section.

Both non-correlated and randomness-assisted simulations are special cases of no-signaling simulations:
For non-correlated codes, their joint encoder-decoder maps are in the product form of $\mathcal{E}(i|x)\cdot\mathcal{D}(y|j)$ where $\mathcal{E}$ and $\mathcal{D}$ are the encoding and the decoding functions (or conditional pmfs) of the code.
The randomness-assisted channel simulation scheme described in Section~\ref{sec:introduction} has joint encoder-decoder map in the following form
\begin{equation}
N_{\rv{IY}|\rv{XJ}}(i,y|x,j) = \sum_{s\in\set{S}} p_\rv{S}(s)\cdot \mathcal{E}_s(i|x) \cdot \mathcal{D}_s(y|j).
\end{equation}
Another example of a no-signaling resource is a pair of entangled quantum systems.
In the following, we refer to $N_{\rv{IY}|\rv{XJ}}$ as a no-signaling code, and $M$ as the (alphabet) size of the code.

Given a fixed communication cost $c$ (in terms of alphabet size), we denote $\epsilon^{\fnc{NS}}(W_{\rv{Y}\vert\rv{X}},c)$ the minimal deviation (in TVD) to simulate $W_{\rv{Y}|\rv{X}}$ using a no-signaling code of size at most $c$.
On the other hand, given a fixed deviation tolerance $\epsilon$ (in TVD), we denote $C^{\fnc{NS}}(W_{\rv{Y}|\rv{X}},\epsilon)$ the minimal alphabet size of a no-signaling code simulating $W_{\rv{Y}|\rv{X}}$ within TVD tolerance $\epsilon$.
In the remainder of this section, we first express $\epsilon^{\fnc{NS}}(W_{\rv{Y}\vert\rv{X}},c)$ and $C^{\fnc{NS}}(W_{\rv{Y}|\rv{X}},\epsilon)$ as linear programs, and then relate the latter to our results in the previous sections.

\begin{proposition}\label{prop:NS}
Let $W_{\rv{Y}|\rv{X}}\in\set{P}(\set{Y}|\set{X})$ be a channel. 
For any integer $c\geq 2$, it holds that 
\begin{equation}\label{eq:NS:LP:eps:c}\begin{aligned}
\epsilon^{\fnc{NS}}(W_{\rv{Y}|\rv{X}},c) = \inf\ & \gamma \\
\fnc{s.t.}\ & \mu_\rv{XY}\in\mathbb{R}^{\set{X}\times\set{Y}} \\
& \sum_{y\in\set{Y}} \mu_\rv{XY}(x,y)\leq \gamma \quad \forall x\in\set{X}\\
& \mu_\rv{XY}(x,y)\geq \tilde{W}_{\rv{Y}|\rv{X}}(y|x) - W_{\rv{Y}|\rv{X}}(y|x) \quad\forall (x,y)\in\set{X}\times\set{Y}\\
& \tilde{W}_{\rv{Y}|\rv{X}}\in\set{P}(\set{Y}|\set{X})\\
& \tilde{W}_{\rv{Y}|\rv{X}}(y|x) \leq \zeta(y) \quad\forall (x,y)\in\set{X}\times\set{Y}\\
& \sum_{y\in\set{Y}} \zeta(y) = c
\end{aligned}.\end{equation}
Moreover, for any $\epsilon\in(0,1)$, it holds that 
\begin{align}
    \label{eq:NS:LP:c:eps}
    C^{\fnc{NS}}(W_{\rv{Y}|\rv{X}},\epsilon) &= \begin{aligned}[t]\inf\ & \ceil{\sum_{y\in\set{Y}} \zeta(y)} \\
    \fnc{s.t.}\ & \tilde{W}_{\rv{Y}|\rv{X}}\in\set{P}(\set{Y}|\set{X})\\
    & \norm{\tilde{W}_{\rv{Y}|\rv{X}} - W_{\rv{Y}|\rv{X}}}_\fnc{tvd} \leq \epsilon\\
    & \tilde{W}_{\rv{Y}|\rv{X}}(y|x) \leq \zeta(y) \quad\forall (x,y)\in\set{X}\times\set{Y}
    \end{aligned}\\
    \label{eq:NS:LP:c:eps:1}
    &= \inf_{\tilde{W}_{\rv{Y}|\rv{X}}:\norm{\tilde{W}_{\rv{Y}|\rv{X}}-W_{\rv{Y}|\rv{X}}}_\fnc{tvd}\leq\epsilon}\ \ceil{\sum_{y\in\set{Y}} \max_{x\in\set{X}} \tilde{W}_{\rv{Y}|\rv{X}}(y|x)}.
\end{align}
\end{proposition}

Proposition~\ref{prop:NS} could in principle be deduced from the results in~\cite{fang2019quantum}, where the authors expressed the cost of no-signaling quantum channel simulation as a semi-definite program (cf.~\cite[App.~A]{fang2019quantum}. However, we give a direct, self-contained proof in the framework of classical information theory.
\begin{proof}[Proof of~\eqref{eq:NS:LP:eps:c}]
We start by expressing the minimum deviation (in TVD) for simulating channel $W_{\rv{Y}|\rv{X}}\in\set{P}(\set{Y}|\set{X})$ within some fixed communication cost $c\geq 2$ as the following linear program
\begin{equation}\label{eq:NS:LP:1}\begin{aligned}
\epsilon^{\fnc{NS}}(W_{\rv{Y}\vert\rv{X}},c) = \inf\ & \max_{x} \norm{W_{\rv{Y}|\rv{X}}(\cdot|x) - \sum_{i\in\set{M}}N_{\rv{IY}|\rv{XJ}}(i,\cdot|x,i)}_{\fnc{tvd}} \\
\fnc{s.t.}\ & N_{\rv{IY}|\rv{XJ}}\in\set{P}(\set{M}\times\set{Y}|\set{X}\times\set{M})\\
& \sum_{y\in\set{Y}} N_{\rv{IY}|\rv{XJ}}(i,y|x,j) = N_{\rv{I}|\rv{X}}(i|x) \quad \forall j\in\set{J}\\
& \sum_{i\in\set{M}} N_{\rv{IY}|\rv{XJ}}(i,y|x,j) = N_{\rv{Y}|\rv{J}}(y|j) \quad \forall x\in\set{X}\\
& \size{\set{M}} = c
\end{aligned}.\end{equation}
Notice that one can rewrite the TVD-distance between channels as
\begin{align}
\norm{W_{\rv{Y}|\rv{X}} - \tilde{W}_{\rv{Y}|\rv{X}}}_{\fnc{tvd}}
&= \max_{x\in\set{X}} \norm{W_{\rv{Y}|\rv{X}}(\cdot|x) - \tilde{W}_{\rv{Y}|\rv{X}}(\cdot|x)}_\fnc{tvd}
=\frac{1}{2} \max_{x\in\set{X}} \norm{W_{\rv{Y}|\rv{X}}(\cdot|x) - \tilde{W}_{\rv{Y}|\rv{X}}(\cdot|x)}_1 \\
\label{eq:eqiv:channle:tvd}
&= \inf\left\{\gamma\in\mathbb{R}\middle\vert \mu_\rv{XY}\in\mathbb{R}_{\geq 0}^{\set{X}\times\set{Y}}, \sum_{y\in\set{Y}}\mu_\rv{XY}(\cdot,y )\leq \gamma, \mu_\rv{XY}\geq \tilde{W}_{\rv{Y}|\rv{X}}-W_{\rv{Y}|\rv{X}} \right\}.
\end{align}
This enables us to rewrite~\eqref{eq:NS:LP:1} as 
\begin{equation}\label{eq:NS:LP:2}\begin{aligned}
\epsilon^{\fnc{NS}}(W_{\rv{Y}\vert\rv{X}},c) = \inf\ & \gamma \\
\fnc{s.t.}\ & \mu_\rv{XY}\in\mathbb{R}_{\geq 0}^{\set{X}\times\set{Y}} \\
& \sum_{y\in\set{Y}} \mu_\rv{XY}(x,y)\leq \gamma \quad \forall x\in\set{X}\\
& \mu_\rv{XY}(x,y)\geq \sum_{i\in\set{M}}N_{\rv{IY}|\rv{XJ}}(i,y|x,i) - W_{\rv{Y}|\rv{X}}(y|x) \quad\forall (x,y)\in\set{X}\times\set{Y}\\
& N_{\rv{IY}|\rv{XJ}}\in\set{P}(\set{M}\times\set{Y}|\set{X}\times\set{M})\\
& \sum_{y\in\set{Y}} N_{\rv{IY}|\rv{XJ}}(i,y|x,j) = N_{\rv{I}|\rv{X}}(i|x) \quad \forall j\in\set{J}\\
& \sum_{i\in\set{M}} N_{\rv{IY}|\rv{XJ}}(i,y|x,j) = N_{\rv{Y}|\rv{J}}(y|j) \quad \forall x\in\set{X}\\
& \size{\set{M}} = c
\end{aligned}.\end{equation}
It remains to show~\eqref{eq:NS:LP:2} to be equivalent to~\eqref{eq:NS:LP:eps:c}.
We first show that for any feasible $\tilde{W}_{\rv{Y}|\rv{X}}$ and $\zeta$ for~\eqref{eq:NS:LP:eps:c}, one can construct some $N_{\rv{IY}|\rv{XJ}}\in\set{P}(\set{M}\times\set{Y}|\set{X}\times\set{M})$ satisfying the conditions in~\eqref{eq:NS:LP:2} such that
\begin{equation}
\tilde{W}_{\rv{Y}|\rv{X}}(y|x) = \sum_{i\in\set{M}}N_{\rv{IY}|\rv{XJ}}(i,y|x,i) \quad \forall (x,y)\in\set{X}\times\set{Y}.
\end{equation}
This can be shown by constructing $N_{\rv{IY}|\rv{XJ}}$ as
\begin{equation}
N_{\rv{IY}|\rv{XJ}}(i,y|x,j) = \begin{cases}
\frac{\tilde{W}_{\rv{Y}|\rv{X}}(y|x)}{c} &\text{for $i=j$}\\
\frac{\zeta(y) - \tilde{W}_{\rv{Y}|\rv{X}}(y|x)}{c(c-1)} &\text{for $i\neq j$}
\end{cases}.
\end{equation}
It is straightforward to check that $\sum_{i\in\set{M}}N_{\rv{IY}|\rv{XJ}}(i,y|x,i)  = \tilde{W}_{\rv{Y}|\rv{X}}(y|x)$.
Since $\zeta(y)\geq W_{\rv{Y}|\rv{X}}(y|x)$ and $c\geq 2$ we know that $N_{\rv{IY}|\rv{XJ}}$ is nonnegative.
To verify $N_{\rv{IY}|\rv{XJ}}$ to be a conditional pmf, we have
\begin{equation}
    \sum_{i\in\set{M},y\in\set{Y}} N_{\rv{IY}|\rv{XJ}}(i,y|x,j) = \sum_y \left(\frac{\tilde{W}_{\rv{Y}|\rv{X}}(y|x)}{c} + \frac{\zeta(y) - \tilde{W}_{\rv{Y}|\rv{Y}}(y|x)}{c(c-1)}\cdot (c-1)\right)
    =\sum_{y\in\set{Y}}\frac{\zeta(y)}{c} = 1.
\end{equation}
To verify $N_{\rv{IY}|\rv{XJ}}$ to be no-signaling, we have
\begin{align}
    \sum_{y\in\set{Y}} N_{\rv{IY}|\rv{XJ}}(i,y|x,j) &= \frac{1}{c} && \forall j\in\set{M},\\
    \sum_{i\in\set{M}} N_{\rv{IY}|\rv{XJ}}(i,y|x,j) &= \frac{\tilde{W}_{\rv{Y}|\rv{X}}(y|x)}{c} + \sum_{i\neq j} \frac{\zeta(y) - \tilde{W}_{\rv{Y}|\rv{X}}(y|x)}{c(c-1)} = \frac{\zeta(y)}{c} && \forall x\in\set{X}.
\end{align}
On the other hand, suppose $N_{\rv{IY}|\rv{XJ}}$ being some feasible no-signaling conditional pmf for~\eqref{eq:NS:LP:2}, we claim that $\tilde{W}_{\rv{Y}|\rv{X}}$ and $\zeta$ defined below are feasible for~\eqref{eq:NS:LP:eps:c}:
\begin{align}
\tilde{W}_{\rv{Y}|\rv{X}}(y|x) &\defeq \sum_{i\in\set{M}}N_{\rv{IY}|\rv{XJ}}(i,y|x,i) && \forall (x,y)\in\set{X}\times\set{Y},\\
\zeta(y) &\defeq \sum_{j\in\set{M}} N_{\rv{Y}|\rv{J}}(y|j) && \forall y\in\set{Y}.
\end{align}
Since it is straightforward to check $\tilde{W}_{\rv{Y}|\rv{X}}$ to be a conditional pmf and that $\sum_y \zeta(y) = c$, it remains to verify $\tilde{W}_{\rv{Y}|\rv{X}}(y|x)\leq\zeta(y)$ for all $y\in\set{Y}$, which is indeed the case since
\begin{equation}
\tilde{W}_{\rv{Y}|\rv{X}}(y|x)= \sum_{i\in\set{M}}N_{\rv{IY}|\rv{XJ}}(i,y|x,i)
\leq \sum_{i,j\in\set{M}}N_{\rv{IY}|\rv{XJ}}(i,y|x,j) = \sum_{j\in\set{M}} N_{\rv{Y}|\rv{J}}(y|j) = \zeta(y).
\end{equation}
Combining the above two arguments, we have established the equivalence between~\eqref{eq:NS:LP:2} and~\eqref{eq:NS:LP:eps:c}.
\end{proof}
\begin{proof}[Proof of~\eqref{eq:NS:LP:c:eps}]
As a direct result of~\eqref{eq:NS:LP:eps:c}, we know a tolerance-cost pair $(\epsilon,c)$ (where the integer $c\geq 2$) is attainable for if and only if their exists some $\mu_\rv{XY}\in\mathbb{R}_{\geq 0}^{\set{X}\times\set{Y}}$ and $\tilde{W}_{\rv{Y}|\rv{X}}\in\set{P}(\set{Y}|\set{X})$ such that 
\begin{align}
\sum_{y\in\set{Y}}\mu_\rv{XY}(x,y)&\leq \epsilon &&\forall x\in\set{X}\\
\mu_\rv{XY}(x,y)&\geq \tilde{W}_{\rv{Y}|\rv{X}}(y|x) - W_{\rv{Y}|\rv{X}}(y|x) &&\forall (x,y)\in\set{X}\times\set{Y}\\
\tilde{W}_{\rv{Y}|\rv{X}}(y|x) &\leq \zeta(y) &&\forall (x,y)\in\set{X}\times\set{Y}\\
\sum_{y\in\set{Y}} \zeta(y) &= c.
\end{align}
Thus, for a given deviation $\epsilon$ (in TVD), one can express the minimal communication cost achieving such deviation as the following optimization problem.
\begin{equation}\label{eq:cost:NS:1}
\begin{aligned}
    C^{\fnc{NS}}(W_{\rv{Y}|\rv{X}},\epsilon) = \inf\ & \ceil{\sum_{y\in\set{Y}} \zeta(y)} \\
    \fnc{s.t.}\ & \mu_\rv{XY}\in\mathbb{R}^{\set{X}\times\set{Y}} \\
    & \sum_{y\in\set{Y}} \mu_\rv{XY}(x,y)\leq \epsilon \quad \forall x\in\set{X}\\
    & \mu_\rv{XY}(x,y)\geq \tilde{W}_{\rv{Y}|\rv{X}}(y|x) - W_{\rv{Y}|\rv{X}}(y|x) \quad\forall (x,y)\in\set{X}\times\set{Y}\\
    & \tilde{W}_{\rv{Y}|\rv{X}}\in\set{P}(\set{Y}|\set{X})\\
    & \tilde{W}_{\rv{Y}|\rv{X}}(y|x) \leq \zeta(y) \quad\forall (x,y)\in\set{X}\times\set{Y}
\end{aligned}.\end{equation}
Using~\eqref{eq:eqiv:channle:tvd}, we can simplify the above linear program into~\eqref{eq:NS:LP:c:eps}.
\end{proof}
\par
In the following, we show how $C^{\fnc{NS}}(W_{\rv{Y}|\rv{X}},\epsilon)$ as in~\eqref{eq:NS:LP:c:eps:1} is related to the cost of the random-assisted channel simulation.
In particular, we have the following theorem.
\begin{theorem}\label{thm:NS}
Let $W_{\rv{Y}|\rv{X}}$ be a channel from $\set{X}$ to $\set{Y}$, and let $\epsilon\in(0,1)$.
Let $C^{\fnc{NS}}(W_{\rv{Y}|\rv{X}},\epsilon)$ denote the minimal attainable size of no-signaling $\epsilon$-simulation codes for $W_{\rv{Y}|\rv{X}}$, then
\begin{equation}\label{eq:cost:NS:4}
    \log{C^{\fnc{NS}}(W_{\rv{Y}|\rv{X}},\epsilon)} = \ceil{I_{\max}^\epsilon(W_{\rv{Y}|\rv{X}})}_{\log{\mathbb{Z}_{>0}}}
\end{equation}
where for any $r\in\mathbb{R}$ we denote $\ceil{r}_{\log{\mathbb{Z}_{>0}}}$ the smallest number in $\log{\mathbb{Z}_{>0}}\defeq\left\{\log{a}:a\in\mathbb{Z}, a>0\right\}$ that is no smaller than $r$.
\end{theorem}
Compare the RHS of~\eqref{eq:cost:NS:4}  with~\eqref{eq:converse:bound}.
We see that the task of no-signaling channel simulations is no more difficult than the random-assisted channel simulation.
Since, in the $n$-fold case for the random-assisted channel simulation, the converse bounds as in~\eqref{eq:converse:bound} is tight up to the second order expansion term, no-signaling channel simulations must also have the same first and second order term in the expansions.
To prove Theorem~\ref{thm:NS}, we need the following lemma.
\begin{lemma}\label{lem:inf:max}
Let $\set{X}$ be a finite set, and $f$ be a non-negative valued function on $\set{X}$.
It holds that
\begin{equation}
    \inf_{p\in\set{P}(\set{X})} \max_{x\in\set{X}} \frac{f(x)}{p(x)} = \sum_{x\in\set{X}} f(x).
\end{equation}
\end{lemma}
\begin{proof}
See Appendix~\ref{app:proof:lem:inf:max}.
\end{proof}
\begin{proof}[Proof of Theorem~\ref{thm:NS}]
Let $W_{\rv{Y}|\rv{X}}$ be a channel from $\set{X}$ to $\set{Y}$, and let $\epsilon\in(0,1)$.
Starting from~\eqref{eq:NS:LP:c:eps:1}, we have the following chain of identities:
    \begin{equation}
        C^{\fnc{NS}}(W_{\rv{Y}|\rv{X}},\epsilon) = \inf_{\tilde{W}_{\rv{Y}|\rv{X}}:\norm{\tilde{W}_{\rv{Y}|\rv{X}}-W_{\rv{Y}|\rv{X}}}_\fnc{tvd}\leq\epsilon}\ \ceil{\sum_{y\in\set{Y}} \max_{x\in\set{X}} \tilde{W}_{\rv{Y}|\rv{X}}(y|x)} \hspace{147pt} \text{copied from~\eqref{eq:NS:LP:c:eps:1}}
    \end{equation}
    \begin{equation}\label{eq:thm:NS:1}
        \phantom{C^{\fnc{NS}}(W_{\rv{Y}|\rv{X}},\epsilon)} = \ceil{ \inf_{\tilde{W}_{\rv{Y}|\rv{X}}: \norm{\tilde{W}_{\rv{Y}|\rv{X}} - W_{\rv{Y}|\rv{X}}}_\fnc{tvd}\leq\epsilon}\ \sum_{y\in\set{Y}} \max_{x\in\set{X}} \tilde{W}_{\rv{Y}|\rv{X}}(y|x)} \hspace{146pt} \text{see the text below}
    \end{equation}
    \begin{equation}
        \phantom{C^{\fnc{NS}}(W_{\rv{Y}|\rv{X}},\epsilon)} = \ceil{\multiadjustlimits{
            \inf_{\tilde{W}_{\rv{Y}|\rv{X}}:\norm{\tilde{W}_{\rv{Y}|\rv{X}}-W_{\rv{Y}|\rv{X}}}_\fnc{tvd}\leq\epsilon}, 
            \inf_{q_\rv{Y}\in\set{P}(\set{Y})}, 
            \max_{y\in\set{Y}}
        } \frac{\max_{x\in\set{X}} \tilde{W}_{\rv{Y}|\rv{X}}(y|x)}{q_\rv{Y}(y)}} \hspace{111pt} \text{by Lemma~\ref{lem:inf:max}}
    \end{equation}
    \begin{equation}\label{eq:thm:NS:2}
        \phantom{C^{\fnc{NS}}(W_{\rv{Y}|\rv{X}},\epsilon)} = \ceil{\multiadjustlimits{
            \inf_{q_\rv{Y}\in\set{P}(\set{Y})}, 
            \inf_{\tilde{W}_{\rv{Y}|\rv{X}}:\norm{\tilde{W}_{\rv{Y}|\rv{X}}-W_{\rv{Y}|\rv{X}}}_\fnc{tvd}\leq\epsilon}, 
            \max_{x\in\set{X}}
        } 2^{D_{\max}\infdiv*{\tilde{W}_{\rv{Y}|\rv{X}}(\cdot|x)}{q_\rv{Y}}}} \hspace{67pt} \text{by the definition of } D_{\max}
    \end{equation}
    \begin{equation}
        \phantom{C^{\fnc{NS}}(W_{\rv{Y}|\rv{X}},\epsilon)} = \ceil{2^{\multiadjustlimits{
            \inf_{q_\rv{Y}\in\set{P}(\set{Y})}, 
            \inf_{\tilde{W}_{\rv{Y}|\rv{X}}:\norm{\tilde{W}_{\rv{Y}|\rv{X}} - W_{\rv{Y}|\rv{X}}}_\fnc{tvd}\leq\epsilon}, 
            \max_{x\in\set{X}}
        } D_{\max}\infdiv*{\tilde{W}_{\rv{Y}|\rv{X}}(\cdot|x)}{q_\rv{Y}}}} \hspace{69pt} x\mapsto 2^x \text{ is strictly increasing}
    \end{equation}
Here,~\eqref{eq:thm:NS:1} is due to the compactness of $\left\{\tilde{W}_{\rv{Y}|\rv{X}}: \norm{\tilde{W}_{\rv{Y}|\rv{X}} - W_{\rv{Y}|\rv{X}}}_\fnc{tvd}\leq\epsilon\right\}$ and the continuity of the function $\tilde{W}_{\rv{Y}|\rv{X}} \mapsto \sum_{y\in\set{Y}} \max_{x\in\set{X}} \tilde{W}_{\rv{Y}|\rv{X}}(y|x)$.
In particular, for any compact set $\set{A}\subset\mathbb{R}^n$ and any continuous function $f:\mathbb{R}^n\to\mathbb{R}$, we have 
$\ceil{\inf_{x\in\set{A}} f(x)} = \ceil{\min_{x\in\set{A}} f(x)} = \ceil{f(x^\star)} = \min_{x\in\set{A}} \ceil{f(x)}$
where $x^\star$ is a minimizer of both $f(x)$ and $\ceil{f(x)}$ in $\set{A}$.
Eq.~\eqref{eq:thm:NS:2} utilizes the fact that $x\mapsto 2^x$ is a strictly increasing function.
Notice that the logarithmic function is increasing.
Hence, we can take the logarithm on both sides of the last inequality and get 
\begin{equation*}
\log{C^{\fnc{NS}}(W_{\rv{Y}|\rv{X}},\epsilon)}
= \ceil{\multiadjustlimits{
    \inf_{q_\rv{Y}\in\set{P}(\set{Y})}, 
    \inf_{\tilde{W}_{\rv{Y}|\rv{X}}:\norm{\tilde{W}_{\rv{Y}|\rv{X}} - W_{\rv{Y}|\rv{X}}}_\fnc{tvd}\leq\epsilon}, 
    \max_{x\in\set{X}}
    } D_{\max}\infdiv*{\tilde{W}_{\rv{Y}|\rv{X}}(\cdot|x)}{q_\rv{Y}}}_{\log{\mathbb{Z}_{>0}}}
= \ceil{I_{\max}^\epsilon(W_{\rv{Y}|\rv{X}})}_{\log{\mathbb{Z}_{>0}}} . \qedhere
\end{equation*}
\end{proof}


\subsection*{Example: $n$-fold binary symmetric channel}

We now investigate the case of simulating $n$ copies of a binary symmetric channel with crossover probability $\delta$.
In this case, the linear programs \eqref{eq:NS:LP:eps:c} and \eqref{eq:cost:NS:1} simplify considerably.
We have
\begin{align}
     W_{\rv{Y}^n|\rv{X}^n} = \begin{pmatrix}
     1-\delta & \delta \\ \delta & 1-\delta
     \end{pmatrix}^{\otimes n}   
\end{align}
For $n, k\in \mathbb{N}$ and $0\leq k\leq n$, let $P^n_k(I, X)$ refer to the binary matrix obtained by the sum of all $n$-fold tensor factors of $I$ and $X$, where $I = \left(\begin{smallmatrix}1&0\\0&1\end{smallmatrix}\right)$ and $X = \left(\begin{smallmatrix}0&1\\1&0\end{smallmatrix}\right)$, with exactly $k$ factors of $X$.
For instance, $P^3_1(I, X) = I\otimes I\otimes X + I\otimes X\otimes I + X\otimes I\otimes I$.
We make the following observations
\begin{itemize}
    \item $W_{\rv{Y}^n|\rv{X}^n} = ((1-\delta)I + \delta X)^{\otimes n} = \sum_{k=0}^n (1-\delta)^{n-k}\delta^k \cdot P^n_k(I, X)$.
    \item The symmetries of $W_{\rv{Y}^n|\rv{X}^n}$ allow us to choose $\tilde{W}_{\rv{Y}^n|\rv{X}^n} = \sum_{k=0}^n r_k P^n_k(I, X)$, $p_{\rv{XY}} = \sum_{k=0}^n y_k P^n_k(I, X)$ and $\zeta(y) = s$, where $s$ is some constant without loss of generality.
    \item The set of matrices $\{P^n_1(I,X), P^n_2(I,X),..., P^n_n(I,X)\}\}$ form a ``partition" of the entries of a $2^n\times 2^n$ matrix. That is, $\sum_{k=0}^n P^n_k(I,X) = J_{2^n}$, where $J_m$ is the $m\times m$ matrix with all entries equal to $1$. Additionally, since for each $k$, $P^n_k(I,X)$ is a $2^n\times 2^n$ binary matrix, it also holds that if $P^n_k(I,X)(i,j) = 1$, then $P^n_{k'}(I,X)(i,j) = 0$ for all $k\neq k'$ and for all $i,j$.
\end{itemize}

These observations allow us to simplify the linear program in \eqref{eq:NS:LP:eps:c} for an $n$-fold binary symmetric channel with crossover probability $\delta$ as follows.
\begin{equation}
\label{eq:NS:LP:3:symmetrized}
\begin{aligned}
\epsilon^{\fnc{NS}}(W_{\rv{Y}^n|\rv{X}^n},c) = \inf\ & \gamma \\
\fnc{s.t.}\ & y_k \geq 0\quad \forall k\\
& \sum_{k=0}^n \binom{n}{k} y_k\leq \gamma\\
& y_k\geq r_k - (1-\delta)^{n-k}\delta^k \quad\forall k\\
& \sum_{k=0}^n \binom{n}{k} r_k= 1\\
& 0\leq r_k \leq s \quad\forall k\\
& 2^n s = c.
\end{aligned}
\end{equation}

Similarly, we can simplify \eqref{eq:cost:NS:1} as follows. 

\begin{equation}
\label{eq:cost:NS:3:symmetrized}
\begin{aligned}
    C^{\fnc{NS}}(W_{\rv{Y}^n|\rv{X}^n},\epsilon) = \inf & \ceil{2^n s} \\
    \fnc{s.t.}\ & y_k \geq 0\ \quad \forall k\\
    & \sum_{k=0}^n \binom{n}{k} y_k\leq \epsilon \\
    & y_k\geq r_k - (1-\delta)^{n-k}\delta^k \quad\forall k\\
    & 0\leq r_k \leq s \quad\forall k\\
    & \sum_{k=0}^n \binom{n}{k} r_k= 1
    \end{aligned}
\end{equation}

In Fig.~\ref{fig:BSC:NS:sim-intro} in the introduction, we solve the linear program in \eqref{eq:cost:NS:3:symmetrized} for a binary symmetric channel with $\delta = 0.1$ and $\epsilon = 0.05$.
The second-order approximation of the simulation cost, the capacity of the channel and a finite blocklength converse bound on the capacity are also shown.


\section{Generalization to Broadcast Channels}
\label{sec:broadcast_channels}

In this section, we consider the problem of simulating a broadcast channel using identity channels from the sender to each of the receivers assisted by unconstrained shared randomness between the sender and each of the receivers.
For better readability, we restrict our discussion to bipartite broadcast channels, \ie, with two receivers.
All discussions in this section can be generalized to general broadcast channels with arbitrary (but finite) number of receivers.
Note that some of the important techniques employed in this section are either pre-existing ones (\eg, Lemma~\ref{lem:convex:split:bipartite}) or their generalizations (\eg, Lemma~\ref{lem:common:non-lockability} and Lemma~\ref{lem:Ds:2}).
Fortunately, this does not prevent us from obtaining new one-shot bounds and interesting asymptotic results.

Let $W_{\rv{YZ}|\rv{X}}\in\set{P}(\set{Y}\times\set{Z}|\set{X})$ be a broadcast channel.
We are interested in constructing a channel $\tilde{W}_{\rv{YZ}|\rv{X}}\in\set{P}(\set{Y}\times\set{Z}|\set{X})$ such that
\begin{equation}\label{eq:boardcast:simulation:epsilon}
    \norm{\tilde{W}_{\rv{YZ}|\rv{X}}-W_{\rv{YZ}|\rv{X}}}_\fnc{tvd} =\sup_{x\in\set{X}}\norm{\tilde{W}_{\rv{YZ}|\rv{X}}(\cdot,\cdot|x) - W_{\rv{YZ}|\rv{X}}(\cdot,\cdot|x)}_\fnc{tvd} \leq \epsilon
\end{equation}
using a pair of identity channels $\ID_M$ and $\ID_N$ from the sender to each of the two receivers in the random-assisted scenario (also see Fig.~\ref{fig:broadcast:simulation}).
\begin{figure}
\centering\includegraphics{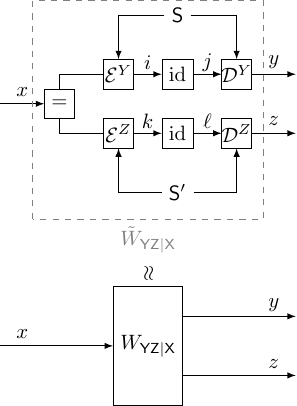}
\caption{The task of simulating a bipartite broadcast channel using shared randomness and a pair of identity channels.}
\label{fig:broadcast:simulation}
\end{figure}
In particular, $\tilde{W}_{\rv{YZ}|\rv{Z}}$ is constructed as
\begin{align}
\tilde{W}_{\rv{YZ}|\rv{X}}(y,z|x) &= \!\sum_{s\in\set{S} \atop s'\in\set{S}'}\!\! p_\rv{S}(s)\!\cdot\! p_{\rv{S}'}(s') \!\cdot\!\!\! \sum_{i,j\in[M] \atop k,\ell\in[N]} \!\!\!\mathcal{E}^Y_{\rv{M}|\rv{XS}}(i|x,s) \!\cdot\! \mathcal{E}^Z_{\rv{N}|\rv{XS}'}(k|x,s') \!\cdot\! \ID_M(j|i) \!\cdot\! \ID_{N}(\ell|k) \!\cdot\!\mathcal{D}^Y_{\rv{Y}|\rv{MS}}(y|j,s) \!\cdot\! \mathcal{D}^Z_{\rv{Z}|\rv{NS}'}(z|\ell,s')\\
&= \sum_{s\in\set{S} \atop s'\in\set{S}'} p_\rv{S}(s)\cdot p_{\rv{S}'}(s') \cdot \sum_{m\in[M], n\in[N]} \mathcal{E}^Y_{\rv{M}|\rv{XS}}(m|x,s) \cdot \mathcal{E}^Z_{\rv{N}|\rv{XS}'}(n|x,s')\cdot \mathcal{D}^Y_{\rv{Y}|\rv{MS}}(y|m,s) \cdot \mathcal{D}^Z_{\rv{Z}|\rv{NS}'}(z|n,s')
\end{align}
where $\mathcal{E}^Y_{\rv{M}|\rv{XS}}\in\set{P}([M]|\set{X}\times\set{S})$, $\mathcal{E}^Z_{\rv{N}|\rv{XS}'}\in\set{P}([N]|\set{X}\times\set{S}')$ and $\mathcal{D}^Y_{\rv{Y}|\rv{M}\rv{S}}\in\set{P}(\set{Y}|[M]\times\set{S})$, $\mathcal{D}^Z_{\rv{Z}|\rv{N}\rv{S}'}\in\set{P}(\set{Z}|[N]\times\set{S}')$ are some randomized encoding and decoding maps on the sender and the receivers side, respectively; and $\rv{S}$ and $\rv{S}'$ are some random variables shared between the sender and each of the two receivers.

Similar to the point-to-point case, if~\eqref{eq:boardcast:simulation:epsilon} is satisfied, we refer to $(\mathcal{E}_{\rv{M}|\rv{XS}}^Y, \mathcal{E}_{\rv{N}|\rv{XS}'}^Z, \mathcal{D}_{\rv{Y}|\rv{MS}}^Y, \mathcal{D}_{\rv{Z}|\rv{NS}'}^Z, \rv{S}, \rv{S}')$ as a size-$(M,N)$ $\epsilon$-simulation code for $W_{\rv{YZ}|\rv{X}}$.
On the other hand, an integer pair $(M,N)$ is said to be $\epsilon$-attainable if there exists a size-$(M,N)$ $\epsilon$-simulation code for $W_{\rv{YZ}|\rv{X}}$.
A fundamental question is to characterize the set of all $\epsilon$-attainable pairs, denoted by $\set{M}^\star_\epsilon(W_{\rv{YZ}|\rv{X}})$.
As a first observation, if $(M,N)$ is attainable, so is $(M^+,N^+)$ for any integers $M^+\geq M$ and $N^+\geq N$.
In analogy to the upper and lower bounds of the minimum size of attainable messages as in the previous sections devoted to point-to-point cases, in this section, the achievability and converse bounds are presented as subset and superset of $\set{M}^\star_\epsilon(W_{\rv{YZ}|\rv{X}})$ instead.


\subsection{One-Shot Achievability Bound}
In the following, we present a protocol (see Proposition~\ref{prop:1abfsbc}) to show an one-shot achievability bound for simulating bipartite broadcast channel.
The protocol is a direct generalization of our earlier work in the point-to-point case~\cite{cao2022one-shot}, and is based on the bipartite convex split lemma as stated below.
\begin{lemma}[Bipartite Convex Split Lemma~{\cite[Fact~7, modified]{anshu2017unified}}]
    \label{lem:convex:split:bipartite}
    Let $\epsilon, \delta_1,\delta_2,\delta_3\in(0,1)$, such that $\delta_1^2+\delta_2^2+\delta_3^2 \leq \epsilon^2$.
    Let $(\rv{X},\rv{Y},\rv{Z})$ be distributed jointly on $\set{X}\times\set{Y}\times\set{Z}$ with pmf $p_{\rv{XYZ}}$.
    Let $q_\rv{Y}$ and $r_\rv{Z}$ be two pmfs over $\set{Y}$ and $\set{Z}$, respectively.
    Let $M$ and $N$ be two positive integers such that \begin{align}
    \label{eq:bcsl:anshu:1}
    \log{M} &\geq D_{s+}^{\epsilon_1}\infdiv*{p_\rv{XY}}{p_\rv{X}\times q_\rv{Y}} - \log{\delta_1^2},\\
    \label{eq:bcsl:anshu:2}
    \log{N} &\geq D_{s+}^{\epsilon_2}\infdiv*{p_\rv{XZ}}{p_\rv{X}\times r_\rv{Z}} - \log{\delta_2^2},\\
    \label{eq:bcsl:anshu:3}
    \log{M} + \log{N} &\geq D_{s+}^{\epsilon_3}\infdiv*{p_\rv{XYZ}}{p_\rv{X}\times q_\rv{Y} \times r_\rv{Z}} - \log{\delta_3^2},
    \end{align}
    for some $\epsilon_1, \epsilon_2, \epsilon_3\in(0,1)$ such that $\epsilon_1+\epsilon_2+\epsilon_3\leq \epsilon-\sqrt{\delta_1^2+\delta_2^2+\delta_2^3}$.
    Let $\rv{J}$ and $\rv{K}$ be uniformly distributed independently in $\{1,\ldots,M\}$ and $\{1,\ldots,N\}$, respectively.
    Let the joint random variables $(\rv{J},\rv{K},\rv{X},\rv{Y}_1,\ldots,\rv{Y}_M,\rv{Z}_1,\ldots,\rv{Z}_N)$ be distributed according to
    \begin{equation}\label{eq:lemma:convex:split}
    p_{\rv{X},\rv{Y}_1,\ldots,\rv{Y}_M,\rv{Z}_1,\ldots,\rv{Z}_N|\rv{J},\rv{K}}(x,y_1,\ldots,y_M,z_1,\ldots,z_N|j,k) = p_{\rv{XYZ}}(x,y_j,z_k)\cdot \prod_{i\neq j} q_\rv{Y}(y_i) \cdot \prod_{\ell\neq k} r_\rv{Z}(z_\ell).
    \end{equation}
    Then
    \begin{equation}\label{eq:convex:split:lemma:anshu}
    \norm{p_{\rv{X},\rv{Y}_1,\ldots,\rv{Y}_M,\rv{Z}_1,\ldots,\rv{Z}_N} - p_\rv{X}\times q_{\rv{Y}_1} \times \cdots \times q_{\rv{Y}_M} \times r_{\rv{Z}_1} \times \cdots \times r_{\rv{Z}_N}}_\fnc{tvd} \leq \epsilon.
    \end{equation}
\end{lemma}
\begin{proof}
This lemma can be proven via a rather straightforward generalization of the proof for~\cite[Fact~7]{anshu2017unified}.
For completeness, we include the proof in Appendix~\ref{app:proof:lem:convex:split:bipartite}.
\end{proof}
\begin{proposition}\label{prop:1abfsbc}
    Let $W_{\rv{YZ}|\rv{X}}\in\set{P}(\set{Y}\times\set{Z}|\set{X})$ be a bipartite broadcast channel, and let $\epsilon\in(0,1)$.
    The following set is a subset of $\set{M}^\star_\epsilon(W_{\rv{YZ}|\rv{X}})$:
    \begin{equation}\label{eq:broadcast:achievability}
    \set{M}^\text{in}_\epsilon(W_{\rv{YZ}|\rv{X}}) \defeq \left\{ (M,N)\in\mathbb{Z}_{>0}^2 \middle\vert 
    \begin{aligned}
    \log{M} &> D_{s+}^{\epsilon_1-\delta_1}\infdiv*{p_\rv{X}\cdot W_{\rv{Y}|\rv{X}}}{p_\rv{X}\times q_\rv{Y}} - \log{\delta_1^2} \\
    \log{N} &> D_{s+}^{\epsilon_2-\delta_2}\infdiv*{p_\rv{X}\cdot W_{\rv{Z}|\rv{X}}}{p_\rv{X}\times r_\rv{Z}} - \log{\delta_2^2} \\
    \log{MN} &> D_{s+}^{\epsilon_3-\delta_3}\infdiv*{p_\rv{X}\cdot W_{\rv{YZ}|\rv{X}}}{p_\rv{X}\times q_\rv{Y}\times r_\rv{Z}} - \log{\delta_3^2}
    \end{aligned}
    \quad\forall p_\rv{X}\in\set{P}(\set{X}),
    \right\}
    \end{equation}
    where $\epsilon_1,\epsilon_2,\epsilon_3$ are positive numbers such that $\epsilon_1+\epsilon_2+\epsilon_3\leq\epsilon$, $\delta_1\in(0,\epsilon_1)$, $\delta_2\in(0,\epsilon_2)$, $\delta_3\in(0,\epsilon_3)$, $q_\rv{Y}\in\set{P}(\set{Y})$, $r_\rv{Z}\in\set{P}(\set{Z})$, and the reduced channels $W_{\rv{Y}|\rv{X}}$ and $W_{\rv{Z}|\rv{X}}$ are defined as
    \begin{equation}
        W_{\rv{Y}|\rv{X}}(y|x) \defeq \sum_{z\in\set{Z}} W_{\rv{YZ}|\rv{X}}(y,z|x), \quad
        W_{\rv{Z}|\rv{X}}(y|x) \defeq \sum_{y\in\set{Y}} W_{\rv{YZ}|\rv{X}}(y,z|x).
    \end{equation}
\end{proposition}
\begin{proof}
For arbitrary $q_\rv{Y}\in\set{P}(\set{Y})$, $r_\rv{Z}\in\set{P}(\set{Z})$, we present a protocol for simulating $W_{\rv{YZ}|\rv{X}}$ by sending messages with alphabet sizes $M$ and $N$ to each of the receivers, respectively, where $(M,N)$ is any integer pair satisfying the conditions on the RHS of~\eqref{eq:broadcast:achievability}.
The protocol is as follows:
\begin{enumerate}
\item Let the sender and first receiver share i.i.d. random variables $(\rv{Y}_1,\ldots,\rv{Y}_M)$ where $\rv{Y}_k\sim q_\rv{Y}$ for each $k$.
\item Let the sender and second receiver share i.i.d. random variables $(\rv{Z}_1,\ldots,\rv{Z}_N)$ where $\rv{Z}_k\sim r_\rv{Z}$ for each $k$.
\item Upon receiving input $\rv{X}=x$, the sender generates a pair of random variables $\rv{J},\rv{K}$ (distributed on $\{1,\ldots,M\}\times\{1,\ldots,N\}$) according to the following conditional pmf
\begin{equation}
\tilde{p}_{\rv{J},\rv{K}|\rv{X},\rvs{Y}_1^M,\rvs{Z}_1^N}(j,k|x,\mathbf{y}_1^M,\mathbf{z}_1^N) \defeq \frac{W_{\rv{YZ}|\rv{X}}(y_j,z_k|x)\cdot \prod_{i\neq j} q_\rv{Y}(y_i) \cdot \prod_{\ell\neq k} r_\rv{Z}(z_\ell)}{\sum_{j'=1}^M \sum_{k'=1}^N W_{\rv{YZ}|\rv{X}}(y_{j'},z_{k'}|x)\cdot \prod_{i\neq {j'}} q_\rv{Y}(y_i) \cdot \prod_{\ell\neq {k'}} r_\rv{Z}(z_\ell)}.
\end{equation}
\item The sender sends $\rv{J}$ and $\rv{K}$ losslessly to the first and second receiver using $\log{M}$ bits and $\log{N}$ bits, respectively.
\item Upon receiving $\rv{J}$, the first receiver outputs $\rv{Y}_\rv{J}$.
\item Upon receiving $\rv{K}$, the second receiver outputs $\rv{Z}_\rv{K}$.
\end{enumerate}
It suffices to show that the joint pmf of the random variables $\rv{X}\rv{Y}_\rv{J}\rv{Z}_\rv{K}$ generated by the above protocol is $\epsilon$-close (in TVD) to $p_\rv{XYZ}\defeq p_\rv{X}\cdot W_{\rv{YZ}|\rv{X}}$ for any input distribution $p_\rv{X}$.

Let $\tilde{p}$ denote (joint/marginal/conditional, depending on the subscript) pmfs of the random variables $\rv{J}$, $\rv{K}$, $\rv{X}$, $\rv{Y}_1,$ $\ldots$ , $\rv{Y}_M$, $\rv{Z}_1$, $\ldots$ , $\rv{Z}_N$ as in the above protocol.
Define the joint pmf $p_{\rv{J},\rv{K},\rv{X},\rv{Y}_1,\ldots,\rv{Y}_M,\rv{Z}_1,\ldots,\rv{Z}_N}$ as 
\begin{equation}
    p_{\rv{J},\rv{K},\rv{X},\rvs{Y}_1^M,\rvs{Z}_1^N}(j,k,x,\mathbf{y}_1^M,\mathbf{z}_1^N) \defeq \frac{1}{M\cdot N} \cdot p_{\rv{X},\rvs{Y}_1^{M},\rvs{Z}_1^{N}|\rv{J},\rv{K}}(x,\mathbf{y}_1^{M},\mathbf{z}_1^{N}|j,k),
    \end{equation}
where $p_{\rv{X},\rvs{Y}_1^{M},\rvs{Z}_1^{N}|\rv{J},\rv{K}}$ has been defined in~\eqref{eq:lemma:convex:split}.
Note that, by definition, $p_\rv{X}=\tilde{p}_\rv{X}$ holds.
As a direct result of the protocol, we have
\begin{equation}
    \tilde{p}_{\rv{J},\rv{K},\rv{X},\rvs{Y}_1^M,\rvs{Z}_1^N}(j,k,x,\mathbf{y}_1^M,\mathbf{z}_1^N)
    = p_\rv{X}(x) \cdot \prod_{i=1}^M q_\rv{Y}(y_i) \cdot \prod_{\ell=1}^N r_\rv{Z}(z_\ell) \cdot \frac{W_{\rv{YZ}|\rv{X}}(y_j,z_k|x)\cdot \prod_{i\neq j} q_\rv{Y}(y_i) \cdot \prod_{\ell\neq k} r_\rv{Z}(z_\ell)}{\sum_{j',k'} W_{\rv{YZ}|\rv{X}}(y_{j'},z_{k'}|x)\cdot \prod_{i\neq {j'}} q_\rv{Y}(y_i) \cdot \prod_{\ell\neq {k'}} r_\rv{Z}(z_\ell)}.
\end{equation}
By Lemma~\ref{lem:convex:split:bipartite} and the requirements we imposed on $M$ and $N$ at the beginning of this proof (note that $\epsilon_1-\delta_1+\epsilon_2-\delta_2+\epsilon_2-\delta_2\leq \epsilon-\sqrt{\delta_1^2+\delta_2^2+\delta_3^2}$), it holds that 
\begin{equation}
\norm{\tilde{p}_{\rv{X}\rvs{Y}_1^M\rvs{Z}_1^N} - p_{\rv{X}\rvs{Y}_1^M\rvs{Z}_1^N}}_\fnc{tvd} = \norm{p_\rv{X}\cdot \prod_{i=1}^M q_{\rv{Y}_i}\cdot \prod_{\ell=1}^N r_{\rv{Z}_\ell} - \frac{1}{M\cdot N}\cdot p_\rv{X} \cdot \sum_{j,k} W_{\rv{Y}_j\rv{Z}_k|\rv{X}}\cdot\prod_{i\neq j}q_{\rv{Y}_i} \prod_{\ell\neq k}r_{\rv{Z}_\ell} }_\fnc{tvd} \leq \epsilon.
\end{equation}
Since $\tilde{p}_{\rv{J},\rv{K}|\rv{X},\rvs{Y}_1^M,\rvs{Z}_1^N} = p_{\rv{J},\rv{K}|\rv{X},\rvs{Y}_1^M,\rvs{Z}_1^N}$ (as deliberately designed), we have
\begin{equation}
    \norm{\tilde{p}_{\rv{J},\rv{K},\rv{X},\rvs{Y}_1^M,\rvs{Z}_1^N} - p_{\rv{J},\rv{K},\rv{X},\rvs{Y}_1^M,\rvs{Z}_1^N}}_\fnc{tvd}
    = \norm{\tilde{p}_{\rv{X}\rvs{Y}_1^M\rvs{Z}_1^N} - p_{\rv{X}\rvs{Y}_1^M\rvs{Z}_1^N}}_\fnc{tvd} \leq \epsilon.
\end{equation}
Using the data-processing inequality for the total-variation distance, we have
\begin{equation}
    \epsilon \geq \norm{\tilde{p}_{\rv{X}\rv{Y}_\rv{J}\rv{Z}_\rv{K}} - p_{\rv{X}\rv{Y}_\rv{J}\rv{Z}_\rv{K}}}_\fnc{tvd} = \norm{\tilde{p}_{\rv{X}\rv{Y}_\rv{J}\rv{Z}_\rv{K}} - p_\rv{X}\cdot W_{\rv{YZ}|\rv{X}}}_\fnc{tvd}.
\end{equation}
Since the above discussion holds for all input distributions $p_\rv{X}\in\set{P}(\set{X})$, we have finished the proof.
\end{proof}

\subsection{One-Shot Converse Bound}
In the following, we propose and prove a converse bound of the $\epsilon$-attainable messages-size pairs for bipartite broadcast channel simulation.
Similar to the point-to-point case, we utilize the non-lockability lemma (Lemma~\ref{lem:non-lockability}).
Specifically for the broadcast setup, we develop and use the following bipartite generalization of Lemma~\ref{lem:non-lockability}.

Given discrete random variables $\rv{X}\rv{Y}\rv{Z}$ with joint pmf $p_\rv{XYZ}\in\set{P}(\set{X}\times\set{Y}\times\set{Z})$, the max-mutual information of $\rv{X}$ \vs $\rv{Y}$ and $\rv{Z}$ is defined as
\begin{equation}
    I_{\max}(\rv{X};\rv{Y};\rv{Z}) \defeq \inf_{\dscript{q_\rv{Y}\in\set{P}(\set{Y})}{r_\rv{Z}\in\set{P}(\set{Z})}} D_{\max}\infdiv*{p_{\rv{XYZ}}}{p_\rv{X}\times q_\rv{Y} \times r_\rv{Z}}.
\end{equation}
where $p_\rv{X}$ is the marginal distribution of $\rv{X}$ induced from $p_\rv{XYZ}$.
\begin{lemma}[Bipartite generalization of Lemma~\ref{lem:non-lockability}]\label{lem:common:non-lockability}
Let $(\rv{X}, \rv{Y}, \tilde{\rv{Y}}, \rv{Z}, \tilde{\rv{Z}})$ be joint random variables distributed on $\set{X}\times \set{Y}\times \tilde{\set{Y}}\times \set{Z}\times \tilde{\set{Z}}$.
Suppose that all the sets involved above are finite.
Then
\begin{equation}
    I_{\max}(\rv{X};\rv{Y}\tilde{\rv{Y}};\rv{Z}\tilde{\rv{Z}}) \leq I_{\max}(\rv{X};\rv{Y};\rv{Z}) + \log{\size{\tilde{\set{Y}}}} + \log{\size{\tilde{\set{Z}}}}.
\end{equation}
\end{lemma}
\begin{proof}
Let $q_\rv{Y}^\star$ and $r_\rv{Y}^\star$ denote the optimal distributions on $\set{Y}$ and $\set{Z}$, respectively, which achieve the infimum in the definition of $I_{\max}(\rv{X};\rv{Y};\rv{Z})$.
Then
\begin{align}
    I_{\max}(\rv{X};\rv{Y}\tilde{\rv{Y}};\rv{Z}\tilde{\rv{Z}})
	= & \inf_{\dscript{q_{\rv{Y}\tilde{\rv{Y}}}\in\set{P}(\set{Y}\times\tilde{\set{Y}})}{r_{\rv{Z}\tilde{\rv{Z}}}\in\set{P}(\set{Z}\times\tilde{\set{Z}})}} D_{\max}\infdiv*{p_{\rv{X}\rv{Y}\tilde{\rv{Y}}\rv{Z}\tilde{\rv{Z}}}}{p_\rv{X}\times q_{\rv{Y}\tilde{\rv{Y}}} \times r_{\rv{Z}\tilde{\rv{Z}}}} \\
	\leq & D_{\max}\infdiv*{p_{\rv{X}\rv{Y}\tilde{\rv{Y}}\rv{Z}\tilde{\rv{Z}}}}{p_\rv{X}\times q_\rv{Y}^\star \times \frac{1}{\size{\tilde{\set{Y}}}}\times r_\rv{Z}^\star \times \frac{1}{\size{\tilde{\set{Z}}}}} \\
	= & \log{\left(\size{\tilde{\set{Y}}}\!\cdot\!\size{\tilde{\set{Z}}}\!\cdot\! \sup_{x,y,z,\tilde{y},\tilde{z}} \frac{p_\rv{XYZ}(x,y,z)\cdot p_{\tilde{\rv{Y}}\tilde{\rv{Z}}|\rv{XYZ}}(\tilde{y},\tilde{z}|x,y,z)}{p_\rv{X}(x)\cdot q_\rv{Y}^\star(y) \cdot r_\rv{Z}^\star(z)}\right)} \\
	\leq & \log{\sup_{x,y,z} \frac{p_\rv{XYZ}(x,y,z)}{p_\rv{X}(x)\cdot q_\rv{Y}^\star(y) \cdot r_\rv{Z}^\star(z)}} + \log{\size{\tilde{\set{Y}}}} + \log{\size{\tilde{\set{Z}}}}\\
	= & I_{\max}(\rv{X};\rv{Y};\rv{Y}) + \log{\size{\tilde{\set{Y}}}} + \log{\size{\tilde{\set{Z}}}}. \qedhere
\end{align}
\end{proof}
\begin{proposition}\label{prop:1cbfsbc}
Let $W_{\rv{YZ}|\rv{X}}$ be a channel from $\set{X}$ to $\set{Y}\times\set{Z}$, and let $\epsilon\in(0,1)$.
For any $\delta_1$, $\delta_2$, $\delta_3\in(0,1-\epsilon)$, the following set is a superset of $\set{M}^\star_\epsilon(W_{\rv{YZ}|\rv{X}})$
\begin{equation}
\set{M}^\text{out}_\epsilon(W_{\rv{YZ}|\rv{X}}) \defeq \left\{(M,N)\in\mathbb{Z}_{>0}^2 \middle\vert
\begin{aligned}
\log{M} &\geq \adjustlimits\inf_{q_\rv{Y}} \sup_{p_\rv{X}} D_{s+}^{\epsilon+\delta_1}\infdiv*{p_\rv{X}\cdot W_{\rv{Y}|\rv{X}}}{p_\rv{X}\times q_\rv{Y}} + \log{\delta_1} \\
\log{N} &\geq \adjustlimits\inf_{r_\rv{Z}} \sup_{p_\rv{X}} D_{s+}^{\epsilon+\delta_2}\infdiv*{p_\rv{X}\cdot W_{\rv{Z}|\rv{X}}}{p_\rv{X}\times r_\rv{Z}} + \log{\delta_2} \\
\log{MN}&\geq \adjustlimits\inf_{q_\rv{Y},r_\rv{Z}} \sup_{p_\rv{X}} D_{s+}^{\epsilon+\delta_3}\infdiv*{p_\rv{X}\cdot W_{\rv{YZ}|\rv{X}}}{p_\rv{X}\times q_\rv{Y} \times r_\rv{Z}} + \log{\delta_3}
\end{aligned}
\right\}.
\end{equation}
\end{proposition}
\begin{proof}
Let $(M,N)\in\set{M}^\star_\epsilon(W_{\rv{YZ}|\rv{X}})$, \ie, suppose there exists a size-$(M,N)$ $\epsilon$-simulation code for $W_{\rv{YZ}|\rv{X}}$.
Let $\rv{S}$ and $\rv{S}'$ denote the two shared randomness between the sender and the first and the second receivers, respectively.
Let $\rv{M}$ and $\rv{N}$ denote the two codewords transmitted from the sender to the first and the second receivers, respectively.
Then, for any input source $\rv{X}\sim p_\rv{X}$, we have a Markov chain $\rv{Z}-\rv{N}\rv{S}'-\rv{X}-\rv{MS}-\rv{Y}$ where
\begin{itemize}
    \item $\rv{X}$, $\rv{S}$, and $\rv{S}'$ are independent.
    \item The distribution of $\rv{XYZ}$, denoted by $\tilde{p}_{\rv{XYZ}}$, is $\epsilon$-close (in TVD) to $p_{\rv{XYZ}}\defeq p_\rv{X} \cdot W_{\rv{YZ}|\rv{X}}$.
    \item The marginal distribution $\sum_{y,z}\tilde{p}_{\rv{XYZ}}(x,y,z) = p_\rv{X}(x)$ $\forall x\in\set{X}$.
    \item $\rv{M}$ and $\rv{N}$ are distributed over $\{1,\ldots,M\}$ and $\{1,\ldots,N\}$, respectively.
\end{itemize}
Pick $p_\rv{X}$ to be some pmf with full support.
The following statements hold.
\begin{enumerate}
    \item By Lemma~\ref{lem:non-lockability} and Lemma~\ref{lem:common:non-lockability}, we have  
    \begin{align}
    \log{M} &= I_{\max}(\rv{X};\rv{S}) + \log{M} \geq I_{\max}(\rv{X};\rv{MS}); \\
    \log{N} &= I_{\max}(\rv{X};\rv{S'}) + \log{N} \geq I_{\max}(\rv{X};\rv{NS}'); \\
    \log{M} + \log{N} &= I_{\max}(\rv{X};\rv{S};\rv{S'}) + \log{M} + \log{N} \geq I_{\max}(\rv{X};\rv{MS};\rv{NS}') .
    \end{align}
    \item By the data processing inequality of $I_{\max}$, we have
    \begin{align}
    I_{\max}(\rv{X};\rv{MS}) &\geq I_{\max}(\rv{X};\rv{Y}) \\
    I_{\max}(\rv{X};\rv{NS}') &\geq I_{\max}(\rv{X};\rv{Z}) \\
    I_{\max}(\rv{X};\rv{MS};\rv{NS}') &\geq I_{\max}(\rv{X};\rv{Y};\rv{Z}) \, . 
    \end{align}
    \item By the definition of $I_{\max}$, and noting that $\tilde{p}_\rv{X}=p_\rv{X}$, we have
    \begin{align}
     I_{\max}(\rv{X};Y) &= \inf_{q_{Y}} D_{\max}\infdiv*{\tilde{p}_{\rv{XY}}}{p_{\rv{X}}\times q_{\rv{Y}}} \\
     &= \adjustlimits \inf_{q_{\rv{Y}}} \max_{x} D_{\max}\infdiv*{\tilde{p}_{\rv{Y}|\rv{X}}(\cdot|x)}{q_\rv{Y}}\\
     &\geq \multiadjustlimits{
         \inf_{q_{\rv{Y}}}, 
         \inf_{\tilde{W}_{\rv{Y}|\rv{X}}: \norm{\tilde{W}_{\rv{Y}|\rv{X}} - W_{\rv{Y}|\rv{X}}}_\fnc{tvd}\leq\epsilon}, 
         \max_{x}} D_{\max}\infdiv*{\tilde{W}_{\rv{Y}|\rv{X}}(\cdot|x)}{q_\rv{Y}}\\
     I_{\max}(\rv{X};\rv{Z}) &= \inf_{r_{\rv{Z}}} D_{\max}\infdiv*{\tilde{p}_{\rv{XZ}}}{p_{\rv{X}}\times r_{\rv{Z}}} \\
     &= \adjustlimits \inf_{r_{\rv{Z}}} \max_{x} D_{\max}\infdiv*{\tilde{p}_{\rv{Z}|\rv{X}}(\cdot|x)}{r_\rv{Z}}\\
     &\geq \multiadjustlimits{
         \inf_{r_{\rv{Z}}}, 
         \inf_{\tilde{W}_{\rv{Z}|\rv{X}}: \norm{\tilde{W}_{\rv{Z}|\rv{X}}-W_{\rv{Z}|\rv{X}}}_\fnc{tvd}\leq\epsilon}, 
         \max_{x}} D_{\max}\infdiv*{\tilde{W}_{\rv{Z}|\rv{X}}(\cdot|x)}{r_\rv{Z}}\\
     I_{\max}(\rv{X};\rv{Y};\rv{Z}) &= \adjustlimits\inf_{q_{\rv{Y}}}\inf_{r_{\rv{Z}}} D_{\max}\infdiv*{\tilde{p}_{\rv{XYZ}}}{p_{\rv{X}}\times q_{\rv{Y}}\times r_{\rv{Z}}} \\
     &= \inf_{q_{\rv{Y}}}\inf_{r_{\rv{Z}}} \max_{x} D_{\max}\infdiv*{\tilde{p}_{\rv{YZ}|\rv{X}}(\cdot|x)}{q_\rv{Y}\times r_\rv{Z}}\\
     &\geq \multiadjustlimits{
         \inf_{q_{\rv{Y}}}, 
         \inf_{r_{\rv{Z}}}, 
         \inf_{\tilde{W}_{\rv{YZ}|\rv{X}}: \norm{\tilde{W}_{\rv{YZ}|\rv{X}}-W_{\rv{YZ}|\rv{X}}}_\fnc{tvd}\leq\epsilon}, 
         \max_{x}} D_{\max}\infdiv*{\tilde{W}_{\rv{YZ}|\rv{X}}(\cdot|x)}{q_\rv{Y}\times r_\rv{Z}}.
    \end{align}
\end{enumerate}
The proposition can be proven by combining the above three steps and Lemma~\ref{lem:maxDmax:supDs}.
\end{proof}

\subsection{Asymptotic Analysis}

In this section, we consider the task of simulating $W_{\rv{YZ}|\rv{X}}^{\tensor n}$.
In asymptotic discussions, one is usually more interested in admissible rates rather than admissible message sizes.
For our task of simulating $W_{\rv{YZ}|\rv{X}}$ asymptotically, a rate pair $(r_1,r_2)$ (of positive real numbers) is said to be $\epsilon$-attainable if there exists a sequence of size-$(\floor{2^{nr_1}}, \floor{2^{nr_2}})$ $\epsilon$-simulation codes for $W_{\rv{YZ}|\rv{X}}^{\tensor n}$ for $n$ sufficiently large.
We denote  $\set{R}^\star_\epsilon(W_{\rv{YZ}|\rv{X}})$ the \emph{closure} of the set of all $\epsilon$-attainable rate pairs, \ie, 
\begin{equation}
\set{R}^\star_\epsilon(W_{\rv{YZ}|\rv{X}}) \defeq \fnc{cl}\left(\left\{(r_1,r_2)\in\mathbb{R}_{\geq 0}^2\middle\vert \exists N\in\mathbb{N}\text{ s.t. }(\floor{2^{nr_1}},\floor{2^{nr_2}})\in\set{M}_\epsilon^\star(W_{\rv{YZ}|\rv{X}}^{\tensor n})\text{ for all }n\geq N\right\}\right)
\end{equation}

We have the following single-letter expression of the (random-assisted) simulation rate region of a bipartite broadcast channel, which quantifies the ``minimal'' amount of channel pairs needed to simulate the statistics of a bipartite broadcast channel in asymptotic setups with a fixed (or diminishing) tolerance.
An implication of Theorem~\ref{thm:broadcast:asymptotic} below is the lack of interconvertibility between broadcast channels.
Consider the following example of two bipartite broadcast channels $V\in\set{P}(\{0,1\}\times\{0,1\}|\{0,1\}^2)$ and $W\in\set{P}(\{0,1\}\times\{0,1\}|\{0,1\})$, where
\begin{align}
    V(y,z|x_0,x_1) &= \ID(y|x_0) \cdot \ID(z|x_1), \\
    W(y,z|x) &= \ID(y|x) \cdot \ID(z|x). \label{eq:extreme:bc}
\end{align}
As $V$ is a product channel, we know its (coding) capacity region to be $[1,\infty)^2$ (in bits per channel use), which matches $\set{R}^\star_\epsilon(W)$ provided by Theorem~\ref{thm:broadcast:asymptotic} exactly.
In other words, we need one copy of $V$ per $W$ simulated asymptotically.
On the other hand, one copy of $W$ cannot simulate one copy of $V$ asymptotically, as the capacity region of $W$ is $\{(r_1,r_2)\in\mathbb{R}_{\geq 0}^2|r_1+r_2\leq 1\}$ while the simulation region of $V$ is $\set{R}^\star_\epsilon(V)=[1,\infty)^2$.
\begin{theorem}\label{thm:broadcast:asymptotic}
Let $W_{\rv{YZ}|\rv{X}}$ be a channel from $\set{X}$ to $\set{Y}\times\set{Z}$.
For any $\epsilon\in(0,1)$, it holds that
\begin{equation} \label{eq:broadcast:asymptotic}
\set{R}^\star_\epsilon(W_{\rv{YZ}|\rv{X}}) = \left\{(r_1,r_2)\in\mathbb{R}_{\geq 0}^2 \middle\vert
\begin{aligned}
    r_1 &\geq C(W_{\rv{Y}|\rv{X}})\\
    r_2 &\geq C(W_{\rv{Z}|\rv{X}})\\
    r_1+r_2 &\geq \tilde{C}(W_{\rv{YZ}|\rv{X}})
\end{aligned}\right\},
\end{equation}
where $\tilde{C}(W_{\rv{YZ}|\rv{X}})$ has been defined in~\eqref{eq:def:tilde:C}.
\end{theorem}
\begin{proof}[Achievability Proof of Theorem~\ref{thm:broadcast:asymptotic}]
Applying Proposition~\ref{prop:1abfsbc}, we have the following set being a subset of $\set{M}^\star_\epsilon(W_{\rv{YZ}|\rv{X}}^{\tensor n})$
    \begin{equation}\label{eq:innner:broadcast:n}
    \set{M}^\text{in}_\epsilon(W_{\rv{YZ}|\rv{X}}^{\tensor n}) \defeq \left\{ (M,N)\in\mathbb{Z}_{>0}^2 \middle\vert 
    \begin{aligned}
    \log{M} &> D_{s+}^{\epsilon_1-\delta_1}\infdiv*{p_{\rvs{X}_1^n}\cdot W_{\rv{Y}|\rv{X}}^{\tensor n}}{p_{\rvs{X}_1^n}\times q_{\rvs{Y}_1^n}} - \log{\delta_1^2} \\
    \log{N} &> D_{s+}^{\epsilon_2-\delta_2}\infdiv*{p_{\rvs{X}_1^n}\cdot W_{\rv{Z}|\rv{X}}^{\tensor n}}{p_{\rvs{X}_1^n}\times r_{\rvs{Z}_1^n}} - \log{\delta_2^2} \\
    \log{MN} &> D_{s+}^{\epsilon_3-\delta_3}\infdiv*{p_{\rvs{X}_1^n}\cdot W_{\rv{YZ}|\rv{X}}^{\tensor n}}{p_{\rvs{X}_1^n}\times q_{\rvs{Y}_1^n}\times r_{\rvs{Z}_1^n}} - \log{\delta_3^2}
    \end{aligned} \forall p_{\rvs{X}_1^n} \in\set{P}(\set{X}^n)
    \right\}.
    \end{equation}
for any $\epsilon_1,\epsilon_2,\epsilon_3>0$ such that $\epsilon_1+\epsilon_2+\epsilon_3\leq\epsilon$, $\delta_1\in(0,\epsilon_1)$, $\delta_2\in(0,\epsilon_2)$, $\delta_3\in(0,\epsilon_3)$, $q_{\rvs{Y}_1^n}\in\set{P}(\set{Y}^n)$, and $r_{\rvs{Z}_1^n}\in\set{P}(\set{Z}^n)$.
We pick $q_{\rvs{Y}_1^n}$ and $r_{\rvs{Z}_1^n}$ as 
\begin{align}
q_{\rvs{Y}_1^n} \defeq \sum_{\lambda\in\Lambda_n}\frac{1}{\size{\Lambda_n}}\left(\sum_{\tilde{x},\tilde{z}}W_{\rv{YZ}|\rv{X}}(\cdot,\tilde{z}|\tilde{x})\cdot p_\lambda(\tilde{x})\right)^{\tensor n},\qquad 
r_{\rvs{Z}_1^n} \defeq \sum_{\lambda\in\Lambda_n}\frac{1}{\size{\Lambda_n}}\left(\sum_{\tilde{x},\tilde{y}}W_{\rv{YZ}|\rv{X}}(\tilde{y},\cdot|\tilde{x})\cdot p_\lambda(\tilde{x})\right)^{\tensor n}.
\end{align}
We have the following chain of inequalities.
\begin{align}
& \adjustlimits\inf_{q_{\rvs{Y}_1^n}, r_{\rvs{Z}_1^n}} \sup_{p_{\rvs{X}_1^n}} \frac{1}{n}D_{s+}^{\epsilon_3-\delta_3}\infdiv*{p_{\rvs{X}_1^n}\cdot W_{\rv{YZ}|\rv{X}}^{\tensor n}}{p_{\rvs{X}_1^n}\times q_{\rvs{Y}_1^n}\times r_{\rvs{Z}_1^n}} - \frac{1}{n}\log{\delta_3^2} \nonumber\\
\label{eq:inner:broadcast:asymptotic:a}
\leq & \adjustlimits\inf_{q_{\rvs{Y}_1^n}, r_{\rvs{Z}_1^n}} \sup_{\mathbf{x}_1^n} \frac{1}{n}D_{s+}^{\epsilon_3-\delta_3}\infdiv*{W_{\rv{YZ}|\rv{X}}^{\tensor n}(\cdot|\mathbf{x}_1^n)}{q_{\rvs{Y}_1^n}\times r_{\rvs{Z}_1^n}} - \frac{1}{n}\log{\delta_3^2} \\
\label{eq:inner:broadcast:asymptotic:b}
\leq & \begin{aligned}[t] \sup_{\mathbf{x}_1^n} \frac{1}{n}D_{s+}^{\epsilon_3-\delta_3}\infdiv*{W_{\rv{YZ}|\rv{X}}^{\tensor n}(\cdot|\mathbf{x}_1^n)}{\sum_{\lambda\in\Lambda_n}\frac{1}{\size{\Lambda_n}}\left(\!\sum_{\tilde{x},\tilde{z}}W_{\rv{YZ}|\rv{X}}(\cdot,\tilde{z}|\tilde{x})\cdot p_\lambda(\tilde{x})\!\right)^{\tensor n} \!\!\!\!\!\!\times\! \sum_{\lambda\in\Lambda_n}\frac{1}{\size{\Lambda_n}}\left(\!\sum_{\tilde{x},\tilde{y}}W_{\rv{YZ}|\rv{X}}(\tilde{y},\cdot|\tilde{x})\cdot p_\lambda(\tilde{x})\!\right)^{\tensor n}} \\ - \frac{1}{n}\log{\delta_3^2}\end{aligned} \\
\label{eq:inner:broadcast:asymptotic:c}
\leq & \begin{aligned}[t] \sup_{\mathbf{x}_1^n} \frac{1}{n} D_{s+}^{\epsilon_3-\delta_3}\infdiv*{W_{\rv{YZ}|\rv{X}}^{\tensor n}(\cdot|\mathbf{x}_1^n)}{\left(\sum_{\tilde{x},\tilde{z}}W_{\rv{YZ}|\rv{X}}(\cdot,\tilde{z}|\tilde{x})\cdot f_{\mathbf{x}_1^n}(\tilde{x})\right)^{\tensor n} \times \left(\sum_{\tilde{x},\tilde{y}}W_{\rv{YZ}|\rv{X}}(\tilde{y},\cdot|\tilde{x})\cdot f_{\mathbf{x}_1^n}(\tilde{x})\right)^{\tensor n}} \\ - \frac{1}{n}\log{\delta_3^2} +\frac{2}{n}\log{\size{\Lambda_n}} \end{aligned} 
\end{align}
\begin{align}
\label{eq:inner:broadcast:asymptotic:d}
\leq & \begin{aligned}[t]\sup_{\mathbf{x}_1^n} \begin{aligned}[t]&\left\{\frac{1}{n} \sum_{i=1}^n D\infdiv*{W_{\rv{YZ}|\rv{X}}(\cdot|x_i)}{\left(\sum_{\tilde{x},\tilde{z}}W_{\rv{YZ}|\rv{X}}(\cdot,\tilde{z}|\tilde{x})\cdot f_{\mathbf{x}_1^n}(\tilde{x})\right) \times \left(\sum_{\tilde{x},\tilde{y}}W_{\rv{YZ}|\rv{X}}(\tilde{y},\cdot|\tilde{x})\cdot f_{\mathbf{x}_1^n}(\tilde{x})\right)}\right.\\
    &\left.+ \frac{1}{\sqrt{n(\epsilon_3-\delta_3)}}\cdot \sqrt{\frac{1}{n} \sum_{i=1}^n V\infdiv*{W_{\rv{YZ}|\rv{X}}(\cdot|x_i)}{\left(\sum_{\tilde{x},\tilde{z}}W_{\rv{YZ}|\rv{X}}(\cdot,\tilde{z}|\tilde{x})\cdot f_{\mathbf{x}_1^n}(\tilde{x})\right) \times \left(\sum_{\tilde{x},\tilde{y}}W_{\rv{YZ}|\rv{X}}(\tilde{y},\cdot|\tilde{x})\cdot f_{\mathbf{x}_1^n}(\tilde{x})\right)}}\right\}
    \end{aligned}\\
    - \frac{1}{n}\log{\delta_3^2} +\frac{2}{n}\log{\size{\Lambda_n}}
    \end{aligned}\!\!\\
\label{eq:inner:broadcast:asymptotic:e}
\leq & \sup_{p_\rv{X}} I(\rv{X};\rv{Y};\rv{Z})_{p_\rv{X}\cdot W_{\rv{YZ}|\rv{X}}} + \frac{1}{\sqrt{n(\epsilon_3-\delta_3)}} \cdot \tilde{V}(p_\rv{X}) - \frac{1}{n}\log{\delta_3^2} +\frac{2}{n}\log{\size{\Lambda_n}}
\end{align}
where for~\eqref{eq:inner:broadcast:asymptotic:a} we use the quasi-convexity of $D_{s+}^\epsilon\infdiv{p_\rv{X}\cdot p_{\rv{Y}|\rv{X}}}{p_\rv{X}\cdot q_{\rv{Y}|\rv{X}}}$ in $p_\rv{X}$, for~\eqref{eq:inner:broadcast:asymptotic:b} we pick a pair of specific $q_{\rvs{Y}_1^n}$ and $r_{\rvs{Z}_1^n}$ as aforementioned to upper bound the infimum, for~\eqref{eq:inner:broadcast:asymptotic:c} we use~\cite[Lemma~3]{tomamichel2013tight}, for~\eqref{eq:inner:broadcast:asymptotic:d} we use the Chebyshev-type bound in~\cite[Lemma~5]{tomamichel2013tight}, and~\eqref{eq:inner:broadcast:asymptotic:e} is a result of direct counting and the definition of the \emph{common dispersion} $\tilde{V}$ as
\begin{equation}\label{eq:def:tV}
    \tilde{V}(p) \defeq \sum_{x} p(x) \cdot V\Bigg( W_{\rv{YZ}|\rv{X}}(\cdot,\cdot|x) \Bigg\Vert \sum_{\tilde{x},\tilde{z}}p(x)\cdot W_{\rv{YZ}|\rv{X}}(\cdot,\tilde{z}|\tilde{x}) \times \sum_{\tilde{x},\tilde{y}}p(x)\cdot W_{\rv{YZ}|\rv{X}}(\tilde{y},\cdot|\tilde{x})\Bigg).
\end{equation}
Notice that $\tilde{V}$ is bounded.
Thus, it holds that
\begin{equation}\label{eq:innner:broadcast:n:asymptotic:A}
\limsup_{n\to\infty} \adjustlimits\inf_{q_{\rvs{Y}_1^n}, r_{\rvs{Z}_1^n}} \sup_{p_{\rvs{X}_1^n}} \frac{1}{n}D_{s+}^{\epsilon_3-\delta_3}\infdiv*{p_{\rvs{X}_1^n}\cdot W_{\rv{YZ}|\rv{X}}^{\tensor n}}{p_{\rvs{X}_1^n}\times q_{\rvs{Y}_1^n}\times r_{\rvs{Z}_1^n}} - \frac{1}{n}\log{\delta_3^2} \leq \sup_{p_\rv{X}} I(\rv{X};\rv{Y};\rv{Z})_{p_\rv{X}\cdot W_{\rv{YZ}|\rv{X}}} = \tilde{C}(W_{\rv{YZ}|\rv{X}}).
\end{equation}
Similarly, one can show
\begin{align}
\label{eq:innner:broadcast:n:asymptotic:B}
\limsup_{n\to\infty} \adjustlimits\inf_{q_{\rvs{Y}_1^n}} \sup_{p_{\rvs{X}_1^n}} \frac{1}{n}D_{s+}^{\epsilon_1-\delta_1}\infdiv*{p_{\rvs{X}_1^n}\cdot W_{\rv{Y}|\rv{X}}^{\tensor n}}{p_{\rvs{X}_1^n}\times q_{\rvs{Y}_1^n}} - \frac{1}{n}\log{\delta_1^2} \leq C(W_{\rv{Y}|\rv{X}}),\\
\label{eq:innner:broadcast:n:asymptotic:C}
\limsup_{n\to\infty} \adjustlimits\inf_{r_{\rvs{Z}_1^n}} \sup_{p_{\rvs{X}_1^n}} \frac{1}{n}D_{s+}^{\epsilon_2-\delta_2}\infdiv*{p_{\rvs{X}_1^n}\cdot W_{\rv{Z}|\rv{X}}^{\tensor n}}{p_{\rvs{X}_1^n}\times r_{\rvs{Z}_1^n}} - \frac{1}{n}\log{\delta_2^2} \leq C(W_{\rv{Z}|\rv{X}}).
\end{align}

Combining~\eqref{eq:innner:broadcast:n:asymptotic:A}, \eqref{eq:innner:broadcast:n:asymptotic:B}, and \eqref{eq:innner:broadcast:n:asymptotic:C} with~\eqref{eq:innner:broadcast:n}, it is straightforward to check that any integer pair $(\floor{2^{nr_1}}, \floor{2^{nr_2}})$ with 
\begin{equation}\label{eq:inner:broadcast:asymptotic:2}
    (r_1, r_2) \in \left\{(r_1,r_2)\in\mathbb{R}_{\geq 0}^2 \middle\vert
    \begin{aligned}
    r_1 &> C(W_{\rv{Y}|\rv{X}})\\
    r_2 &> C(W_{\rv{Z}|\rv{X}})\\
    r_1+r_2 &> \tilde{C}(W_{\rv{YZ}|\rv{X}})
    \end{aligned}
    \right\},
\end{equation}
must be in $\set{M}^\text{in}_\epsilon(W_{\rv{YZ}|\rv{X}}^{\tensor n})$ for $n$ sufficiently large, \ie, RHS of~\eqref{eq:inner:broadcast:asymptotic:2}$\subset\set{R}^\star_\epsilon(W_{\rv{YZ}|\rv{X}})$.
This proves the RHS of~\eqref{eq:broadcast:asymptotic} being a subset of $\set{R}^\star_\epsilon(W_{\rv{YZ}|\rv{X}})$ since the latter is a closed set.
\end{proof}
For the converse argument, we need the following generalized version of~\cite[Lemma~10]{anshu2020partially}.
\begin{lemma}\label{lem:Ds:2}
For any $p_\rv{XY}\in\set{P}(\set{X}\times\set{Y}\times\set{Z})$, $\epsilon\in(0,1)$, and $\delta\in(0,\frac{1-\epsilon}{2})$, it holds that 
\begin{equation}
\inf_{q_\rv{Y}\in\set{P}(\set{Y}), r_\rv{Z}\in\set{P}(\set{Z})} D_{s+}^{\epsilon}\infdiv*{p_\rv{XYZ}}{p_\rv{X}\times q_\rv{Y}\times r_\rv{Z}} \geq D_{s+}^{\epsilon+2\delta}\infdiv*{p_\rv{XY}}{p_\rv{X}\times p_\rv{Y}\times p_\rv{Z}} + 2\log{\delta}.
\end{equation}
\end{lemma}
\begin{proof}
Let $a^\star= \inf_{q_\rv{Y}, r_\rv{Z}}D_{s+}^\epsilon\infdiv{p_\rv{XYZ}}{p_\rv{X}\times q_\rv{Y}\times r_\rv{Z}}\geq 0$.
There must exists some pmf $q^\star_\rv{Y}$ and $r^\star_\rv{Z}$ such that $p_\rv{XYZ}(\set{A})\leq \epsilon$ where
\begin{equation}
\set{A} \defeq \left\{(x,y,z)\in\set{X}\times\set{Y}\times\set{Z}: \log{\frac{p_\rv{XYZ}(x,y,z)}{p_\rv{X}(x)\cdot q_\rv{Y}^\star(y)\cdot r_\rv{Z}^\star(z)}} > a^\star\right\}.
\end{equation}
Define the sets 
\begin{align}
\set{B} &\defeq \left\{(x,y,z)\in\set{X}\times\set{Y}\times\set{Z}: q_\rv{Y}^\star(y)> \frac{1}{\delta} p_\rv{Y}(y)\right\},\\
\set{C} &\defeq \left\{(x,y,z)\in\set{X}\times\set{Y}\times\set{Z}: r_\rv{Z}^\star(z)> \frac{1}{\delta} p_\rv{Z}(z)\right\}.
\end{align}
We have 
\begin{align}
p_\rv{XYZ}(\set{B}) &= \sum_{y: q_\rv{Y}^\star(y) > \frac{1}{\delta} p_\rv{Y}(y)} p_\rv{Y}(y) < \sum_{y} \delta\cdot q_\rv{Y}^\star(y) = \delta, \\
p_\rv{XYZ}(\set{C}) &= \sum_{z: r_\rv{Z}^\star(z) > \frac{1}{\delta} p_\rv{Z}(z)} p_\rv{Z}(z) < \sum_{z} \delta\cdot r_\rv{Z}^\star(z) = \delta.
\end{align}
Notice that for all $(x,y,z)\not\in\set{A}\cup\set{B}\cup\set{C}$, we have
\begin{equation}
p_\rv{XYZ}(x,y,z) \leq 2^{a^\star}\cdot p_\rv{X}(x)\cdot q_\rv{Y}^\star(y) \cdot r_\rv{Z}^\star(z) \leq \frac{2^{a^\star}}{\delta^2} p_\rv{X}(x)\cdot p_\rv{Y}(y) \cdot p_\rv{Z}(z).
\end{equation}
Thus,
\begin{align}
p_\rv{XYZ}\left(\left\{(x,y,z): \log{\frac{p_\rv{XYZ}(x,y,z)}{p_\rv{X}(x)\cdot p_\rv{Y}(y)\cdot p_\rv{Z}(z)}} > a^\star-2\log{\delta} \right\}\right) 
&\leq p_\rv{XYZ}(\set{A}\cup\set{B}\cup\set{C})\\
&\leq p_\rv{XYZ}(\set{A}) + p_\rv{XYZ}(\set{B}) + p_\rv{XYZ}(\set{C})< \epsilon+2\delta.
\end{align}
Since $a^\star-2\log{\delta} \geq 0$, we know $D_{s+}^{\epsilon+2\delta}\infdiv*{p_\rv{XY}}{p_\rv{X}\times p_\rv{Y}} \leq a^\star-2\log{\delta}$.
\end{proof}
\begin{proof}[Converse Proof of Theorem~\ref{thm:broadcast:asymptotic}]
Let $(r_1,r_2)$ be arbitrarily pair of non-negative numbers such that $(2^{\floor{nr_1}},2^{\floor{nr_2}})\in\set{M}_\epsilon^\star(W_{\rv{YZ}|\rv{X}}^{\tensor n})$ for $n$ sufficiently large, \ie, $(r_1,r_2)$ is an arbitrarily \emph{interior} point of $\set{R}^\star_\epsilon(W_{\rv{YZ}|\rv{X}})$.
Applying Proposition~\ref{prop:1cbfsbc}, we have the following set being a superset of $\set{M}^\star_\epsilon(W_{\rv{YZ}|\rv{X}}^{\tensor n})$
    \begin{equation}\label{eq:outer:broadcast:n}
    \set{M}^\text{out}_\epsilon(W_{\rv{YZ}|\rv{X}}^{\tensor n}) \defeq \left\{(M,N)\in\mathbb{Z}_{>0}^2 \middle\vert
    \begin{aligned}
    \log{M} &\geq \adjustlimits\inf_{q_{\rvs{Y}_1^n}} \sup_{p_{\rvs{X}_1^n}} D_{s+}^{\epsilon+\delta_1}\infdiv*{p_{\rvs{X}_1^n}\cdot W_{\rv{Y}|\rv{X}}^{\tensor n}}{p_{\rvs{X}_1^n}\times q_{\rvs{Y}_1^n}} + \log{\delta_1} \\
    \log{N} &\geq \adjustlimits\inf_{r_{\rvs{Z}_1^n}} \sup_{p_{\rvs{X}_1^n}} D_{s+}^{\epsilon+\delta_2}\infdiv*{p_{\rvs{X}_1^n}\cdot W_{\rv{Z}|\rv{X}}^{\tensor n}}{p_{\rvs{X}_1^n}\times r_{\rvs{Z}_1^n}} + \log{\delta_2} \\
    \log{MN}&\geq \adjustlimits\inf_{q_{\rvs{Y}_1^n},r_{\rvs{Z}_1^n}} \sup_{p_{\rvs{X}_1^n}} D_{s+}^{\epsilon+\delta_3}\infdiv*{p_{\rvs{X}_1^n}\cdot W_{\rv{YZ}|\rv{X}}^{\tensor n}}{p_{\rvs{X}_1^n}\times q_{\rvs{Y}_1^n} \times r_{\rvs{Z}_1^n}} + \log{\delta_3}
    \end{aligned}
    \right\}.
    \end{equation}
for any $\delta_1$, $\delta_2$, $\delta_3\in(0,1-\epsilon)$.
Thus, we have
\begin{align}
    \label{eq:outer:broadcast:asymptotic:a}
    r_1 &\geq \adjustlimits\inf_{q_{\rvs{Y}_1^n}} \sup_{p_{\rvs{X}_1^n}} \frac{1}{n} D_{s+}^{\epsilon+\delta_1}\infdiv*{p_{\rvs{X}_1^n}\cdot W_{\rv{Y}|\rv{X}}^{\tensor n}}{p_{\rvs{X}_1^n}\times q_{\rvs{Y}_1^n}} + \frac{1}{n}\log{\delta_1} \\
    \label{eq:outer:broadcast:asymptotic:b}
    r_2 &\geq \adjustlimits\inf_{r_{\rvs{Z}_1^n}} \sup_{p_{\rvs{X}_1^n}} \frac{1}{n} D_{s+}^{\epsilon+\delta_2}\infdiv*{p_{\rvs{X}_1^n}\cdot W_{\rv{Z}|\rv{X}}^{\tensor n}}{p_{\rvs{X}_1^n}\times r_{\rvs{Z}_1^n}} + \frac{1}{n}\log{\delta_2} \\
    \label{eq:outer:broadcast:asymptotic:c}
    r_1 + r_2 &\geq \adjustlimits\inf_{q_{\rvs{Y}_1^n},r_{\rvs{Z}_1^n}} \sup_{p_{\rvs{X}_1^n}} \frac{1}{n} D_{s+}^{\epsilon+\delta_3}\infdiv*{p_{\rvs{X}_1^n}\cdot W_{\rv{YZ}|\rv{X}}^{\tensor n}}{p_{\rvs{X}_1^n}\times q_{\rvs{Y}_1^n} \times r_{\rvs{Z}_1^n}} + \frac{1}{n}\log{\delta_3}
\end{align}
for $n$ sufficiently large.
By Lemma~\ref{lem:Ds:2}, we can rewrite~\eqref{eq:outer:broadcast:asymptotic:c} as
\begin{align}
r_1 + r_2 &\geq \adjustlimits \sup_{p_{\rvs{X}_1^n}} \inf_{q_{\rvs{Y}_1^n},r_{\rvs{Z}_1^n}} \frac{1}{n} D_{s+}^{\epsilon+\delta_3}\infdiv*{p_{\rvs{X}_1^n}\cdot W_{\rv{YZ}|\rv{X}}^{\tensor n}}{p_{\rvs{X}_1^n}\times q_{\rvs{Y}_1^n} \times r_{\rvs{Z}_1^n}} + \frac{1}{n}\log{\delta_3} \\
\label{eq:outer:broadcast:asymptotic:d}
 &\geq \sup_{p_{\rvs{X}_1^n}} \frac{1}{n} D_{s+}^{\epsilon+3\delta_3}\infdiv*{p_{\rvs{X}_1^n}\cdot W_{\rv{YZ}|\rv{X}}^{\tensor n}}{p_{\rvs{X}_1^n}\times p_{\rvs{Y}_1^n} \times p_{\rvs{Z}_1^n}} + \frac{3}{n}\log{\delta_3}.
\end{align}
Using the information spectrum method~\cite{verdu1994general}, we know $\lim_{n\to\infty} \text{RHS of~\eqref{eq:outer:broadcast:asymptotic:d}} = \tilde{C}(W_{\rv{YZ}|\rv{X}})$.
Since~\eqref{eq:outer:broadcast:asymptotic:d} holds for all $n$ sufficiently large, the inequality is maintained as $n\to\infty$, \ie, $r_1+r_2\geq \tilde{C}(W_{\rv{YZ}|\rv{X}})$.

Similarly, using~\cite[Lemma~10]{anshu2020partially} with~\eqref{eq:outer:broadcast:asymptotic:a} and~\eqref{eq:outer:broadcast:asymptotic:b}, one can show $r_1\geq C(W_{\rv{Y}|\rv{X}})$ and $r_2\geq C(W_{\rv{Z}|\rv{X}})$, respectively.

Since $(r_1, r_2)$ are picked arbitrarily, we have shown
\begin{equation}
\left\{(r_1,r_2)\in\mathbb{R}_{\geq 0}^2\middle\vert \exists N\in\mathbb{N}\text{ s.t. }(\floor{2^{nr_1}},\floor{2^{nr_2}})\in\set{M}_\epsilon^\star(W_{\rv{YZ}|\rv{X}}^{\tensor n})\text{ for all }n\geq N\right\} \subset \text{LHS of~\eqref{eq:broadcast:asymptotic}}.
\end{equation}
Finally, taking closure of the sets on both sides we have $\set{R}^\star_\epsilon(W_{\rv{YZ}|\rv{X}})\subset\text{LHS of~\eqref{eq:broadcast:asymptotic}}$.
\end{proof}

\subsection{K-receiver Broadcast Channels}

All previous discussions in this section generalize rather straightforwardly to broadcast channels with more than $2$ receivers.
In particular, we have the following generalized versions of Lemma~\ref{lem:convex:split:bipartite}, Lemma~\ref{lem:common:non-lockability} and Lemma~\ref{lem:Ds:2}.

\begin{lemma}[Multipartite Convex Split Lemma]\label{lem:convex:split:multi-partite}
Let $K$ be a positive integer and $\epsilon\in(0,1)$. 
For each nonempty subset $\set{J}$ of $\{1,\ldots,K\}$, let $\delta_\set{J}\in(0,1)$, and suppose
\begin{equation}
\sum_{\set{J}\subset\{1,\ldots,K\}: \set{J\neq\emptyset}} \delta_\set{J}^2 \leq \epsilon^2.
\end{equation}
Let $(\rv{X},\rv{Y}_1,\ldots,\rv{Y}_K)$ be jointly distributed over $\set{X}\times\set{Y}_1\times\cdots\times\set{Y}_K$ with pmf $p_{\rv{X}\rvs{Y}_1^K}$.
For each $i\in\{1,\ldots,K\}$, let $q_{\rv{Y}_i}$ be a pmf over $\set{Y}_i$. (Note that in general $q_{\rv{Y}_i}\neq q_{\rv{Y}_j}$ for $i\neq j$.)
Let $\mathbf{M}_1^K=(M_1,\ldots,M_K)$ be a $K$-length vector of positive integers such that for all nonempty subset $\set{J}$ of $\{1,\ldots,K\}$
\begin{equation}
\sum_{i\in\set{J}} \log{M}_i \geq D_{s+}^{\epsilon_\set{J}}\infdiv*{p_{\rv{X}\rvs{Y}_\set{J}}}{p_\rv{X}\times \prod_{i\in\set{J}}q_{\rv{Y}_i}} - \log{\delta_\set{J}^2}
\end{equation}
for some $\{\epsilon_\set{J}\}_{\set{J}\subset\{1,\ldots,K\}:\set{J}\neq\emptyset}$ such that
\begin{equation}
\sum_{\set{J}\subset\{1,\ldots,K\}:\set{J}\neq\emptyset} \epsilon_\set{J} \leq \epsilon - \sqrt{\sum_{\set{J}\subset\{1,\ldots,K\}:\set{J}\neq\emptyset} \delta_\set{J}^2}.
\end{equation}
For each $i\in\{1,\ldots,K\}$, let $\rv{J}_i$ be independently uniformly distributed on $\{1,\ldots,M_i\}$.
Let the joint random variables $(\rvs{J}_1^K,\rv{X},\rvs{Y}_{1,1}^{M_1},\ldots,\rvs{Y}_{K,1}^{M_K})$ be distributed according to 
\begin{equation}
p_{\rv{X}, \rvs{Y}_{1,1}^{M_1},\ldots \rvs{Y}_{K,1}^{M_K}|\rvs{J}_1^K}(x,\mathbf{y}_{1,1}^{M_1},\ldots,\mathbf{y}_{K,1}^{M_K}|\mathbf{j}_1^K) = p_{\rv{X}\rvs{Y}_1^K}(x,y_{1,j_1},\ldots,y_{K,j_K})\cdot \prod_{i=1}^{K} \prod_{j\in\{1,\ldots,M_i\}\setminus\{j_i\}} q_{\rv{Y}_i}(y_{i,j}).
\end{equation}
Then,
\begin{equation}
\norm{p_{\rv{X},\rvs{Y}_{1,1}^{M_1},\ldots,\rvs{Y}_{K,1}^{M_K}} - p_\rv{X}\times \prod_{i=1}^{K} \prod_{j=1}^{M_i} q_{\rv{Y}_{i,j}}}_{\fnc{tvd}} \leq \epsilon
\end{equation}
where for each $i=1,\ldots,K$, $q_{\rv{Y}_{i,j}}= q_{\rv{Y}_{i}}$ for all $j=1,\ldots, M_i$.
\end{lemma}

Given discrete random variables $\rv{X}\rv{Y}_1\ldots\rv{Y}_k$ with joint pmf $p_{\rv{X}\rvs{Y}_1^k}\in\set{P}(\set{X}\times\set{Y}^k)$ where $k$ is some positive integer, the max-mutual information of $\rv{X}$ \vs $\rv{Y}_1\ldots\rv{Y}_k$ is defined as
\begin{equation}\label{def:K-partite:Imax}
    I_{\max}(\rv{X};\rv{Y}_1;\cdots;\rv{Y}_k) \defeq \inf_{q_{\rv{Y}_i}\in\set{P}(\set{Y}_i), i=1,\ldots,k} D_{\max}\infdiv*{p_{\rv{X\rvs{Y}_1^k}}}{p_\rv{X}\times q_{\rv{Y}_1} \times \cdots \times q_{\rv{Y}_k}}
\end{equation}
where $p_\rv{X}$ is the marginal distribution of $\rv{X}$ induced from $p_{\rv{X}\rvs{Y}_1^k}$.

\begin{lemma}\label{lem:k:non-lockability}
Let $(\rv{X}, \rv{Y}_1, \tilde{\rv{Y}}_1, \ldots, \rv{Y}_k, \tilde{\rv{Y}}_k)$ be joint random variables distributed on $\set{X}\times\set{Y}_1\times\tilde{\set{Y}}_1\times\cdots\set{Y}_k\times\tilde{\set{Y}}_k$ where $k$ is some positive integer.
Suppose that all the sets involved above are finite.
Then,
\begin{equation}
    I_{\max}(\rv{X};\rv{Y}_1\tilde{\rv{Y}}_1;\cdots;\rv{Y}_k\tilde{\rv{Y}}_k) \leq I_{\max}(\rv{X};\rv{Y}_1;\rv{Y}_k) + \sum_{i=1}^k \log{\size{\tilde{\set{Y}}_i}}.
\end{equation}
\end{lemma}

\begin{lemma}\label{lem:Ds:k}
Suppose $p_\rv{X\rvs{Y}_1^k}\in\set{P}(\set{X}\times\set{Y}_1\times\cdots\times\set{Y}_k)$ and $q_{\rv{Y}_i}\in\set{P}(\set{Y}_i)$ for each $i\in\{1,\ldots,k\}$, where $k$ is some positive integer.
For any $\epsilon\in(0,1)$, and $\delta\in(0,\frac{1-\epsilon}{k})$, we have
\begin{equation}
\inf_{q_{\rv{Y}_1},\ldots, q_{\rv{Y}_k}} D_{s+}^{\epsilon}\infdiv*{p_{\rv{X}\rvs{Y}_1^k}}{p_\rv{X}\times q_{\rv{Y}_1}\times \cdots \times q_{\rv{Y}_k}} \geq D_{s+}^{\epsilon+k\delta}\infdiv*{p_{\rv{X}\rvs{Y}_1^k}}{p_\rv{X}\times p_{\rv{Y}_1}\times\cdots\times p_{\rv{Y}_k}} + k\log{\delta}.
\end{equation}
\end{lemma}

\par
Using Lemma~\ref{lem:convex:split:multi-partite} and Lemma~\ref{lem:k:non-lockability}, we can generalize Proposition~\ref{prop:1abfsbc} and Proposition~\ref{prop:1cbfsbc} to $K$-receiver broadcast channels, respectively, as follows.

\begin{proposition}\label{prop:1abfsbc:K}
    Let $K$ be a positive integer.
    Let $W_{\rvs{Y}_1^K|\rv{X}}$ be a channel from $\set{X}$ to $\set{Y}_1\times\cdots\times\set{Y}_K$, and let $\epsilon\in(0,1)$.
    For each subset $\set{J}$ of $\{1,\ldots,K\}$, let $\epsilon_\set{J}$ and $\delta_\set{J}<\epsilon_\set{J}$ be positive real numbers.
    Suppose
    \begin{equation}
    \sum_{\set{J}\subset\{1,\ldots,K\}:\set{J}\neq\emptyset} \epsilon_\set{J}\leq\epsilon.
    \end{equation}
    Let $q_{\rv{Y}_i}\in\set{P}(\set{Y}_i)$ for each $i=1,\ldots,K$.
    The following set is a subset of $\set{M}^\star_\epsilon(W_{\rvs{Y}_1^K|\rv{X}})$
    \begin{equation}
    \set{M}^\text{in}_\epsilon(W_{\rvs{Y}_1^K|\rv{X}}) \defeq \left\{ \mathbf{M}\in\mathbb{Z}_{>0}^K \middle\vert 
    \begin{aligned}
    \sum_{i\in\set{J}}\log{M_i} > D_{s+}^{\epsilon_\set{J}-\delta_\set{J}}\infdiv*{p_\rv{X}\cdot W_{\rvs{Y}_\set{Y}|\rv{X}}}{p_\rv{X}\times \prod_{i\in\set{J}}q_{\rv{Y}_i}} - \log{\delta_\set{J}^2} \\
    \forall \set{J}\subset\{1,\ldots,K\} \text{ with } \set{J}\neq\emptyset, \forall p_\rv{X}\in\set{P}(\set{X})
    \end{aligned}
    \right\}.
    \end{equation}
\end{proposition}

\begin{proposition}\label{prop:1cbfsbc:K}
    Let $K$ be a positive integer.
    Let $W_{\rvs{Y}_1^K|\rv{X}}$ be a channel from $\set{X}$ to $\set{Y}_1\times\cdots\times\set{Y}_K$, and let $\epsilon\in(0,1)$.
    For each subset $\set{J}$ of $\{1,\ldots,K\}$, let $\delta_\set{J}\in(0,1)$.
    The following set is a superset of $\set{M}^\star_\epsilon(W_{\rvs{Y}_1^K|\rv{X}})$
    \begin{equation}
    \set{M}^\text{out}_\epsilon(W_{\rvs{Y}_1^K|\rv{X}}) \defeq \left\{ \mathbf{M}\in\mathbb{Z}_{>0}^K \middle\vert
    \begin{aligned}
    \sum_{i\in\set{J}}\log{M_i} \geq \adjustlimits\inf_{\{q_{\rv{Y}_i}\}_{i\in\set{J}}} \sup_{p_\rv{X}} D_{s+}^{\epsilon+\delta_\set{J}}\infdiv*{p_\rv{X}\cdot W_{\rvs{Y}_\set{J}|\rv{X}}}{p_\rv{X}\times \prod_{i\in\set{J}} q_{\rv{Y}_i}} + \log{\delta_\set{J}}\\
    \forall \set{J}\subset\{1,\ldots,K\} \text{ with } \set{J}\neq\emptyset
    \end{aligned}
    \right\}.
    \end{equation}
\end{proposition}

\par
With the help of Lemma~\ref{lem:Ds:k}, and following a similar path as that led to Theorem~\ref{thm:broadcast:asymptotic}, we have the following asymptotic results for simulating $K$-receiver broadcast channels.

\begin{theorem}\label{thm:broadcast:asymptotic:K}
    Let $W_{\rvs{Y}_1^K|\rv{X}}$ be a channel from $\set{X}$ to $\set{Y}_1\times\cdots\times\set{Y}_K$, and let $\epsilon\in(0,1)$.
    It holds that
    \begin{equation}
    \set{R}^\star_\epsilon(W_{\rvs{Y}_1^K|\rv{X}})=\left\{(r_1,\cdots,r_n)\in\mathbb{R}_{>0}^n\middle|\sum_{i\in\set{J}} r_i\geq\tilde{C}(W_{\rvs{Y}_\set{J}|\rv{X}})\quad\forall \set{J}\subset\{1,\ldots,K\}\text{ with }\set{J}\neq\emptyset\right\},
    \end{equation}
    where $\tilde{C}(W_{\rvs{Y}_\set{J}|\rv{X}})$ has been defined in~\eqref{eq:def:K-partite:tilde:C}.
\end{theorem}

\subsection{Blahut--Arimoto Algorithms and Numerical Results}

In the remainder of this section, we present a Blahut--Arimoto type algorithm to compute the mutual information and multipartite common information of broadcast channels together with a couple of numerical examples. Generalizing the Blahut-Arimoto algorithm to compute the multipartite common information is fairly straightforward, but has not been previously done in the literature to the best of our knowledge.

We consider a broadcast channel with $K$ receivers. For the simplest case where $K$ being $1$, the capacity of this (point-to-point) channel can computed using the Blahut--Arimoto algorithm~\cite{blahut1972computation, arimoto1972algorithm}.
Specifically, given $W_{\rv{Y}|\rv{X}}$, the optimization problem
\begin{align}
    \max_{p_{\rv{X}}\in\set{P}(\set{X})} I(\rv{X};\rv{Y})_{p_\rv{X}\cdot W_{\rv{Y}|\rv{X}}}
\end{align}
can be solved in an iterative manner, \ie, at the $t$-th step, we compute
\begin{equation}
    p^{(t)}_{\rv{X}} \propto p^{(t-1)}_{\rv{X}} \cdot \exp{\left(D(W_{\rv{Y}|\rv{X}}\|p^{(t-1)}_{\rv{Y}})\right)} ,
\end{equation}
where $p^{(t-1)}_{\rv{Y}} \defeq \sum_{x\in\set{X}}W_{\rv{Y}|\rv{X}}(\cdot|x)\cdot p^{(t-1)}_{\rv{X}}(x)$.
An \textit{a priori} guarantee can be made regarding the convergence speed for such an algorithm.
Namely, with the initial guess being the uniform distribution, at the $n$-th iteration, we have
\begin{align}
    \abs{C(W_{\rv{Y}|\rv{X}}) - C(n))} \leq \frac{\log |\mathcal{X}|}{n}\, ,
\end{align}
where $C(n)$ is the estimate of the capacity after $n$ iterations.

Regarding the $K$-partite common information of the broadcast channel $W_{\rvs{Y}_1^K|\rv{X}}$, we consider the following optimization problem
\begin{align}
    \tilde{C}(W_{\rvs{Y}_1^K|\rv{X}})= \max_{p_\rv{X}\in\set{P}(\set{X})} I(\rv{X};\rv{Y}_1;\rv{Y}_2;\ldots;\rv{Y}_K)_{p_\rv{X}\cdot W_{\rvs{Y}_1^K|\rv{X}}}.
\end{align}
Using similar techniques, we can solve this optimization problem in an iterative manner with the following update rule 
\begin{align}
\label{eq:update_rule_classical_multipartite}
    p^{(t)}_{\rv{X}} \propto p^{(t-1)}_{\rv{X}}\cdot \exp{\left(\frac{D\infdiv*{W_{\rvs{Y}_1^K|\rv{X}}}{p^{(t-1)}_{\rv{Y}_1}\times\cdots\times p^{(t-1)}_{\rv{Y}_K}}}{K}\right)} \qquad t=1,2,\ldots,n,\ldots,
\end{align}
where $p^{(t-1)}_{\rv{Y}_i} \defeq \sum_{x\in\set{X}}W_{\rv{Y}_i|\rv{X}}(\cdot|x)\cdot p^{(t-1)}_{\rv{X}}(x)$.
This allows us to efficiently compute the multipartite common information terms in Theorems~\ref{thm:broadcast:asymptotic} and~\ref{thm:broadcast:asymptotic:K}.
In the $K$-partite case, the speed of convergence is given by 
\begin{align}
    \abs{\tilde{C}(W_{\rvs{Y}_1^K|\rv{X}}) - \tilde{C}(n)} \leq \frac{K\log |\mathcal{X}|}{n}\, .
\end{align}
A more detailed discussion is included in Appendix~\ref{app:blahut:arimoto:details}. In the following, we present a couple of numerical examples using the aforementioned method.\\


\subsubsection{Two Receivers}

Here, we compute the simulation region for a two receiver broadcast channel. For binary random variables $\rv{X}, \rv{Y}$ and $\rv{Z}$, let $W_{\rv{YZ}|\rv{X}} = \fnc{BSC}_{\rv{Z}|\rv{Y}}\circ \fnc{BSC}_{\rv{Y}|\rv{X}}$, where $\fnc{BSC}_{\rv{Z}|\rv{Y}}$ and $\fnc{BSC}_{\rv{Y}|\rv{X}}$ are binary symmetric channels with crossover probability $\delta = 0.3$. 
Fig,~\ref{fig:degraded_BSC_region-intro} shows the simulation region computed using the Blahut--Arimoto type algorithm.\\


\subsubsection{Three Receivers}

\begin{figure}
\centering
\includegraphics{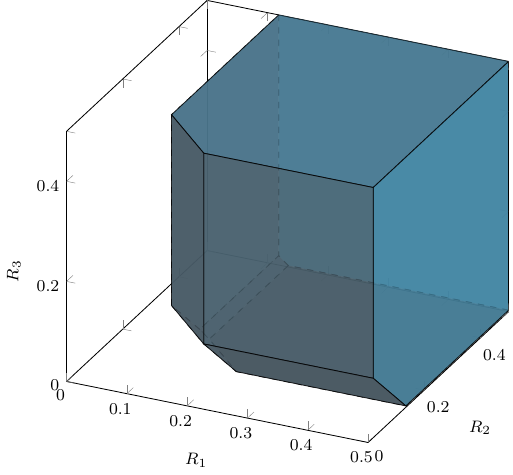}
\caption{Asymptotic simulation region for the broadcast channel $W_{\rvs{Y}_1^3|\rv{X}} = \fnc{BSC}_{\rv{Y}_3|\rv{Y}_2}\circ \fnc{BSC}_{\rv{Y}_2|\rv{Y}_1}\circ \fnc{BSC}_{\rv{Y}_1|\rv{X}}$, where $\fnc{BSC}$ refers to a binary symmetric channel with crossover probability $\delta = 0.3$.}    
\label{fig:degraded_BSC_3_region}
\end{figure}

Here, we consider a broadcast channel with three receivers. For binary random variables $\rv{X}, \rv{Y_1}, \rv{Y_2}$ and $\rv{Y_3}$, let $W_{\rv{Y_1Y_2Y_3}|\rv{X}} = \fnc{BSC}_{\rv{Y_3}|\rv{Y_2}}\circ \fnc{BSC}_{\rv{Y_2}|\rv{Y_1}}\circ \fnc{BSC}_{\rv{Y_1}|\rv{X}}$, where $\fnc{BSC}$ is a binary symmetric channel with crossover probability $\delta = 0.3$. Fig,~\ref{fig:degraded_BSC_3_region} shows the simulation region computed using the Blahut--Arimoto type algorithm.



\section{Discussions and Conclusions}
\label{sec:conclusion}


In the first part of our work, we have introduced one-shot and finite-blocklength bounds for the amount of communication needed for channel simulation in both (unconstrained) random-assisted and no-signaling-assisted scenarios.
Using these bounds, we have derived tight higher-order asymptotic expansions in the small and moderate-deviation regimes, and our results imply that channel interconversion becomes irreversible for finite blocklengths.
Despite our seemingly tight one-shot bounds (see~\eqref{eq:achievability}), our (small-deviation) asymptotic expansion is only tight up to the second order (see~\eqref{eq:second-order-simulation}).
On the one hand, this is due to the lack of a direct asymptotic analysis of the $\epsilon$-smoothed channel max-mutual information, \ie, a higher-order expansion of $I_{\max}^\epsilon(W_{\rv{Y}|\rv{X}}^{\tensor n})$.
On the other hand, following the techniques in~\cite{tomamichel2013tight} carefully, we end up with upper and lower bounds on the $n$-fold channel simulation cost for most (\ie, non-exotic~\cite{polyanskiy2010channel}) channels that differ by $\log n + o(\log n)$.

It is worth noting that, in the channel coding case, the third-order terms of the binary erasure channels and the binary symmetric channels differ by $\frac{1}{2}\log{n}$, whereas in the channel simulation case it remains an open question whether the third-order term is channel-dependent. 
Nevertheless, the high-order asymptotic expansion of the smoothed channel max-mutual information is an open and interesting problem.

In the second part of our work, we have extended our results to network settings, where we have characterized the amount of communication needed for the simulation of broadcast channels in both one-shot and asymptotic settings.
Despite the similarity between the one-shot results for the point-to-point channel simulation (see Proposition~\ref{prop:Ds:bounds}) and the broadcast channel simulation (see Propositions~\ref{prop:1abfsbc} and~\ref{prop:1cbfsbc}), the high-order asymptotic expansion of the latter turned out to be more elusive due to what is known as the joint smoothing problem.

\section*{Acknowledgments}

We thank Patrick Hayden for suggesting the topic of broadcast channel simulation (\cf~the abstract~\cite{hayden2007reverse}), Rahul Jain for helpful discussions on the techniques of rejection sampling, and Vincent Tan and Lei Yu for pointing out typos and some related literature.
This research is supported by the National Research Foundation, Singapore, and A*STAR under its CQT Bridging Grant.
MC and MT are supported by NUS startup grants (R-263-000-E32-133 and R-263-000-E32-731).
MB acknowledges funding from the European Research Council (ERC Grant Agreement No.~948139).  


\bibliographystyle{IEEEtran}
\bibliography{reference}


\appendices
\section{Proof of Lemma~\ref{lem:quasi:convex:Ds}}\label{app:proof:lem:quasi:convex:Ds}
\begin{proof}
    Let $a\defeq\sup_{x\in\set{X}} D_{s+}^\epsilon\infdiv*{p_{\rv{Y}|\rv{X}}(\cdot|x)}{q_{\rv{Y}|\rv{X}}(\cdot|x)}$.
    By the definition of $D_{s+}$, we know that for all $\delta>0$
    \begin{equation}
    \Pr\nolimits_{\rv{Y}_x\sim p_{\rv{Y}|\rv{X}}(\cdot|x)}\left[\log{\frac{p_{\rv{Y}|\rv{X}}(\rv{Y}_x|x)}{q_{\rv{Y}|\rv{X}}(\rv{Y}_x|x)}}>a+\delta\right] < \epsilon \quad \forall x\in\set{X}.
    \end{equation}
    For each $x\in\set{X}$, denote $\set{A}\defeq\left\{(x,y)\in\set{X}\times\set{Y}\middle\vert \log{\frac{p_{\rv{Y}|\rv{X}}(y|x)}{q_{\rv{Y}|\rv{X}}(y|x)}}>a+\delta\right\}$, and $\set{A}_x\defeq\left\{y\in\set{Y}\middle\vert\log{\frac{p_{\rv{Y}|\rv{X}}(y|x)}{q_{\rv{Y}|\rv{X}}(y|x)}}>a+\delta\right\}$.
    We have 
    \begin{equation}
    p_{\rv{XY}}(\set{A})  = \sum_{x\in\set{X}} p_\rv{X}(x) \cdot p_{\rv{Y}|\rv{X}=x}(\set{A}_x) = \sum_{x\in\set{X}} p_\rv{X}(x) \cdot \Pr\nolimits_{\rv{Y}_x\sim p_{\rv{Y}|\rv{X}}(\cdot|x)}\left[\log{\frac{p_{\rv{Y}|\rv{X}}(\rv{Y}_x|x)}{q_{\rv{Y}|\rv{X}}(\rv{Y}_x|x)}}>a+\delta\right] < \epsilon.
    \end{equation}
    Thus, by definition, $D_{s+}^\epsilon\infdiv*{p_{\rv{Y}|\rv{X}}\cdot p_\rv{X}}{q_{\rv{Y}|\rv{X}}\cdot p_\rv{X}}\leq a+\delta$.
    Since this holds for all $\delta>0$, it must also hold when $\delta=0$, which finishes the proof.
\end{proof}
\section{Proof of Lemma~\ref{lem:maxDmax:supDs}}\label{app:proof:lem:maxDmax:supDs}
The following proof is partially inspired by ~\cite[Eq.~(48)]{anshu2020partially} and~\cite[Eq.~(17)]{tomamichel2013hierarchy}.
\begin{proof}
Let $d\defeq\inf_{\tilde{W}_{\rv{Y}|\rv{X}}:\norm{\tilde{W}_{\rv{Y}|\rv{X}}-W_{\rv{Y}|\rv{X}}}_\fnc{tvd}\leq\epsilon} \max_{x\in\set{X}} D_{\max}\infdiv*{\tilde{W}_{\rv{Y}|\rv{X}}(\cdot|x)}{q_\rv{Y}}$, and let $\hat{W}_{\rv{Y}|\rv{X}}$ be an optimal DMC for this optimization problem, \ie, 
\begin{enumerate}
    \item $\norm{\hat{W}_{\rv{Y}|\rv{X}}-W_{\rv{Y}|\rv{X}}}_\fnc{tvd}\leq\epsilon$;
    \item $d= \max_{x\in\set{X}} D_{\max}\infdiv*{\hat{W}_{\rv{Y}|\rv{X}}(\cdot|x)}{q_\rv{Y}}$.
\end{enumerate}
Starting with 1), we know for any $\xi\in\mathbb{R}$
\begin{align}
\epsilon &\geq \sum_{y\in\set{Y}: W_{\rv{Y}|\rv{X}}(y|x)>2^{d+\xi}\cdot q_\rv{Y}(y)} W_{\rv{Y}|\rv{X}}(y|x) - \hat{W}_{\rv{Y}|\rv{X}}(y|x)  \\
&\geq \sum_{y\in\set{Y}: W_{\rv{Y}|\rv{X}}(y|x)>2^{d+\xi}\cdot q_\rv{Y}(y)} W_{\rv{Y}|\rv{X}}(y|x) - \sum_{y\in\set{Y}: W_{\rv{Y}|\rv{X}}(y|x)>2^{d+\xi}\cdot q_\rv{Y}(y)} 2^{d} \cdot q_\rv{Y}(y) \\
&\geq \sum_{y\in\set{Y}: W_{\rv{Y}|\rv{X}}(y|x)>2^{d+\xi}\cdot q_\rv{Y}(y)} W_{\rv{Y}|\rv{X}}(y|x) - \sum_{y\in\set{Y}: W_{\rv{Y}|\rv{X}}(y|x)>2^{d+\xi}\cdot q_\rv{Y}(y)} 2^{d} \cdot 2^{-d-\xi} \cdot  W_{\rv{Y}|\rv{X}}(y|x) \\
&\geq \sum_{y\in\set{Y}: W_{\rv{Y}|\rv{X}}(y|x)>2^{d+\xi}\cdot q_\rv{Y}(y)} W_{\rv{Y}|\rv{X}}(y|x) - 2^{-\xi},
\end{align}
\ie, $\Pr_{\rv{Y}\sim W_{\rv{Y}|\rv{X}=x}}\left[\log{\frac{W_{\rv{Y}|\rv{X}}(\rv{Y}|x)}{q_\rv{Y}(\rv{Y})}}>d+\xi\right] = \sum_{y\in\set{Y}: W_{\rv{Y}|\rv{X}}(y|x)>2^{d+\xi}\cdot q_\rv{Y}(y)} W_{\rv{Y}|\rv{X}}(y|x) \leq \epsilon+2^{-\xi}$ for all $x\in\set{X}$.
Thus, for any $p_\rv{X}\in\set{P}(\set{X})$, 
\begin{equation}
\Pr\nolimits_{\rv{XY}\sim W_{\rv{Y}|\rv{X}}}\left[\log{\frac{p_\rv{X}(\rv{X})\cdot W_{\rv{Y}|\rv{X}}(\rv{Y}|\rv{X})}{p_\rv{X}(\rv{X})\cdot q_\rv{Y}(\rv{Y})}}>d+\xi\right] = \sum_{x\in\set{X}} p_\rv{X}(x)\cdot \Pr_{\rv{Y}\sim W_{\rv{Y}|\rv{X}=x}}\left[\log{\frac{W_{\rv{Y}|\rv{X}}(\rv{Y}|x)}{q_\rv{Y}(\rv{Y})}}>d+\xi\right] \leq \epsilon + 2^{-\xi}.
\end{equation}
Considering $\xi\gets -\log{\delta}+\Delta$ where $\Delta>0$, and by the definition of $D_{s+}^{\epsilon_\delta}\infdiv*{p_\rv{X}\cdot W_{\rv{Y}|\rv{X}}}{p_\rv{X}\times q_\rv{Y}}$, we know (for all $p_\rv{X}\in\set{P}(\set{X})$)
\begin{equation}
D_{s+}^{\epsilon+\delta}\infdiv*{p_\rv{X}\cdot W_{\rv{Y}|\rv{X}}}{p_\rv{X}\times q_\rv{Y}} \leq d -\log{\delta}+\Delta 
\end{equation}
for all $\Delta>0$.
Therefore,
\begin{equation}
\sup_{p_\rv{X}\in\set{P}(\set{X})} D_{s+}^{\epsilon+\delta}\infdiv*{p_\rv{X}\cdot W_{\rv{Y}|\rv{X}}}{p_\rv{X}\times q_\rv{Y}} \leq d -\log{\delta}
\end{equation}
which is equivalent to~\eqref{eq:maxDmax:supDs}.
\end{proof}
\section{Proof of~(\ref{eq:Chebyshev:lower:Ds})} \label{app:Chebyshev:lower:Ds}
\begin{lemma}\label{lem:5a}
Let $p$ and $q\in\set{P}(\set{Y})$ where $\set{X}$ is some finite set.
For all positive integer $n$, it holds that
\begin{equation}
D_{s+}^\epsilon\infdiv*{p^{\tensor}}{q^{\tensor n}} \geq nD\infdiv*{p}{q} - \sqrt{\frac{nV\infdiv*{p}{q}}{1-\epsilon}}.
\end{equation}
\end{lemma}
\begin{proof}
By the Chebyshev's inequality, for all $a< nD\infdiv*{p}{q}$, it holds that
\begin{equation}
\Pr\nolimits_{\rvs{X}_1^n\sim p^{\tensor n}}\left[\sum_{i=1}^n \log{\frac{p(\rv{X}_i)}{q(\rv{X}_i)}}>a\right] \geq 1-\frac{nV\infdiv*{p}{q}}{(a-nD\infdiv*{p}{q})^2}.
\end{equation}
Thus, 
\begin{align}
D_{s+}^{\epsilon}\infdiv*{p^{\tensor n}}{q^{\tensor n}}
= & \inf\left\{a\geq 0\middle\vert \Pr\nolimits_{\rvs{X}_1^n\sim p^{\tensor n}}\left[\sum_{i=1}^n \log{\frac{p(\rv{X}_i)}{q(\rv{X}_i)}}>a\right]<\epsilon \right\}\\
\tos{a)}{\geq} & \inf\left\{a<nD\infdiv*{p}{q} \middle\vert \Pr\left[\sum_{i=1}^n \log{\frac{p(\rv{X}_i)}{q(\rv{X}_i)}}>a\right]<\epsilon \right\} \\
\geq & \inf\left\{a<nD\infdiv*{p}{q} \middle\vert 1-\frac{nV\infdiv*{p}{q}}{(a-nD\infdiv*{p}{q})^2}<\epsilon \right\} \\
=& nD\infdiv*{p}{q} - \sqrt{\frac{nV\infdiv*{p}{q}}{1-\epsilon}},
\end{align}
where, in step a), we have assumed the set on the RHS to be nonempty, since otherwise one shall have $D_{s+}^{\epsilon}\infdiv*{p^{\tensor n}}{q^{\tensor n}}\geq nD\infdiv*{p}{q}$ and the lemma holds trivially.
\end{proof}
\section{Proof of Lemma~\ref{lem:inf:max}} \label{app:proof:lem:inf:max}
\begin{proof}
When $f$ is the zero function, the statement holds trivially. 
Otherwise, we first consider the case when $f$ is strictly positive.
In this case, we define a pmf on $\set{X}$ as $p^\star(x)\defeq f(x)/\sum_{x}f(x)$.
It is obvious that $\max_{x}\frac{f(x)}{p^\star(x)}=\sum_{x}f(x)$.
Assume that there exists some other pmf $p\in\set{P}(\set{X})$ such that $\max_{x}\frac{f(x)}{p(x)}<\max_{x}\frac{f(x)}{p^\star(x)}=\sum_{x}f(x)$.
Then it must hold that $p(x)>p^\star(x)$ for all $x\in\set{X}$, which is impossible.
Hence, $p^\star$ is a minimizer, and 
\begin{equation}
\inf_{p\in\set{P}(\set{X})} \max_{x\in\set{X}} \frac{f(x)}{p(x)} = \max_{x}\frac{f(x)}{p^\star(x)}=\sum_{x}f(x).
\end{equation}
Now, consider the case when $f$ is zero for some (but not all) $x\in\set{X}$.
Denote $\set{X}_0\defeq\{x\in\set{X}|f(x)=0\}$.
We define a sequence of pmfs as 
\begin{equation}
p_n(x) \defeq\begin{cases} 2^{-\frac{1}{n}}\cdot p^\star_\rv{X}(x) &\text{ if } f(x)>0 \\
    \frac{1-2^{-\frac{1}{n}}}{\size{\set{X}_0}} &\text{ otherwise.} \end{cases}
\end{equation}
In this case, we have 
\begin{equation}
\max_{x\in\set{X}} \frac{f(x)}{p_n(x)} = 2^{\frac{1}{n}}\cdot \sum_{x\in\set{X}} f(x),
\end{equation}
which is monotonically decreasing, and tends to $\sum_{x\in\set{X}} f(x)$.
This shows that
\begin{equation}
\inf_{p\in\set{P}(\set{X})} \max_{x\in\set{X}} \frac{f(x)}{p(x)}  \leq \sum_{x\in\set{X}} f(x).
\end{equation}
On the other hand , assume that there exists some pmf $p\in\set{P}(\set{X})$ such that $\max_{x}\frac{f(x)}{p(x)}<\sum_{x}f(x)$.
Then, it must hold that $\frac{f(x)}{p(x)}< 2^{\frac{1}{n}}\cdot \sum_{x\in\set{X}} f(x) = \frac{f(x)}{p_n(x)}$ for all $x\in\set{X}\setminus\set{X}_0$ for all $n$, \ie, $p(x)>p_n(x)$ for all $x\in\set{X}\setminus\set{X}_0$ for all $n$.
This implies that $p(x)\geq \lim_{n\to\infty}p_n(x)=p^\star(x)$ for all $x\in\set{X}\setminus\set{X}_0$.
Since $p^\star(x)=0$ for $x\in\set{X}_0$, it also holds that $p(x)\geq p^\star(x)$ for all $x\in\set{X}_0$, and thus $p=p^\star$, which contradicts with the assumption that $\max_{x}\frac{f(x)}{p(x)}<\sum_{x}f(x)$.
Hence, 
\begin{equation}
\inf_{p\in\set{P}(\set{X})} \max_{x\in\set{X}} \frac{f(x)}{p(x)}  \geq \sum_{x\in\set{X}} f(x),
\end{equation}
which finishes the proof.
\end{proof}
\section{Proof of Lemma~\ref{lem:convex:split:bipartite}}
\label{app:proof:lem:convex:split:bipartite}
Note that the following proof is almost identical to the proof for~\cite[Fact~7]{anshu2017unified}, and is included here mainly for completeness.
In addition, note that this proof can be generalized rather straightforwardly to the case when $K>2$.
\begin{proof}
Define the subset $\set{A}\subset\set{X}\times\set{Y}\times\set{Z}$ as
\begin{equation}
\set{A} \defeq \left\{(x,y,z)\middle\vert \begin{aligned}
    \frac{p_\rv{XY}(x,y)}{p_\rv{X}(x)\cdot q_\rv{Y}(y)} &\leq \delta_1^2 \cdot M\\
    \frac{p_\rv{XYZ}(x,z)}{p_\rv{X}(x) \cdot r_\rv{Z}(z)} &\leq \delta_2^2 \cdot N\\
    \frac{p_\rv{XYZ}(x,y,z)}{p_\rv{X}(x)\cdot q_\rv{Y}(y) \cdot r_\rv{Z}(z)} &\leq \delta_3^2 \cdot MN
\end{aligned}
\right\}.
\end{equation}
In this case, 
\begin{align}
p_\rv{XYZ}(\set{A}^C)
& \leq \Pr\left[\frac{p_\rv{XY}(\rv{X},\rv{Y})}{p_\rv{X}(\rv{X})\cdot q_\rv{Y}(\rv{Y})} > \delta_1^2 \!\cdot\! M\right] &+  \Pr \left[\frac{p_\rv{XZ}(\rv{X},\rv{Z})}{p_\rv{X}(\rv{X})\cdot r_\rv{Z}(\rv{Z})} > \delta_2^2 \!\cdot\! N\right] + \Pr\left[\frac{p_\rv{XYZ}(\rv{X},\rv{Y},\rv{Z})}{p_\rv{X}(\rv{X})\cdot q_\rv{Y}(\rv{Y}) \cdot r_\rv{Z}(\rv{Z})} > \delta_3^2 \!\cdot\! MN\right]\\
&\leq \epsilon_1 + \epsilon_2 + \epsilon_3 \leq \epsilon-\delta,
\end{align}
where $\rv{XYZ}\sim p_\rv{XYZ}$ and $\delta\defeq\sqrt{\delta_1^2+\delta_2^2+\delta_3^2}$.

For each $x\in\set{X}$, define $\set{A}_x\defeq \{(y,z)\in\set{Y}\times\set{Z}: (x,y,z)\in\set{A}\}$.
Denote $\epsilon_x\defeq p_{\rv{YZ}|\rv{X}=x}(\set{A}_x^C)$. 
Then, $\sum_x \epsilon_x \cdot p_\rv{X}(x) = p_\rv{XYZ}(\set{A}^C)\leq \epsilon-\delta$.
Let the random variables $(\rv{X}',\rv{Y}',\rv{Z}')$ be distributed as $p_{\rv{X}'\rv{Y}'\rv{Z}'}(x,y,z) = p_\rv{X}(x) \cdot p_{\rv{Y}'\rv{Z}'|\rv{X}'}(y,z|x)$ where $p_{\rv{Y}'\rv{Z}'|\rv{X}'}$ is defined as 
\begin{equation}
p_{\rv{Y}'\rv{Z}'|\rv{X}'}(y,z|x)
\defeq \begin{cases}p_{\rv{YZ}|\rv{X}}(y,z|x)+\epsilon_x\cdot q_\rv{Y}(y)\cdot r_\rv{Z}(z) &\forall (y,z)\in\set{A}_x\\ \epsilon_x\cdot q_\rv{Y}(y)\cdot r_\rv{Z}(z) & \text{otherwise}\end{cases}.
\end{equation}
This distribution is $(\epsilon-\delta)$-close to the distribution of $(\rv{X},\rv{Y},\rv{Z})$ as 
\begin{align}
\norm{p_\rv{XYZ}-p_{\rv{X}'\rv{Y}'\rv{Z}'}}_\fnc{tvd}
=& \sum_{x\in\set{X}} p_\rv{X}(x)\cdot \norm{p_{\rv{YZ}|\rv{X}=x} - p_{\rv{Y}'\rv{Z}'|\rv{X}'=x}}_\fnc{tvd} \\
\leq& \frac{1}{2} \sum_{x\in\set{X}} p_\rv{X}(x) \cdot \left(\sum_{(y,z)\in\set{Y}\times\set{Z}} \epsilon_x \cdot q_\rv{Y}(y) \cdot r_\rv{Z}(z) + \sum_{(y,z)\not\in\set{A}_x} p_{\rv{YZ}|\rv{X}}(y,z|x) \right)\\
= & \frac{1}{2} \cdot \sum_{x\in\set{X}} p_\rv{X}(x) \cdot\left(\epsilon_x + p_{\rv{YZ}|\rv{X}=x}(\set{A}_x^C)\right) = p_\rv{XYZ}(\set{A}^C) \leq \epsilon-\delta.
\end{align}
Additionally, by the definitions of $\{\set{A}_x\}_{x\in\set{X}}$ and $\set{A}$, we have (notice that $M\cdot N\geq 1/\delta_3^2$ due to~\eqref{eq:bcsl:anshu:3})
\begin{align}
p_{\rv{Y}'\rv{Z}'|\rv{X}'}(y,z|x) &\leq (\delta_3^2\cdot MN+ \epsilon_x) \cdot q_\rv{Y}(y)\cdot r_\rv{Z}(z)\\
&\leq (\delta_3^2\cdot MN+ 1) \cdot q_\rv{Y}(y)\cdot r_\rv{Z}(z)\\
&\leq 2\delta_3^2 \cdot MN \cdot q_\rv{Y}(y)\cdot r_\rv{Z}(z).
\end{align}
Similarly, we have
\begin{align}
p_{\rv{Y}'|\rv{X}'}(y|x) &\leq 2\delta_1^2\cdot M \cdot q_\rv{Y}(y),\\
p_{\rv{Z}'|\rv{X}'}(z|x) &\leq 2\delta_3^2\cdot N \cdot r_\rv{Z}(z).
\end{align}

Now, construct $(\rv{X}',\rv{Y}'_1,\ldots,\rv{Y}'_M,\rv{Z}'_1,\ldots,\rv{Z}'_N)$ from $(\rv{X}',\rv{Y}',\rv{Z}')$ in the same fashion as the construction of $(\rv{X},\rv{Y}_1,\ldots,\rv{Y}_M$, $\rv{Z}_1,\ldots,\rv{Z}_N)$ from $(\rv{X},\rv{Y},\rv{Z})$.
We have
\begin{equation}
\norm{p_{\rv{X}',\rv{Y}'_1,\ldots,\rv{Y}'_M,\rv{Z}'_1,\ldots,\rv{Z}'_N} - p_{\rv{X},\rv{Y}_1,\ldots,\rv{Y}_M,\rv{Z}_1,\ldots,\rv{Z}_N}}_\fnc{tvd} \leq\norm{p_\rv{XYZ}-p_{\rv{X}'\rv{Y}'\rv{Z}'}}_\fnc{tvd} \leq \epsilon-\delta,
\end{equation}
and
\begin{align}
    &\phantom{=} D\infdiv*{p_{\rv{X}',\rv{Y}'_1,\ldots,\rv{Y}'_M,\rv{Z}'_1,\ldots,\rv{Z}'_N}}{p_{\rv{X}'}\times \prod_{i=1}^{M} q_{\rv{Y}'_i} \times \prod_{\ell=1}^N r_{\rv{Z}'_\ell}}\nonumber \\
    &= \frac{1}{MN} \sum_{j,k} 
        \left\{ D\infdiv*{p_{\rv{X}'\rv{Y}'_j\rv{Z}'_k}}{p_{\rv{X}'}\times q_{\rv{Y}'_j}\times r_{\rv{Z}'_k}} - D\infdiv*{p_{\rv{X}'\rv{Y}'_j\rv{Z}'_k}\times \prod_{i\neq j}q_{\rv{Y}'_i}\times \prod_{\ell\neq k}r_{\rv{Z}'_\ell}}{p_{\rv{X}',\rv{Y}'_1,\ldots,\rv{Y}'_M,\rv{Z}'_1,\ldots,\rv{Z}'_N}} \right\} \\
    &\leq \begin{aligned}[t]
        & D\infdiv*{p_{\rv{X}'\rv{Y}'\rv{Z}'}}{p_{\rv{X}'}\times q_{\rv{Y}'}\times r_{\rv{Z}'}} - \\
        & D\infdiv*{p_{\rv{X}'\rv{Y}'\rv{Z}'}}{\frac{1}{MN}p_{\rv{X}'\rv{Y}'\rv{Z}'} + \frac{N-1}{MN}p_{\rv{X}'\rv{Y}'}\times r_{\rv{Z}'} + \frac{M-1}{MN}p_{\rv{X}'\rv{Z}'}\times q_{\rv{Y}'} + \frac{(M-1)(N-1)}{MN} p_{\rv{X}'}\times q_{\rv{Y}'}\times r_{\rv{Z}'}}
        \end{aligned}\\
    &= \sum_{x,y,z} p_{\rv{X}'\rv{Y}'\rv{Z}'}(x,y,z) \cdot \log
    \left(\frac{1}{MN}\cdot \frac{p_{\rv{Y}'\rv{Z}'|\rv{X}'}(y,z|x)}{q_\rv{Y}(y)\cdot r_\rv{Z}(z)} + 
    \frac{N\!-\!1}{MN}\cdot\frac{p_{\rv{Y}'|\rv{X}'}(y|x)}{q_\rv{Y}(y)} + \frac{M\!-\!1}{MN}\cdot\frac{p_{\rv{Z}'|\rv{X}'}(z|x)}{r_\rv{Z}(z)} + \frac{(M\!-\!1)(N\!-\!1)}{MN}\right)\hspace{-3pt}\\
    &\leq  \log\left(1 + 2\delta_1^2 + 2\delta_2^2 + 2\delta_3^2\right) \leq 2 \delta^2,
\end{align}
which, by Pinsker's inequality, implies that
\begin{equation}
\norm{p_{\rv{X}',\rv{Y}'_1,\ldots,\rv{Y}'_M,\rv{Z}'_1,\ldots,\rv{Z}'_N} - p_{\rv{X}'}\times \prod_{i=1}^{M} q_{\rv{Y}'_i} \times \prod_{\ell=1}^N r_{\rv{Z}'_\ell}}_\fnc{tvd} \leq \delta.
\end{equation}
Finally, to show~\eqref{eq:convex:split:lemma:anshu}, one applies triangular inequality while noticing that $p_{\rv{X}'}=p_\rv{X}$.
\end{proof}

\section{Proof of Blahut-Arimoto Results}\label{app:blahut:arimoto:details}

\subsection{Extension Function and Update Rule}

This section extends the framework from~\cite{ramakrishnan2020computing} to the multipartite setting. Recall that we are interested in solving the following optimization problem
\begin{equation}
    \max_{p_{\rv{X}}\in\set{P}(\set{X})} I(\rv{X};\rv{Y_1};\rv{Y_2};\cdots;\rv{Y_K})_{p_\rv{X}\cdot W_{\rvs{Y}_1^K|\rv{X}}}.
\end{equation}
The method is to maximize the following bi-variate function (known as the extension function) alternatively until it converges to $\tilde{C}(W)$:
\begin{align}
    J(p_{\rv{X}}, q_{\rv{X}}) &\defeq I(\rv{X};\rv{Y_1};\rv{Y_2};\cdots;\rv{Y_K})_{p_\rv{X}\cdot W_{\rvs{Y}_1^K|\rv{X}}} + \sum_{i=1}^K D\infdiv*{p_{\rv{Y}_i}}{q_{\rv{Y}_i}} - K\cdot D\infdiv*{p_\rv{X}}{q_\rv{X}}\\
    &=-K\cdot D\infdiv*{p_{\rv{X}}}{q_{\rv{X}}} + \sum_{x\in\mathcal{X}}p_{\rv{X}}(x)\cdot D\infdiv*{W_{\rvs{Y}_1^K|\rv{X}}(\cdot|x)}{q_{\rv{Y}_1}\times\cdots\times q_{\rv{Y}_K}}\\
    &= \sum_{x\in\mathcal{X}}p_{\rv{X}}(x)\left[ -K\cdot\log{p_\rv{X}(x)} + K\cdot\log {q_\rv{X}(x)} + D\infdiv*{W_{\rvs{Y}_1^K|\rv{X}}(\cdot|x)}{q_{\rv{Y}_1}\times\cdots\times q_{\rv{Y}_K}} \right] ,
\end{align}
where $p_{\rv{Y}_i}(\cdot) \defeq \sum_{x\in\mathcal{X}}W_{\rv{Y}_i|\rv{X}}(\cdot|x)\cdot p_{\rv{X}}(x)$ and $q_{\rv{Y}_i}(\cdot) \defeq \sum_{x\in\mathcal{X}}W_{\rv{Y}_i|\rv{X}}(\cdot|x)\cdot q_{\rv{X}}(x)$.
On the one hand, it is clear that
\begin{equation}\label{eq:BA:update:1}
\argmax_{q_\rv{X}\in\set{P}(\set{X})}J(p_\rv{X},q_\rv{X}) = p_\rv{X}
\end{equation}
due to the data processing inequality.
On the other hand, for each $q_{\rv{X}}\in\set{P}(\set{X})$, the funciton $J(p_\rv{X},q_\rv{X})$ is concave in $p_\rv{X}$.
Thus, by requiring the partial derivative of $J$ with respect to $p_{\rv{X}}$ to be $0$, we find the optimal $p_{\rv{X}}$ to be
\begin{equation}\label{eq:BA:update:2}
    p_\rv{X} \propto q_\rv{X}\cdot \exp{\left(\frac{D\infdiv*{W_{\rvs{Y}_1^K|\rv{X}}(\cdot|x)}{q_{\rv{Y}_1}\times\cdots\times q_{\rv{Y}_K}}}{K}\right)}.
\end{equation}
Combining~\eqref{eq:BA:update:1} and~\eqref{eq:BA:update:2}, we end up with the update rule as in~\eqref{eq:update_rule_classical_multipartite}.


\subsection{Speed of Convergence}

At $n$-th step, we estimate the $K$-partite common information as
\begin{align}
    \tilde{C}(n) \defeq J(p_{\rv{X}}^{(n)}, p_{\rv{X}}^{(n\!-\!1)})
    &= \sum_{x\in\mathcal{X}}p^{(n)}_{\rv{X}}(x)\left[ -K\!\cdot\!\log{p^{(n)}_\rv{X}(x)} + K\!\cdot\!\log {p^{(n\!-\!1)}_\rv{X}(x)} + D\infdiv*{W_{\rvs{Y}_1^K|\rv{X}}(\cdot|x)}{p^{(n\!-\!1)}_{\rv{Y}_1}\times\cdots\times p^{(n\!-\!1)}_{\rv{Y}_K}} \right] \\
    &= \sum_{x\in\mathcal{X}}p^{(n)}_{\rv{X}}(x) \cdot K\cdot \log{\left(\sum_{\tilde{x}} p^{(n-1)}_\rv{X}(\tilde{x})\cdot \exp{\left(\frac{D\infdiv*{W_{\rvs{Y}_1^K|\rv{X}}(\cdot|\tilde{x})}{p^{(n-1)}_{\rv{Y}_1}\times\cdots\times p^{(n-1)}_{\rv{Y}_K}}}{K}\right)}\right)}\\
    &= K\cdot \log{\left(\sum_{\tilde{x}} p^{(n-1)}_\rv{X}(\tilde{x})\cdot \exp{\left(\frac{D\infdiv*{W_{\rvs{Y}_1^K|\rv{X}}(\cdot|\tilde{x})}{p^{(n-1)}_{\rv{Y}_1}\times\cdots\times p^{(n-1)}_{\rv{Y}_K}}}{K}\right)}\right)},
\end{align}
where we have used the update rule~\eqref{eq:update_rule_classical_multipartite} for $p_{\rv{X}}^{(n)}$.
Let $p^{\tri}_{\rv{X}}$ be a pmf achieving $\tilde{C}(W)$.
We have
\begin{align}
&\sum_{x\in\set{X}} p^{\tri}_{\rv{X}}(x)\cdot \left( \log{p_\rv{X}^{(n)}(x)} - \log{p_\rv{X}^{(n-1)}(x)}\right) \nonumber \\
= & \sum_{x\in\set{X}} p^{\tri}_{\rv{X}}(x) \cdot \begin{aligned}[t]
&\left( \frac{D\infdiv*{W_{\rvs{Y}_1^K|\rv{X}}(\cdot|x)}{p^{(n-1)}_{\rv{Y}_1}\times\cdots\times p^{(n-1)}_{\rv{Y}_K}}}{K} \right.\\
&\hspace{107pt}\left. - \log{\left(\sum_{\tilde{x}} p^{(n-1)}_\rv{X}(\tilde{x})\cdot \exp{\left(\frac{D\infdiv*{W_{\rvs{Y}_1^K|\rv{X}}(\cdot|\tilde{x})}{p^{(n-1)}_{\rv{Y}_1}\times\cdots\times p^{(n-1)}_{\rv{Y}_K}}}{K}\right)}\right)} \right)
\end{aligned}\\
= & -\frac{1}{K}\tilde{C}(n) + \frac{1}{K}\begin{aligned}[t]\Bigg(
&\sum_{x\in\set{X}}p^{\tri}_{\rv{X}}(x)\cdot D\infdiv*{W_{\rvs{Y}_1^K|\rv{X}}(\cdot|x)}{p^{(n-1)}_{\rv{Y}_1}\times\cdots\times p^{(n-1)}_{\rv{Y}_K}} \\
&\hspace{162pt}+ \sum_{i=1}^K \sum_{y}p^{\tri}_{\rv{Y}_i}(y)\log{p^{\tri}_{\rv{Y}_i}(y)} - \sum_{i=1}^K \sum_{y}p^{\tri}_{\rv{Y}_i}(y)\log{p^{\tri}_{\rv{Y}_i}(y)}
\Bigg)\end{aligned}\\
= & \frac{1}{K}\left( \tilde{C}(W) - \tilde{C}(n) + \sum_{i=1}^K D\infdiv*{p^{\tri}_{\rv{Y}_i}}{p^{(n-1)}_{\rv{Y}_i}} \right)\\
\geq & \frac{1}{K}\left( \tilde{C}(W) - \tilde{C}(n)\right).
\end{align}
Choosing $p^{(0)}_{\rv{X}}$ to be the uniform distribution on $\set{X}$, we have
\begin{align}
    \frac{1}{K}\sum_{i=1}^n \tilde{C}(W) - \tilde{C}(n) &\leq \sum_{i=1}^n \sum_{x\in\set{X}} p^{\tri}_{\rv{X}}(x)\cdot \left( \log{p_\rv{X}^{(i)}(x)} - \log{p_\rv{X}^{(i-1)}(x)}\right) \\
    &= \sum_{x\in\set{X}} p^{\tri}_{\rv{X}}(x)\cdot \left( \log{p_\rv{X}^{(n)}(x)} - \log{p_\rv{X}^{(0)}(x)}\right) \\
    & = D\infdiv*{p^{\tri}_{\rv{X}}}{p_\rv{X}^{(0)}} - D\infdiv*{p^{\tri}_{\rv{X}}}{p_\rv{X}^{(n)}}\\
    &\leq D\infdiv*{p^{\tri}_{\rv{X}}}{p_\rv{X}^{(0)}} \\
    &\leq \log{\size{\set{X}}}.
\end{align}
Note that $\tilde{C}(n)$ is monotonically non-decreasing in $n$ and $\tilde{C}(n)\leq\tilde{C}(W)$ for all $n$.
One must have
\begin{equation}
    \abs{\tilde{C}(W) - \tilde{C}(n)} \leq \frac{K}{n}\log{\size{\set{X}}} .
\end{equation}
Finally, we need to point out that the above discussion is a worst-case analysis and holds for all channels.
In practice, the algorithm converges  much faster.


\end{document}